\documentclass[11pt]{article}

\newif\ifcomments

\commentstrue

\usepackage{fullpage}
\usepackage[utf8]{inputenc}
\usepackage{tikz}
\usetikzlibrary{shapes.geometric, arrows}
\usepackage{graphicx}
\usepackage{amsmath,amsthm,amssymb}
\usepackage{algorithm}
\usepackage{algpseudocode}
\usepackage{multirow}
\usepackage{makecell}
\usepackage{algpseudocode}

\usepackage{thm-restate,color,xspace}
\usepackage{comment}
\usepackage{thmtools}
\usepackage{xcolor}
\usepackage{nameref}
\usepackage{array}
\usepackage{svg}

\definecolor{ForestGreen}{rgb}{0.1333,0.5451,0.1333}
\definecolor{DarkRed}{rgb}{0.65,0,0}
\definecolor{Red}{rgb}{1,0,0}
\usepackage[linktocpage=true,
pagebackref=true,colorlinks,
linkcolor=DarkRed,citecolor=ForestGreen,
bookmarks,bookmarksopen,bookmarksnumbered]
{hyperref}
\usepackage{cleveref}

\usepackage{xcolor}
\usepackage{nameref}

\usepackage{cleveref}

\makeatother

\usepackage{babel}
\usepackage{hyperref}

\declaretheorem[numberwithin=section]{theorem}
\declaretheorem[numberlike=theorem]{lemma}

\declaretheorem[numberlike=theorem]{corollary}

\declaretheorem[numberlike=theorem]{conjecture}

\declaretheorem[numberlike=theorem]{claim}

\declaretheorem[numberlike=theorem,style=definition]{definition}
\declaretheorem[numberlike=theorem,style=definition]{remark}

\global\long\def\poly{\mathrm{poly}}

\newcommand{\flow}{\text{\normalfont flow}}

\newcommand{\fmc}
{\text{\normalfont FMC}}

\newcommand{\conn}{\lambda}

\newcommand{\is}{\text{\normalfont ImpCut}}

\newcommand{\ssc}{\text{\normalfont SSCP}}

\newcommand{\ignore}[1]{}

\ifcomments
\def\thatchaphol#1{\marginpar{$\leftarrow$\fbox{TS}}\footnote{$\Rightarrow$~{\sf\textcolor{purple}{#1 --Thatchaphol}}}}
\def\benyu#1{\marginpar{$\leftarrow$\fbox{BW}}\footnote{$\Rightarrow$~{\sf\textcolor{red}{#1 --Benyu}}}}
\def\gary#1{\marginpar{$\leftarrow$\fbox{GH}}\footnote{$\Rightarrow$~{\sf\textcolor{blue}{#1 --Gary}}}}

\else
\newcommand{\benyu}[1]{}
\newcommand{\gary}[1]{}
\newcommand{\thatchaphol}[1]{}
\fi %ifcomments

\usepackage[utf8]{inputenc} % if you compile with pdfLaTeX
\DeclareUnicodeCharacter{1D405}{$\mathbf{F}$} % or {\textbf{F}} if you want text-bold

\DeclareUnicodeCharacter{1D400}{\ensuremath{\mathbf{A}}}
\DeclareUnicodeCharacter{1D401}{\ensuremath{\mathbf{B}}}
\DeclareUnicodeCharacter{1D402}{\ensuremath{\mathbf{C}}}
\DeclareUnicodeCharacter{1D403}{\ensuremath{\mathbf{D}}}
\DeclareUnicodeCharacter{1D404}{\ensuremath{\mathbf{E}}}
\DeclareUnicodeCharacter{1D405}{\ensuremath{\mathbf{F}}}
\DeclareUnicodeCharacter{1D406}{\ensuremath{\mathbf{G}}}
\DeclareUnicodeCharacter{1D407}{\ensuremath{\mathbf{H}}}
\DeclareUnicodeCharacter{1D408}{\ensuremath{\mathbf{I}}}
\DeclareUnicodeCharacter{1D409}{\ensuremath{\mathbf{J}}}
\DeclareUnicodeCharacter{1D40A}{\ensuremath{\mathbf{K}}}
\DeclareUnicodeCharacter{1D40B}{\ensuremath{\mathbf{L}}}
\DeclareUnicodeCharacter{1D40C}{\ensuremath{\mathbf{M}}}
\DeclareUnicodeCharacter{1D40D}{\ensuremath{\mathbf{N}}}
\DeclareUnicodeCharacter{1D40E}{\ensuremath{\mathbf{O}}}
\DeclareUnicodeCharacter{1D40F}{\ensuremath{\mathbf{P}}}
\DeclareUnicodeCharacter{1D410}{\ensuremath{\mathbf{Q}}}
\DeclareUnicodeCharacter{1D411}{\ensuremath{\mathbf{R}}}
\DeclareUnicodeCharacter{1D412}{\ensuremath{\mathbf{S}}}
\DeclareUnicodeCharacter{1D413}{\ensuremath{\mathbf{T}}}
\DeclareUnicodeCharacter{1D414}{\ensuremath{\mathbf{U}}}
\DeclareUnicodeCharacter{1D415}{\ensuremath{\mathbf{V}}}
\DeclareUnicodeCharacter{1D416}{\ensuremath{\mathbf{W}}}
\DeclareUnicodeCharacter{1D417}{\ensuremath{\mathbf{X}}}
\DeclareUnicodeCharacter{1D418}{\ensuremath{\mathbf{Y}}}
\DeclareUnicodeCharacter{1D419}{\ensuremath{\mathbf{Z}}}

\DeclareUnicodeCharacter{1D41A}{\ensuremath{\mathbf{a}}}
\DeclareUnicodeCharacter{1D41B}{\ensuremath{\mathbf{b}}}
\DeclareUnicodeCharacter{1D41C}{\ensuremath{\mathbf{c}}}
\DeclareUnicodeCharacter{1D41D}{\ensuremath{\mathbf{d}}}
\DeclareUnicodeCharacter{1D41E}{\ensuremath{\mathbf{e}}}
\DeclareUnicodeCharacter{1D41F}{\ensuremath{\mathbf{f}}}
\DeclareUnicodeCharacter{1D420}{\ensuremath{\mathbf{g}}}
\DeclareUnicodeCharacter{1D421}{\ensuremath{\mathbf{h}}}
\DeclareUnicodeCharacter{1D422}{\ensuremath{\mathbf{i}}}
\DeclareUnicodeCharacter{1D423}{\ensuremath{\mathbf{j}}}
\DeclareUnicodeCharacter{1D424}{\ensuremath{\mathbf{k}}}
\DeclareUnicodeCharacter{1D425}{\ensuremath{\mathbf{l}}}
\DeclareUnicodeCharacter{1D426}{\ensuremath{\mathbf{m}}}
\DeclareUnicodeCharacter{1D427}{\ensuremath{\mathbf{n}}}
\DeclareUnicodeCharacter{1D428}{\ensuremath{\mathbf{o}}}
\DeclareUnicodeCharacter{1D429}{\ensuremath{\mathbf{p}}}
\DeclareUnicodeCharacter{1D42A}{\ensuremath{\mathbf{q}}}
\DeclareUnicodeCharacter{1D42B}{\ensuremath{\mathbf{r}}}
\DeclareUnicodeCharacter{1D42C}{\ensuremath{\mathbf{s}}}
\DeclareUnicodeCharacter{1D42D}{\ensuremath{\mathbf{t}}}
\DeclareUnicodeCharacter{1D42E}{\ensuremath{\mathbf{u}}}
\DeclareUnicodeCharacter{1D42F}{\ensuremath{\mathbf{v}}}
\DeclareUnicodeCharacter{1D430}{\ensuremath{\mathbf{w}}}
\DeclareUnicodeCharacter{1D431}{\ensuremath{\mathbf{x}}}
\DeclareUnicodeCharacter{1D432}{\ensuremath{\mathbf{y}}}
\DeclareUnicodeCharacter{1D433}{\ensuremath{\mathbf{z}}}

\DeclareUnicodeCharacter{1D434}{\ensuremath{\mathit{A}}}
\DeclareUnicodeCharacter{1D435}{\ensuremath{\mathit{B}}}
\DeclareUnicodeCharacter{1D436}{\ensuremath{\mathit{C}}}
\DeclareUnicodeCharacter{1D437}{\ensuremath{\mathit{D}}}
\DeclareUnicodeCharacter{1D438}{\ensuremath{\mathit{E}}}
\DeclareUnicodeCharacter{1D439}{\ensuremath{\mathit{F}}}
\DeclareUnicodeCharacter{1D43A}{\ensuremath{\mathit{G}}}
\DeclareUnicodeCharacter{1D43B}{\ensuremath{\mathit{H}}}
\DeclareUnicodeCharacter{1D43C}{\ensuremath{\mathit{I}}}
\DeclareUnicodeCharacter{1D43D}{\ensuremath{\mathit{J}}}
\DeclareUnicodeCharacter{1D43E}{\ensuremath{\mathit{K}}}
\DeclareUnicodeCharacter{1D43F}{\ensuremath{\mathit{L}}}
\DeclareUnicodeCharacter{1D440}{\ensuremath{\mathit{M}}}
\DeclareUnicodeCharacter{1D441}{\ensuremath{\mathit{N}}}
\DeclareUnicodeCharacter{1D442}{\ensuremath{\mathit{O}}}
\DeclareUnicodeCharacter{1D443}{\ensuremath{\mathit{P}}}
\DeclareUnicodeCharacter{1D444}{\ensuremath{\mathit{Q}}}
\DeclareUnicodeCharacter{1D445}{\ensuremath{\mathit{R}}}
\DeclareUnicodeCharacter{1D446}{\ensuremath{\mathit{S}}}
\DeclareUnicodeCharacter{1D447}{\ensuremath{\mathit{T}}}
\DeclareUnicodeCharacter{1D448}{\ensuremath{\mathit{U}}}
\DeclareUnicodeCharacter{1D449}{\ensuremath{\mathit{V}}}
\DeclareUnicodeCharacter{1D44A}{\ensuremath{\mathit{W}}}
\DeclareUnicodeCharacter{1D44B}{\ensuremath{\mathit{X}}}
\DeclareUnicodeCharacter{1D44C}{\ensuremath{\mathit{Y}}}
\DeclareUnicodeCharacter{1D44D}{\ensuremath{\mathit{Z}}}

\DeclareUnicodeCharacter{1D44E}{\ensuremath{\mathit{a}}}
\DeclareUnicodeCharacter{1D44F}{\ensuremath{\mathit{b}}}
\DeclareUnicodeCharacter{1D450}{\ensuremath{\mathit{c}}}
\DeclareUnicodeCharacter{1D451}{\ensuremath{\mathit{d}}}
\DeclareUnicodeCharacter{1D452}{\ensuremath{\mathit{e}}}
\DeclareUnicodeCharacter{1D453}{\ensuremath{\mathit{f}}}
\DeclareUnicodeCharacter{1D454}{\ensuremath{\mathit{g}}}
\DeclareUnicodeCharacter{210E}{\ensuremath{\mathit{h}}}
\DeclareUnicodeCharacter{1D456}{\ensuremath{\mathit{i}}}
\DeclareUnicodeCharacter{1D457}{\ensuremath{\mathit{j}}}
\DeclareUnicodeCharacter{1D458}{\ensuremath{\mathit{k}}}
\DeclareUnicodeCharacter{1D459}{\ensuremath{\mathit{l}}}
\DeclareUnicodeCharacter{1D45A}{\ensuremath{\mathit{m}}}
\DeclareUnicodeCharacter{1D45B}{\ensuremath{\mathit{n}}}
\DeclareUnicodeCharacter{1D45C}{\ensuremath{\mathit{o}}}
\DeclareUnicodeCharacter{1D45D}{\ensuremath{\mathit{p}}}
\DeclareUnicodeCharacter{1D45E}{\ensuremath{\mathit{q}}}
\DeclareUnicodeCharacter{1D45F}{\ensuremath{\mathit{r}}}
\DeclareUnicodeCharacter{1D460}{\ensuremath{\mathit{s}}}
\DeclareUnicodeCharacter{1D461}{\ensuremath{\mathit{t}}}
\DeclareUnicodeCharacter{1D462}{\ensuremath{\mathit{u}}}
\DeclareUnicodeCharacter{1D463}{\ensuremath{\mathit{v}}}
\DeclareUnicodeCharacter{1D464}{\ensuremath{\mathit{w}}}
\DeclareUnicodeCharacter{1D465}{\ensuremath{\mathit{x}}}
\DeclareUnicodeCharacter{1D466}{\ensuremath{\mathit{y}}}
\DeclareUnicodeCharacter{1D467}{\ensuremath{\mathit{z}}}

\DeclareUnicodeCharacter{1D7CE}{\ensuremath{\mathbf{0}}}
\DeclareUnicodeCharacter{1D7CF}{\ensuremath{\mathbf{1}}}
\DeclareUnicodeCharacter{1D7D0}{\ensuremath{\mathbf{2}}}
\DeclareUnicodeCharacter{1D7D1}{\ensuremath{\mathbf{3}}}
\DeclareUnicodeCharacter{1D7D2}{\ensuremath{\mathbf{4}}}
\DeclareUnicodeCharacter{1D7D3}{\ensuremath{\mathbf{5}}}
\DeclareUnicodeCharacter{1D7D4}{\ensuremath{\mathbf{6}}}
\DeclareUnicodeCharacter{1D7D5}{\ensuremath{\mathbf{7}}}
\DeclareUnicodeCharacter{1D7D6}{\ensuremath{\mathbf{8}}}
\DeclareUnicodeCharacter{1D7D7}{\ensuremath{\mathbf{9}}}
\DeclareUnicodeCharacter{2265}{\ensuremath{\geq}}     % ≥
\DeclareUnicodeCharacter{2212}{\ensuremath{-}}        % − minus sign
\DeclareUnicodeCharacter{2032}{\ensuremath{^\prime}}  % ′ prime
\DeclareUnicodeCharacter{2113}{\ensuremath{\ell}}     % ℓ
\DeclareUnicodeCharacter{22EF}{\ensuremath{\cdots}}   % ⋯
\DeclareUnicodeCharacter{2033}{\ensuremath{^{\prime\prime}}}   % ″
\DeclareUnicodeCharacter{2034}{\ensuremath{^{\prime\prime\prime}}} % ‴

\begin{document}

% \title{Fault-Tolerant Symmetric Connectivity Preservers}
% \title{Fault-Tolerant Connectivity Preservers in Directed Graphs}
% \title{Near-Optimal Fault-Tolerant Connectivity Preservers in Directed Graphs}
\title{Near-Optimal Fault-Tolerant Strong Connectivity Preservers}

\author{
  Gary Hoppenworth\thanks{\texttt{garytho@umich.edu}. Supported in part by NSF grant NSF:AF 2153680.}
  \and Thatchaphol Saranurak\thanks{\texttt{thsa@umich.edu}. Supported by NSF Grant CCF-2238138.}
  \and Benyu Wang\thanks{\texttt{benyuw@umich.edu}.}
}

\date{}
\maketitle

\begin{abstract}

A \emph{$k$-fault-tolerant connectivity preserver} of a directed $n$-vertex graph $G$ is a subgraph $H$ such that, for any edge set $F \subseteq E(G)$ of size $|F| \le k$, the strongly connected components of $G - F$ and $H - F$ are the same. While some graphs require a preserver with $\Omega(2^{k}n)$ edges \cite{baswana2018fault}, the best-known upper bound is $\tilde{O}(k2^{k}n^{2-1/k})$ edges~\cite{scc}, leaving a significant gap of $\Omega(n^{1-1/k})$. In contrast, there is no gap in \emph{undirected} graphs; the optimal bound of $\Theta(kn)$ has been well-established since the 90s \cite{nagamochi1992linear}.

We nearly close the gap for directed graphs; we prove that there exists a $k$-fault-tolerant connectivity preserver with $O(k4^{k}n\log n)$ edges, and we can construct one with $O(8^{k}n\log^{5/2}n)$ edges in $\poly(2^{k}n)$ time.

Our results also improve the state-of-the-art for a closely related object; a \emph{$k$-connectivity preserver} of $G$ is a subgraph $H$ where, for all $i \le k$, the strongly $i$-connected components of $G$ and $H$ agree. By a known reduction, we obtain a $k$-connectivity preserver with $O(k4^{k}n\log n)$ edges, improving the previous best bound of $\tilde{O}(k2^{k}n^{2-1/(k-1)})$ \cite{scc}. Therefore, for any constant $k$, our results are optimal to a $\log n$ factor for both problems.

Lastly, we show that the exponential dependency on $k$ is not inherent for $k$-connectivity preservers by presenting another construction with $O(n \sqrt{kn})$ edges.
% \ghnoteinline{One high-level comment about the abstract: we do not use the word ``symmetric'' at all, which may make the ``symmetric'' in the title confusing for people. Maybe we could replace ``symmetric'' with ``strong'' or we could remove ``symmetric'' from the title.}
% \bwnoteinline{Agree.}
\end{abstract}

% \begin{comment}

% Since every $k$-fault-tolerant connectivity preserver is also a $(k+1)$-connectivity preserver

% \end{comment}

%\input{pairwise}

% \input{1-layer}
%\input{n-10-7}
%\input{multi-layer}
\thispagestyle{empty}
\newpage

\tableofcontents
\thispagestyle{empty}
\newpage

\pagenumbering{gobble}
\pagebreak
 \pagenumbering{arabic}

\section{Introduction}

A $k$-fault-tolerant ($k$-FT) connectivity preserver of a \emph{directed} graph $G$ is a subgraph $H$ that maintains the strongly connected components of $G$ even after $k$ edge failures. Specifically, for any edge set $F$ with size $|F| \le k$, the strongly connected components of $G-F$ and $H-F$ are the same. 
In undirected graphs, a $k$-FT connectivity preserver simplifies to a subgraph that preserves connected components even after edge failures. 
We study the following question:

\begin{center}

\emph{Does every graph admit a sparse $k$-FT connectivity preserver?}

\par\end{center}

\paragraph{Undirected Graphs: Settled.}
This question is well understood in undirected graphs. Since the 90s, Nagamochi and Ibaraki \cite{nagamochi1992linear} developed a linear-time algorithm to compute a $k$-FT connectivity preserver with $O(nk)$ edges, which is optimal.\footnote{Strictly speaking, \cite{nagamochi1992linear} proved that their algorithm outputs a \emph{$k$-connectivity preserver}, which is slightly weaker and will be defined soon in \Cref{subsec:intro kconn}. However, it is well known in the community that their output is stronger. Benczúr and Karger \cite{benczur1996approximating} explicitly proved that it preserves all cuts of size at most $k$, implying it is a $k$-FT connectivity preserver. See also \cite{choudhary2021kft}.}

Since then, sparse $k$-FT connectivity preservers have been highly influential. They were first used to accelerate algorithms for minimum cuts and cut sparsifiers \cite{benczur1996approximating,benczur2015randomized}. Subsequently, they have been studied and applied in various computational models, including dynamic algorithms \cite{dynamic1,thorup2007fully,abraham2016fully}, parallel and distributed algorithms \cite{thurimella1995sub,daga2019distributed,parter2019small}, and streaming algorithms \cite{ahn2012graph,guha2015vertex,assadi2023tight}. Another line of work strengthens the definition to address stronger types of failures, such as both edge and vertex failures \cite{cheriyan1993scan,frank1993sparse,even1995mixed,fulop2005sparse,berg2005sparse,nagamochi2006sparse}, bounded-degree edge failures \cite{boundeddegree1}, and color failures \cite{color0,color2}. Even in these generalized settings, near-optimal bounds of $\tilde{O}(nk)$ have been established.\footnote{$\tilde{O}(\cdot)$ hides logarithmic factors}

In the past decade, extensive research has also focused on other fault-tolerant subgraphs beyond connectivity, such as approximate distances \cite{spanner1,spanner2,vertexfault1,spanner4,vertexfault2,vertexfault4,boundeddegree0,color0,boundeddegree1,color2,spanner3}, exact distances \cite{exact_undirected_2,exact_undirected_1,exact_directed_1}, diameter \cite{diameter}, and spectral properties \cite{spectral}. Unfortunately, despite success in undirected graphs, progress in directed graphs has been very limited.

\paragraph{State-of-the-Art in Directed Graphs.}

Directed graphs are fundamentally harder to compress. Even without failures, a standard information-theoretic argument shows that directed graphs cannot be compressed into any $o(n^{2})$-sized data structure that answers reachability queries, i.e., whether vertex $a$ has a path to $b$. This presents a significant barrier for any attempt to compress \emph{asymmetric} connectivity information.

This barrier motivates the study of $k$-FT connectivity preservers because a strongly connected component is a symmetric notion. In a breakthrough result \cite{baswana2018fault}, Baswana, Choudhary, and Roditty completely solved the special case of $k$-FT \emph{single-source} connectivity preservers, which preserve the strongly connected component containing a \emph{single} vertex $s$ instead of all strongly connected components.\footnote{This is equivalent, up to a constant factor, to the \emph{$k$-FT single-source reachability preserver} as stated in \cite{baswana2018fault}.} Their preserver contains $O(2^{k}n)$ edges, and this bound is tight even in this special case.

However, for the general case, a significant gap remains. The only non-trivial $k$-FT connectivity preservers were shown by Chakraborty and Choudhary \cite{scc} and require $\tilde{O}(k2^{k}n^{2-1/k})$ edges. In particular, even their $2$-FT connectivity preservers require $\Omega(n\sqrt{n})$ edges. They posed an open question whether the large $\tilde{\Omega}(n^{1-1/k})$-factor gap between upper and lower bounds can be closed.

\subsection{Main Result: Near-Optimal $k$-Fault-Tolerant Connectivity Preservers}

\label{subsec:mainresult}

We present $k$-FT connectivity preservers with $O(n\log n)$ edges for every constant $k$, reducing the gap from $\tilde{\Omega}(n^{1-1/k})$ to $O(\log n)$.

\begin{restatable}{theorem}{existential}

\label{thm:FT} For every positive integer $k$, every $n$-vertex directed graph admits a $k$-fault-tolerant connectivity preserver with $O(k4^{k}n\log n)$ edges. Moreover, we can construct one with $O(8^{k}n\log^{5/2}n)$ edges in $\poly(2^{k}n)$ time.

\end{restatable}

For constant $k$, the $k$-FT connectivity preservers of $O(n)$ size prior to \Cref{thm:FT} are known only for undirected graphs \cite{nagamochi1992linear} or for the single-source version in directed graphs \cite{baswana2018fault}. \Cref{thm:FT} generalizes both results to directed graphs with only an $O(\log n)$-factor increase in the preserver size. Quantitatively, \Cref{thm:FT} also significantly improves the bound of $\tilde{O}(n^{2-1/k})$ \cite{scc} to $O(n\log n)$.

\Cref{thm:FT} directly implies a data structure of size $O(k4^{k}n\log n)$ that, given a query set $F$ of size $|F|\le k$, reports all strongly connected components of $G-F$ in $O(k4^{k}n\log n)$ time. This already significantly improves upon the space of the previous data structure in \cite{scc_data_structure} which required $O(2^{k}n^{2})$ space and $O(2^{k}n\log^{2}n)$ query time. It also generalizes the data structure of \cite{georgiadis2020strong} with optimal $O(n)$ space and $O(n)$ query time, but that works only for $k=1$.

\paragraph{Construction Time.}

A straightforward construction of our preservers would take $n^{O(k)}$ time. In the second statement of \Cref{thm:FT}, we show an FPT-time construction, at the cost of a slightly larger size bound. It is very interesting if near-linear time construction is possible. However, this remains open even for the single-source version; the state-of-the-art algorithm \cite{baswana2018fault} still requires $\Omega(mn)$ construction time for constant $k$.

\begin{comment}

\paragraph{Technical Highlights.}

Interestingly, \Cref{thm:FT} is obtained via a new application of directed expander hierarchy recently introduced in \cite{combmaxflow} to reduce the problem from general directed graphs to the problem in a graph with good expansion (see \Cref{sec:techoverview} for a detailed overview). To obtain an FPT-time construction, as a byproduct, we observe a new structure about \emph{important separators}, one of the central objects in parameterized graph algorithms CITE. The lemma is stated in REF and might be of independent interest.

\end{comment}

\paragraph{Strong Failure Models.}

\Cref{thm:FT} can be extended to address \emph{vertex failures} using the standard vertex-splitting reduction.

Recently, two stronger types of failures, namely \emph{bounded-degree edge failures} \cite{boundeddegree0} and \emph{color failures} \cite{color0}, were introduced. In \cite{color0,color2,boundeddegree0,boundeddegree1},  fault-tolerant connectivity preservers (and even fault-tolerant spanners) in undirected graphs have been strengthened to handle these types of failures.

We show that this generalization is impossible for directed graphs. In particular, we show directed graphs undergoing just $1$-bounded-degree edge failures or $1$-color failures require a fault-tolerant connectivity preserver with $\tilde{\Omega}(n^{2})$ edges. This marks a fundamental distinction between undirected and directed graphs. We formally define these fault models and prove these lower bounds in \Cref{app:other fault models}.

\subsection{Side Result: $k$-Connectivity Preservers}

\label{subsec:intro kconn}

\Cref{thm:FT} also immediately advances the state-of-the-art of $k$-connectivity preservers. At a high level, $k$-connectivity preservers are  subgraphs that preserve higher connectivity but are not fault-tolerant. Before providing a formal definition, let us recall some basic concepts. Let $\flow(G,s,t)$ denote the maximum number of edge-disjoint paths from $s$ to $t$. We say that $s$ and $t$ are (strongly) $i$-connected if $\flow(G,s,t)\ge i$ and $\flow(G,t,s)\ge i$. For any $i$, the $i$-connected components uniquely partition the vertex set into parts such that two vertices are in the same part iff they are $i$-connected. The graph $G$ is $i$-connected if all pairs of vertices are $i$-connected.

A \emph{$k$-connectivity preserver} of a directed graph $G$ is a subgraph $H$ where, for all $i\le k$, the $i$-connected components of $G$ and $H$ are the same. Note that every $(k-1)$-FT connectivity preserver of $G$ must be a $k$-connectivity preserver, but the reverse is false; there exists a directed graph $G$ that admits a $k$-connectivity preserver with $O(n)$ edges, but requires $\Omega(n2^{k})$ edges to preserve strongly connected components under $k$ edge faults. See \Cref{thm:conn ft gap} for the proof.

\paragraph{State-of-the-Art in Directed Graphs.}

In undirected graphs, the $O(nk)$ optimal bounds for $k$-connectivity preservers have long been established \cite{nagamochi1992linear}, but they are much less understood in directed graphs. Currently, tight bounds are known only in special cases. Let $\lambda(G)$ denote the minimum cut size of $G$. Since the '70s, \cite{dalmazzo1977nombre,mader1985minimal} (cf.~\cite{berg2005minimally}) showed that every directed graph $G$ contains a $k$-connected subgraph $H$ with $O(nk)$ edges for all $k\le\lambda(G)$, and this is tight. This implies optimal $k$-connectivity preservers of size $\Theta(nk)$ when $k\le\lambda(G)$. However, this argument fails when $k > \lambda(G)$. Much later, in 2016, \cite{connectivity_certificate} showed a $2$-connectivity preserver of size $O(n)$. Very recently, \cite{bansal2024faulttolerantboundedflowpreservers} implies a \emph{single-source} version of $k$-connectivity preservers that only preserves, for all $i\le k$, the strongly $i$-connected component containing the fixed vertex $s$. Their preserver has an optimal size of $\Theta(nk)$.

In the general case, the best-known $k$-connectivity preservers prior to our work follow from the best-known construction of $k$-FT connectivity preservers by \cite{scc}, which implies the bound of $\tilde{O}(k2^{k}n^{2-1/(k-1)})$ edges. \Cref{thm:FT} directly provides a significant improvement upon this bound.

\begin{corollary}
\label{thm:kconn}For every positive integer $k$, every $n$-vertex directed graph admits a $k$-connectivity preserver with $O(k4^{k}n\log n)$ edges.

\end{corollary}

\Cref{thm:kconn} is optimal up to an $O(\log n)$ factor when $k$ is constant. But, are there non-trivial $k$-connectivity preservers even when $k$ is large? We give a strong affirmative answer by showing the first $o(n^{2})$-sized $k$-connectivity preserver for any $k=o(n)$.

\begin{restatable}{theorem}{lambdaconn}
    \label{thm:lambda_conn}
    \label{thm:polysize}For every positive integer $k$, every $n$-vertex directed graph admits a $k$-connectivity preserver with $O(n\sqrt{nk})$ edges. The preserver can be constructed in polynomial time.
\end{restatable}

\Cref{thm:polysize} establishes a strong separation between $k$-FT connectivity preservers and $k$-connectivity preservers in directed graphs. The former suffers from an exponential $\Omega(2^{k})$ factor (due to \cite{baswana2018fault}) and cannot admit an $o(n^{2})$ upper bound when $k \ge \log_{2}n$. In contrast, the latter can achieve the $o(n^{2})$ upper bound for any $k=o(n)$ by \Cref{thm:lambda_conn}. Notably, in undirected graphs, the optimal bounds of the two notions coincide at $\Theta(kn)$ \cite{nagamochi1992linear}.

\subsection{Mapping the Landscape of Connectivity Preservers in Directed Graphs}

The strong separation between $k$-connectivity and $k$-FT connectivity established by \Cref{thm:polysize} suggests a rich landscape of connectivity preserver problems in directed graphs. \Cref{tab:bounds} lists the best-known upper and lower bounds for all variants of $k$-FT connectivity preservers and $k$-connectivity preservers in the literature.

\begin{table}[ht]
\centering
\small{
\begin{tabular}{|c|cc|cc|} 
\hline
 & \multicolumn{2}{c|}{$k$-fault connectivity }    & \multicolumn{2}{c|}{$k$-connectivity }    \\ \hline
& \multicolumn{1}{c|}{Upper Bound} & Lower Bound & \multicolumn{1}{c|}{Upper Bound} & Lower Bound 
\\ \hline

$s$-$t$ & \multicolumn{1}{c|}{$O(n2^{k})$ [Thm \ref{thm:reductions}] } & $\Omega(n2^{k/2})$ [Thm \ref{thm:stlower}] & \multicolumn{1}{c|}{$O(n\sqrt{k})$  \cite{karger1998finding}} &  $\Omega(n\sqrt{k})$ \cite{karger1998finding}
\\ \hline
global & \multicolumn{1}{c|}{$O(n2^{k})$ [Thm \ref{thm:reductions}]} & $\Omega(n2^{k/2})$ [Thm \ref{thm:reductions}] & \multicolumn{1}{c|}{$O(nk)$ \cite{dalmazzo1977nombre,berg2005minimally}} & $\Omega(nk)$ \cite{dalmazzo1977nombre,berg2005minimally}
\\ \hline
single-source  & \multicolumn{1}{c|}{$O(n2^{k})$ \cite{baswana2018fault}} &  $\Omega(n2^{k})$ \cite{baswana2018fault} & \multicolumn{1}{c|}{$O(nk)$ \cite{bansal2024faulttolerantboundedflowpreservers} } & $\Omega(nk)$ [Thm \ref{thm:reductions}]
\\ \hline
all-pairs  & \multicolumn{1}{c|}{$\Tilde{O}(nk4^k)$ [Thm \ref{thm:FT}]} & $\Omega(n2^{k})$ \cite{scc} & \multicolumn{1}{c|}{\makecell{$O(n\sqrt{nk})$   [Thm \ref{thm:polysize}] \\ $\Tilde{O}(nk4^k)$ [Thm \ref{thm:kconn}]}} & $\Omega(nk)$ [Thm \ref{thm:reductions}] \\ \hline
\end{tabular}
}
\caption{Upper and lower bounds for connectivity preserver variants in directed graphs.}
\label{tab:bounds}
\end{table}
% \benyu{Should I also add 1.2 here?}\thatchaphol{yes please in the table.}

These variants include the $s$-$t$, global, single-source, and all-pairs versions. See \Cref{sec:reduc} for the formal definitions and reductions between them. For each variant,  $(k-1)$-FT connectivity preservers are also $k$-connectivity preservers. It is clear from the definitions that the $s$-$t$ version is a special case of the single-source version, which in turn is a special case of the all-pairs version. Interestingly, we can show reductions that place the global version between the $s$-$t$ and single-source versions, establishing a linear order of difficulty among all versions.

Our results (\Cref{thm:FT,thm:kconn,thm:polysize}) for the all-pairs version, along with the reduction in \Cref{thm:reductions}, and a lower bound for $k$-fault-tolerant $s$-$t$ connectivity preservers in \Cref{thm:stlower}, collectively provide a comprehensive understanding of the landscape of connectivity preservers in directed graphs. Notably, our reduction implies that the global version of $k$-fault-tolerant connectivity must exhibit an exponential dependency on $k$ in its size, which was unclear prior to our work.

Given our results, the bounds for all variants of $k$-FT connectivity preservers are nearly tight up to an $O(\log n)$ factor and a constant in the exponent of $2^{\Theta(k)}$.
The bounds for all variants of $k$-connectivity preservers are tight except for all-pairs $k$-connectivity preservers, where our upper bounds of $O(n\sqrt{nk})$ and $O(k4^k n\log n)$ do not yet match the $\Omega(nk)$ lower bound.

% , the upper and lower bounds are nearly tight up to an $O(\log n)$ factor and a constant in the exponent of $2^{\Theta(k)}$\benyu{Should we mention both bounds we proved?}. We leave tightening these bounds as an open problem.

\section{Technical Overview}
\label{sec:techoverview}
% Technical overview: make it self-contained, but not too formal

In this section, we give an overview of the techniques and proof ideas in our work. We begin by describing a simple greedy strategy for constructing connectivity preservers that is standard in the area of graph spanners and preservers.

\paragraph{Greedy connectivity preservers.} Given any graph $G$, we can construct a connectivity preserver $H$ of $G$ using the following greedy algorithm. Initially, let $H \gets G$. While there exists an edge $e \in E(H)$ such that $H-e$ is a valid connectivity preserver of $H$, we let $H \gets H-e$. This algorithm terminates and outputs a subgraph $H \subseteq G$ when no such edge $e \in E(H)$ exists. By a transitivity argument, $H$ is a valid connectivity preserver of $G$.   
This algorithm easily extends to $k$-FT connectivity preservers and $k$-connectivity preservers as well.

This greedy algorithm is useful for analyzing $k$-FT connectivity preservers (respectively, $k$-connectivity preservers) because it allows us to assume wlog that the input graph $G$ is an edge-minimal $k$-FT connectivity preserver (respectively, $k$-connectivity preserver) of itself. For simplicity, throughout this section we will assume all  graphs we consider have this edge minimality property.

% does not have a proper  subgraph $G' \subsetneq G$ that is a valid $k$-FT connectivity preserver (respectively, $k$-connectivity preserver) of $G$. For simplicity, when discussing grap

\subsection{Fault-Tolerant Connectivity Preservers}

We now discuss some of the proof ideas for  our construction of fault-tolerant connectivity preservers.

\paragraph{Unbreakability.} A key technical tool used in our fault-tolerant connectivity preserver construction is a directed expander hierarchy (formally defined in \Cref{lem:expander_hierarchy}). At a high level, using a directed expander hierarchy allows us to assume the graph has certain good expansion properties. The specific graph expansion property we use is called  \textit{unbreakability} \cite{cygan2014minimum,cygan2020randomized}. %\thatchaphol{cite papers https://dl.acm.org/doi/abs/10.1145/2591796.2591852 and https://arxiv.org/pdf/1810.06864}  

\begin{restatable}[unbreakable]{definition}{unbreakable}
    \label{def:unbreakable}
    Let $G$ be a directed graph, and let $U \subseteq V$. Let $q \geq k$ be positive integers. We say that $U$ is $(q, k)$-unbreakable in $G$ if for every cut $(S, V-S)$ in $G$ such that $|\delta^+_G(S)| \leq k$ or $|\delta^-_G(S)| \leq k$, we have either $|S \cap U| \leq q$ or $|U-S| \leq q$. 
\end{restatable}

Here $\delta_G^+(S)$ denotes the outgoing edges from $S$, and $\delta_G^-(S)$ denotes the incoming edges to $S$. We say that a graph $G$ is $(q, k)$-unbreakable if its vertex set $V(G)$ is $(q, k)$-unbreakable in $G$. To give some intuition for why unbreakability might be a useful concept here, we quickly show that unbreakable graphs have sparse $k$-FT connectivity preservers.

\paragraph{Sparse $k$-FT connectivity preserver for unbreakable graphs (\#1).} 
Let $G$ be a $(q, k+1)$-unbreakable directed graph that is an edge-minimal $k$-FT connectivity preserver of itself.  We will show that $G$ has at most $|E(G)| = O((q+k)n)$ edges. In the existential proof of
our $k$-FT connectivity preservers in \Cref{sec:conn_pres}, we let $q = O(k)$, so $|E(G)|$ will be appropriately sparse. 

Since graph $G$ is an edge-minimal $k$-FT connectivity preserver of itself, subgraph $G-e$ is not a $k$-FT connectivity preserver of $G$ for any $e \in E(G)$. Then for each edge $e = (u, v) \in E(G)$,  there exists vertices $s, t \in V(G)$ and a fault set $F \subseteq E(G)$ of size at most $|F| \leq k$ such that $s$ and $t$ are strongly connected in $G-F$ but not strongly connected in $(G - e) - F$. In particular, this implies that  there exists a $(u, v)$-cut $(S, V-S)$ in $G$ of size at most $|\delta^+(S)| \leq k+1$.

Now we observe that, for any edge $(u,w) \in E(G)$, if $(u, w) \notin \delta^+(S)$, then $w \in S$. Therefore, if $|\delta^+(u)| > q+k+1$, then there will be more than $q$ out-neighbors of $u$ in $S$ since $|\delta^+(S)| \leq k+1$. Similarly, if $|\delta^-(v)| > q+k+1$, then there will be more than $q$ vertices in $V-S$. By unbreakability, we know either $|\delta^+(u)| \leq q+k+1$ or $|\delta^-(v)|\leq q+k+1$. Therefore, we can charge edge $e=(u, v)$ to $u$ if $|\delta^+(u)| \leq q+k+1$ or to $v$ if $|\delta^-(v)| \leq q+k+1$. We observe that under this scheme, each vertex in $V(G)$ is charged $O(q+k)$ times and thus the total number of edges in $G$ is $|E(G)| = O((q+k)n)$. %\thatchaphol{Think later: whether everything works in multigraphs?}
\vspace{1mm}

This simple argument shows that unbreakable graphs have sparse $k$-FT connectivity preservers. Frequently in graph algorithms and data structures, if we can solve a problem on graphs with good expansion,  then we can extend our solution to general graphs by decomposing the graph into subgraphs with good expansion (e.g., using expander decomposition). In our case, we use a specific decomposition which we  call the \textit{directed expander hierarchy} (formally defined in \Cref{lem:expander_hierarchy}). This directed expander hierarchy is a directed version of the (undirected) expander hierarchy proved in \cite{long2024connectivity}, and it follows from a similar proof. However, to successfully apply the directed expander hierarchy, we will actually need a different argument for the sparsity of $k$-FT connectivity preservers of  $(q, k)$-unbreakable graphs that is more adaptable to our purposes.

\paragraph{Sparse $k$-FT connectivity preserver for unbreakable graphs (\#2).} Let $G$ be a directed $(q, k)$-unbreakable  graph that is an edge-minimal $k$-FT connectivity preserver of itself.  We will give an informal argument that $G$ has at most $|E(G)| = O(2^kqn)$ edges. 
We need the following property of $(q, k)$-unbreakable sets, which we prove in \Cref{sec:omit_giant}. 
% \begin{claim}[Giant component]
%     Let $G$ be a $(q, k)$-unbreakable directed graph, and let  $F \subseteq E(G)$ be a fault set of size at most $|F| \leq k$. Then $G-F$ contains a strongly connected component $S$ of size $|S| \geq n - 2q$.  We refer to $S$ as a giant component of $G-F$. 
%     \label{clm:giant_component_simpler}
% \end{claim}
\begin{claim}[Giant component]
    Let $G$ be a directed graph, let $U \subseteq V(G)$ be a $(q, k)$-unbreakable set in $G$, and let  $F \subseteq E(G)$ be a fault set of size at most $|F| \leq k$. Then $G-F$ contains a strongly connected component $C$ containing at least $|C \cap U| \geq |U|-2q$ of the vertices in $U$.  We refer to $C$ as a giant component of $U$ in $G-F$. 
    \label{clm:giant_component_first}
\end{claim}
Since graph $G$ is $(q, k)$-unbreakable, \Cref{clm:giant_component_first} implies that after any $k$ edge faults $F$ in graph $G$, at least $n-2q$ vertices in $G$ are contained in the same strongly connected component (SCC) $S$ in $G-F$.

By the edge-minimality of $G$, for each edge $e = (u, v) \in E(G)$, there exists a fault set $F_e \subseteq E(G)$ of size at most $|F_e| \leq k$ such that $u$ and $v$ are strongly connected in $G-F_e$ but not strongly connected in $(G - e) - F_e$.\footnote{This corresponds to \Cref{claim:decremental_struct} in \Cref{subsec:sourcewise}.} Let $C_e \subseteq V(G)$ be the SCC in $G-F_e$ containing vertices $u$ and $v$. We will bound the number of edges $e \in E(G)$  in two cases,   based on whether or not set $C_e$ is a giant component in $G-F_e$. 

\begin{figure}[ht]
    \centering
    \includegraphics[width=0.5\textwidth]{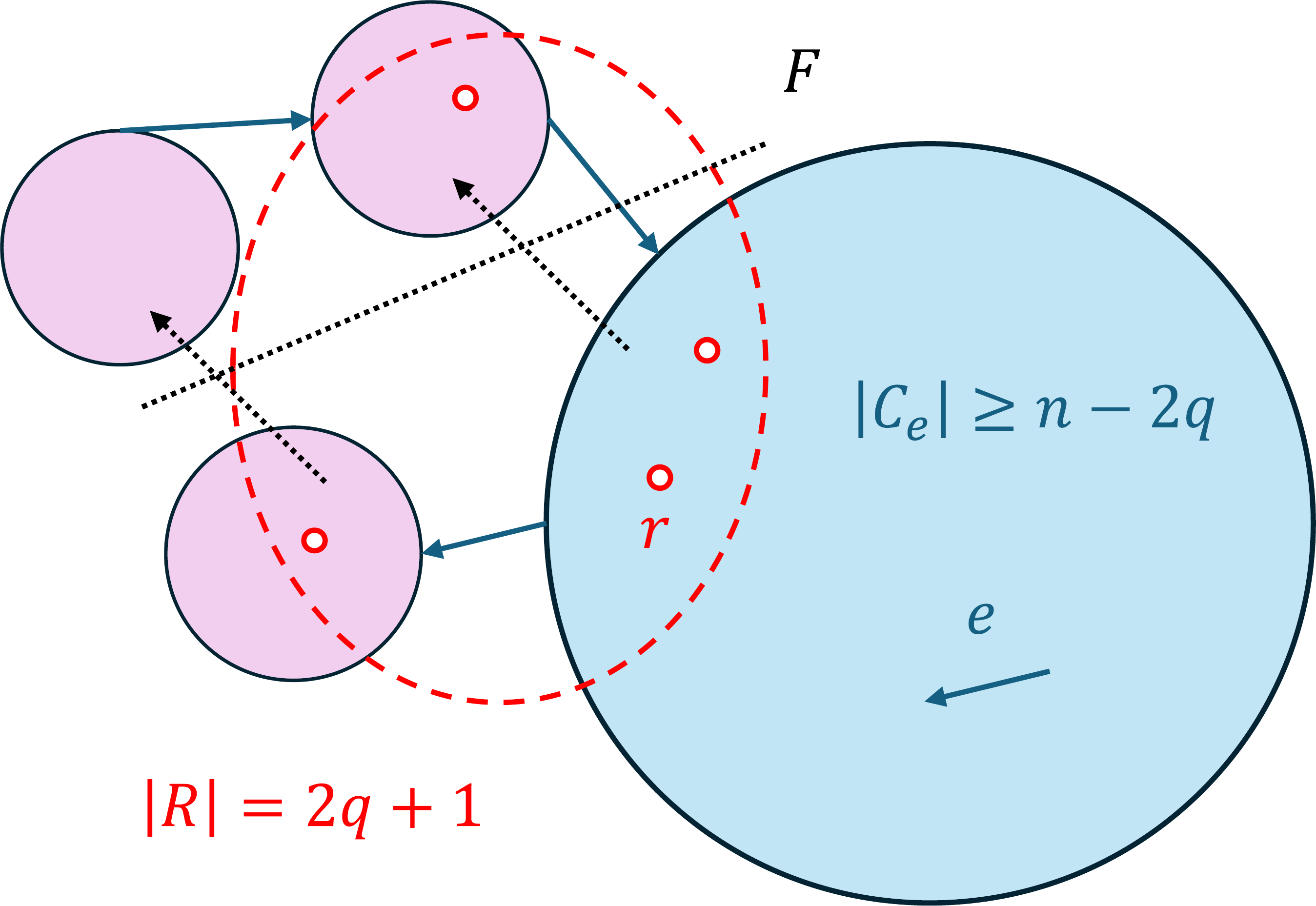}
    \caption{Illustration of the giant component analysis in Case 1.}
    \label{fig:giant-scc}
\end{figure}

\textbf{Case 1:} $C_e$ is a giant component in $G-F_e$. 
In this case, we will directly use the $k$-FT single-source connectivity preservers of \cite{baswana2018fault}.
% (Recall that a $k$-FT single-source symmetric connectivity preserver of $G$ rooted at $r$ preserves the strong connectivity between $r$ and $V(G)$ under $k$ edge faults.)
Let $R \subseteq V(G)$ be an arbitrary set of  vertices of size $|R| = 2q+1$. Let $H \subseteq G$ be the graph obtained by including a $k$-FT single-source connectivity preserver of $G$ rooted at $r$ for all $r \in R$. 
Each of these $|R|$ preservers has size $O(2^kn)$ by \cite{baswana2018fault}, so graph $H$ has $|E(H)| = O(2^kqn)$ edges. 

We claim that graph $H$ contains every edge $e = (u, v)$ whose connected component $C_e$ in $G-F_e$ is a giant component.  Since $C_e$ is a giant component in $G-F_e$, we know that $|C_e| \geq n - 2q$, so $C_e \cap R \neq \emptyset$. Then subgraph $H$ contains a $k$-FT single-source connectivity preserver of $G$ rooted at some vertex $r \in C_e$, so $C_e$ is an SCC of $H - F_e$.

Finally, since $u$ and $v$ are strongly connected in $H-F_e$ but not strongly connected in $(H-e)-F_e$, we conclude that $e \in E(H)$. The argument in this case roughly corresponds to \Cref{clm:large_scc} in \Cref{sec:conn_pres}.

\begin{figure}[ht]
    \centering
    \includegraphics[width=0.4\textwidth]{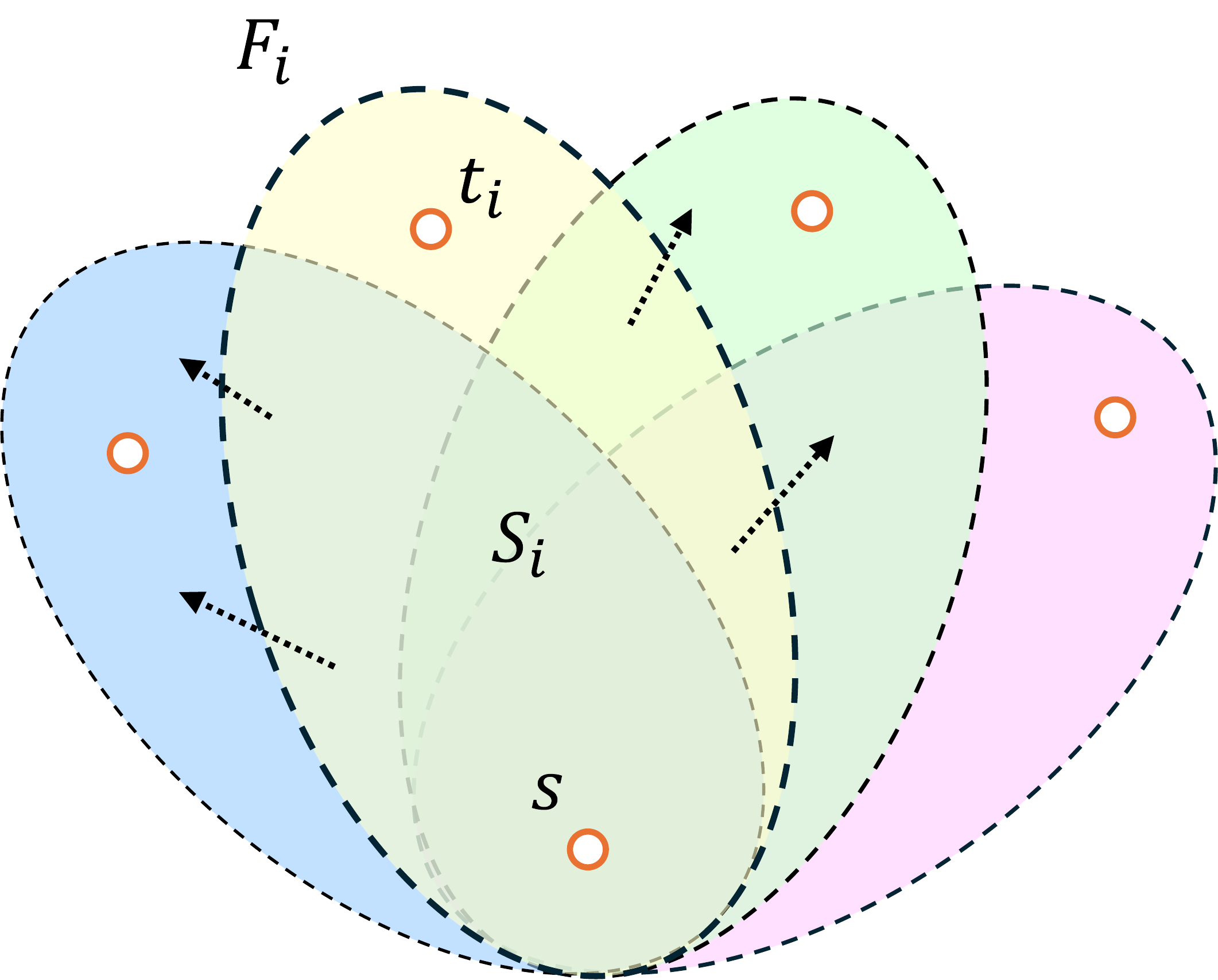}
    \caption{Illustration of the anti-isolation argument in Case 2.}
    \label{fig:small-scc}
\end{figure}

\textbf{Case 2:} $C_e$ is not a giant component in $G-F_e$. Let $E' \subseteq E(G)$ be the set of all edges $e \in E(G)$ such that $C_e$ is not a giant component. We will give an informal argument that $|E'| = O(2^kqn)$. By an averaging argument, there must exist vertices $s, t_1, \dots, t_r \in V(G)$ such that $(s, t_i) \in E'$ for $i \in [1, r]$ and $r = \Theta(|E'|/n)$. Let $e_i = (s, t_i) \in E'$ for $i \in [1, r]$. 
By \Cref{clm:giant_component_first}, the SCC $C_{e_i}$ in $G - F_{e_i}$ containing vertices $\{s, t_i\}$ has size at most $|C_{e_i}| \le 2q$ for each $i \in [1, r]$. Since the sets $C_{e_i}$ each have size $|C_{e_i}| = O(q)$, we can ensure that $t_j \not \in C_{e_i}$ for all $j \neq i$, using a greedy argument that decreases the size of our set of vertices $\{t_1, \dots, t_r\}$ by at most a factor of $q$ (see \Cref{clm:degeneracy} for details).

Next, we will apply the so-called \textit{anti-isolation lemma}, which was previously proved in \cite{pilipczuk2018directed} with the weaker bound of $r \le k4^k$. Below, we state an improved bound using insights from \cite{baswana2018fault}. See \Cref{subsec:isolation} for the formal proof.

% We can now apply the \textit{anti-isolation lemma}, which we formally prove in \Cref{subsec:isolation}.\footnote{The anti-isolation was previously proved in \cite{pilipczuk2018directed} with the weaker bound of $r \le k4^k$.}

\begin{restatable}[Improved anti-isolation]{lemma}{isolation} \label{lem:anti isolation}
Let $s, t_1, \dots, t_r$ be vertices in a directed graph $G$, and let $F_1, \dots, F_r \subseteq E(G)$ be edge sets each of size at most $k$ such that for all $i, j \in [1, r]$ there is an $(s, t_j)$-path in $G - F_i$ if and only if $i = j$. Then $r \leq 2^k$. 
\end{restatable}

% \thatchaphol{highlight it a bit. Say that it is based on new insight on important cuts. Should we have a paragraph sketching how to prove it?}

By our earlier discussion, we know that $s$ and $t_j$ are strongly connected in $G - F_{e_i}$ if and only if $i = j$ (see \Cref{fig:small-scc} for reference). For simplicity, we will further assume that $s$ can reach $t_j$ in $G-F_{e_i}$ if and only if $i = j$. Then the anti-isolation lemma directly implies that $\Theta(|E'|/(qn)) = r \leq 2^k$, so we conclude that $|E'| = O(2^kqn)$.  This argument roughly corresponds to \Cref{clm:L_bound}. 

Our two-case analysis shows that $(q, k)$-unbreakable graphs have $k$-FT connectivity preservers of size $O(2^kqn)$. Our proof in \Cref{subsec:existential} roughly follows this argument. We now build some intuition for how we combine these ideas with the directed expander hierarchy.

\paragraph{Directed expander hierarchy with two levels.} The two-level directed expander hierarchy has a particularly simple interpretation. It implies that given a directed graph $G$, there exists a set of vertices $U \subseteq V(G)$ such that $U$ is $(q, k)$-unbreakable in $G$ and every SCC $C$ of $G - U$ is $(q, k)$-unbreakable in $G-U$, for an appropriate choice of $q$. 

We briefly describe how we can use the ideas we have developed, along with this two-level hierarchy, to bound the number of edges we need to preserve in a $k$-FT connectivity preserver of $G$. As we have done throughout this section, we will assume that $G$ is an edge-minimal $k$-FT connectivity preserver of itself;  our goal then is to upper bound $|E(G)|$.

If $C$ is an SCC in $G - U$, then since $G[C]$ is $(q, k)$-unbreakable, we can upper bound the number of edges in $G[C]$ using the sparse $k$-FT connectivity preservers for unbreakable graphs that we developed earlier; the same holds for $G[U]$ as well. The remaining edges in $E(G)$ that we need to upper bound fall into two categories:
 1) edges between distinct SCC's $C_1, C_2$ in $G - U$, and
2) edges between $U$ and $V(G) - U$. 

Let $H_U \subseteq G$ be an edge-minimal subgraph of $G$ that preserves the $k$-FT connectivity between all pairs of vertices  $(u, v)$ in $U \times V(G)$. We say that $H_U$ is a $k$-FT \textit{sourcewise} connectivity preserver of $G$ with respect to set $U$.\footnote{We formally define $k$-FT sourcewise connectivity preservers in \Cref{def:sourcewise}.} It is not difficult to see that $H_U$ must contain all edges in both of the two categories of edges in $E(G)$ identified above. Then in order to upper bound $|E(G)|$ it is sufficient to upper bound $|E(H_U)|$.

As it turns out, the second sparse preserver argument we saw in this section can be easily extended to solve this source-restricted connectivity problem using only $O(4^kqn)$ edges ( \Cref{lem:unbreak_pres}).

\paragraph{Full directed expander hierarchy.} In our full directed expander hierarchy in \Cref{lem:expander_hierarchy}, we will have $\ell = O(\log n)$ levels, which allows us to set $q = O(k)$. The ideas discussed in the previous paragraph easily extend to this full expander hierarchy, leading to a final size bound of $O(\ell \cdot 4^kqn) = O(k4^k n \log n)$, as we prove in \Cref{subsubsec:expanderhierarch}.

% \thatchaphol{Sketch what you need to do to not assume greedy.}

\paragraph{Computing a $k$-FT  connectivity preserver in FPT-time.} So far, all of the arguments in this section use a greedy algorithm to construct edge-minimal $k$-FT connectivity preservers. Naively implementing this algorithm requires $n^{O(k)}$ time to check every fault set $F \subseteq E(G)$ of size $|F| \leq k$. We implement a modified version of the greedy algorithm in $\text{poly}(2^kn)$ time by adapting some of the ideas in our existential bound. Our FPT-time algorithm uses an algorithmic lemma about important cuts (\Cref{lem:important_separators}), which we prove in \Cref{sec:cut_theorem}.

\subsection{$k$-Connectivity Preservers}

Next, we describe the high-level idea how to show a $k$-connectivity preserver with $O(n\sqrt{nk})$ edges. Again, we start by assuming the graph $G$ is an edge-minimal $k$-connectivity preserver and argue that $G$ must have $O(n\sqrt{nk})$ edges. 
% Note that, for $k$-connectivity preservers, one can construct an edge-minimal graph in pol

% Below, we consider the high-level ideas for our $k$-connectivity preservers with size $O(n\sqrt{nk})$. 
% In the following, we let $q=\sqrt{nk}$ and will perform a different two-level decomposition with regard to $(q,k)$-unbreakability. Here, we still assume the graph $G$ is an edge-minimal $k$-connectivity preserver.

\paragraph{Decomposition.} Our argument is based on the following two-level decomposition. Set $q=\sqrt{nk}$. 
We maintain a partition $\mathcal{S}$ of vertex set $V$ and a collection of cuts $\mathcal{C}$. Initially, we let $\mathcal{S} = \{V\}$ and $\mathcal{C} = \emptyset$. While there exists a set $S \in \mathcal{S}$ and a cut $(L,R)$ of size  $|\delta^+(L)| \leq k$ such that $|S \cap L| \geq q$ and $|S \cap R|\geq q$, we replace $S$ using $S \cap L$ and $S \cap R$ in $\mathcal{S}$ and add the cut $(L, R)$ to $\mathcal{C}$. When the process terminates, we observe that:
\begin{itemize}
    \item Every set $S \in \mathcal{S}$ has size $|S| \geq q$ and is $(q,k)$-unbreakable.
    \item All edges between components in $\mathcal{S}$ form a subset of $\cup_{(L,R) \in \mathcal{C}} (\delta^+(L) \cup \delta^-(L))$.
\end{itemize}

By these two observations, we bound the number of edges inside some final component $S$, or between two final components, respectively. 

\paragraph{Inside components.} We bound the edges inside each component using the same strategy as we bound the size of a $k$-FT preserver in $(q,k+1)$-unbreakable graphs. By minimality, we can prove similarly that every edge $e \in G$ is in a small cut $(X,V-X)$ where $|\delta^+(X)| \leq k$. Let us consider any $S \in \mathcal{S}$ and any $e=(u,v)\in G[S]$. We use the same charging argument, with the only difference that, now we instead argue that either the out-degree of $u$ that is \emph{inside} $G[S]$ is at most $q+k$, or the in-degree of $v$ that is \emph{inside} $G[S]$ is at most $q+k$. Finally, we deduce that the total number of edges inside $G[S]$ is $O((q+k)|S|)$ for any final component $S \in \mathcal{S}$ and all components sum up to $O((q+k)n) = O(n\sqrt{nk})$.

\paragraph{Between components.} Since all of the edges between components in $\mathcal{S}$ are contained in the set $\cup_{(L,R) \in \mathcal{C}} (\delta^+(L) \cup \delta^-(L))$, we aim to bound the number of edges $|\delta^+(L) \cup \delta^-(L)|$ crossing each cut $(L,R)$ we operated on in the decomposition. By our decomposition, $|\delta^+(L)| \leq k$. We will now show that $|\delta^-(L)| \leq kn$.

Let us first think about a \emph{path} from any vertex $s$ to any vertex $t$. Since it will go through the edges $\delta^+(L)$ and $\delta^-(L)$ interleavingly, and $|\delta^+(L)|\leq k$, the number of edges in $\delta^-(L)$ it goes through can be at most $k+1$. Similarly, for a \emph{flow} with at most $k$ edge-disjoint paths from any vertex $s$ to any vertex $t$, since every flow path will still go through the edges $\delta^+(L)$ and $\delta^-(L)$ interleavingly, and all of the paths go through $\delta^+(L)$ at most $k$ times in total, the flow goes through $\delta^-(L)$ at most $k+k=2k$ times. Therefore, $2k$ edges in $\delta^-(L)$ will be sufficient to preserve the $k$-bounded connectivity between any vertex pair $(s,t)$.

Another observation is, to preserve $k$-connectivity for \emph{all pairs} of vertices $(s,t)$, it suffices that we preserve $k$-connectivity for only $O(n)$ \emph{demand pairs} of vertices $(s,t)$. These demand pairs can be viewed as all edges in some analog of Gomory-Hu trees \cite{gomory1961multi,cheng1991ancestor} for symmetric connectivity in directed graphs.\footnote{To be more specific, the trees lose some ``cut-equivalency" but still satisfy that the connectivity between any pair of vertices is the minimum connectivity on the tree path. See \Cref{lemma: demand-pairs} for details.} Therefore, to preserve $k$-connectivity for these $O(n)$ \emph{demand pairs}, we  only need $O(nk)$ edges in $\delta^-(L)$. This bound on the number of edges crossing cut $(L, R)$ in $\mathcal{C}$ is sufficient to show that there are at most $O(n \sqrt{nk})$ edges between components.

% \newpage

\section{Preliminaries}
\label{sec:prelim}

Let $G = (V, E)$ be a directed graph with $n$ vertices and $m$ edges.  For any edge set $F \subseteq V \times V$, we use $G-F$ to denote the graph formed by removing edges in $F$ from $G$, and we use $G + F$ to denote the graph formed by adding edges in $F$ to $G$. When $F=\{e\}$ is a singleton set, we  abbreviate this to $G-e$ or $G+e$. For a vertex set $S \subseteq V$, we use $G[S]$ to denote the induced subgraph of $S$ in $G$. 

For any vertex $v \in V$, we use $\delta_G^+(v)$ and $\delta_G^-(v)$ to denote the set of outgoing edges from $v$ and incoming edges to $v$ respectively. When $(u,v) \in E$, we say $v$ is an out-neighbor of $u$, and $u$ is an in-neighbor of $v$. We denote the set of out- and in- neighbors of $v\in V$ in $G$ by $N_G^+(v)$ and $N_G^-(v)$.

We define a cut as a partition $(S, V-S)$ of the vertex set $V$ into two disjoint subsets, where $S \subseteq V$. We use $\delta_G^+(S)$ to denote the set of outgoing edges from $S$, and $\delta_G^-(S)$ to denote the set of incoming edges to $S$. For two disjoint sets $X, Y \subseteq V$, we say that $(S, V-S)$ is an $(X, Y)$-cut if $X \subseteq S$ and $Y \subseteq V-S$. In particular, if $X = \{x\}$ and $Y = \{y\}$, we call $(S, V-S)$ an $(x, y)$-cut. 
We typically refer to the size of a cut $(S, V-S)$ as the number of outgoing edges from $S$, $|\delta^+_G(S)|$. In particular, we say that an  $(X, Y)$-cut $(S, V-S)$ is  a minimum $(X, Y)$-cut if $|\delta^+_G(S)| \leq |\delta^+_G(S')|$ for every $(X, Y)$-cut $(S', V-S')$. 

We define $\flow(G, s,t)$ to be the maximum number of pairwise edge-disjoint paths from $s$ to $t$ in $G$. Additionally, we define the \emph{symmetric} connectivity $\lambda_G(s,t)$ between $s$ and $t$ to be $\lambda_G(s, t) = \min(\flow(G, s,t),\flow(G, t,s))$. We define the $k$-bounded connectivity to be $\lambda_{G}^{k}(s,t)=\min(\lambda_{G}(s,t),k)$. For terminal sets $X,Y$, we define $\flow(G,X,Y)$ to be the maximum number of pairwise edge-disjoint $(X, Y)$-paths.

% For two disjoint vertex sets $S,T \subseteq V$, we can similarly define $(R,V-R)$ as an $S-T$ (edge) cut if $S \subseteq R$ and $T \subseteq(V-R)$. 

\subsection{Farthest minimum cuts and important cuts}

We need to define a class of well-structured $(X, Y)$-cuts, which we call reachable cuts.

\begin{definition}[Reachable Cuts]
        We say that an $(X, Y)$-cut $(S, V-S)$ is \textit{out-reachable} if for every $s \in S$ there exists an $x \in X$ such that $x$ can reach $s$ in $G - \delta_G^+(S)$.  Likewise, we say that $(S, V-S)$ is \textit{in-reachable} if for every $s \in S$ there exists an $x \in X$ such that $s$ can reach $x$ in $G - \delta^-_G(S)$.  
\end{definition}

In some sense, reachable cuts are without loss of generality. 

\begin{claim} \label{clm: in out wlog}
    Let $(S, V-S)$ be an $(X, Y)$-cut. Then there exists an out-reachable  $(X, Y)$-cut $(S', V-S')$ such that $\delta^+_G(S') \subseteq \delta^+_G(S)$. Likewise, there exists an in-reachable $(X, Y)$-cut $(S'', V-S'')$ such that  $\delta^-_G(S'') \subseteq \delta^-_G(S)$.
\end{claim}

\begin{proof}
    Consider the out-reachable case. For any $(X, Y)$-cut $(S, V-S)$, we let $S'$ be the set that $X$ can reach in $G-\delta^+_G(S)$. It follows that $S'$ is out-reachable and $\delta^+_G(S') \subseteq \delta^+_G(S)$. The in-reachable case can be proved in the same way.
\end{proof}

A classic result of Ford-Fulkerson implies the existence of a unique min-cut that is ``farther'' than every other min-cut. 

\begin{lemma}[Farthest Minimum Cuts \cite{ford2015flows}]
Given a directed graph $G$, we say that an $(X, Y)$-cut $(S, V-S)$ is a farthest minimum $(X, Y)$-cut in $G$ if $(S, V-S)$ is an out-reachable, minimum $(X, Y)$-cut, and there does not exist an out-reachable, minimum $(X, Y)$-cut $(S', V-S')$ such that $S \subset S'$. 
For disjoint vertex sets $X,Y \subseteq V(G)$, the unique farthest minimum $(X, Y)$-cut always exists and can be computed in max-flow time.
    \label{lemma:FMC}
\end{lemma}

We denote the farthest minimum $(X, Y)$-cut in $G$ by $\fmc(G, X, Y)$. Important cuts generalize the notion of farthest minimum cuts to cuts that are not min-cuts.

% \begin{definition}[Minimal Cuts]
%     We say that an $(X, Y)$-cut $(S, V-S)$ is \textit{out-minimal} if $\delta^+(S)$ is minimal over all $(X, Y)$-cuts. More precisely, $(S, V-S)$ is out-minimal there does not exist an $(X, Y)$-cut $(S', V-S')$ such that $\delta^+(S') \subset \delta^+(S)$. If instead $\delta^-(S)$ is minimal, then we say that cut $(S, V-S)$ is \textit{in-minimal}. 
% \end{definition}

\begin{definition}[Important Cuts \cite{marx2011important}]
We say that an $(X, Y)$-cut $(S, V-S)$ is an important $(X, Y)$-cut if $(S, V-S)$ is out-reachable, and there does not exist an out-reachable $(X, Y)$-cut $(S', V-S')$ with size at most $|\delta^+_G(S')| \leq |\delta^+_G(S)|$ such that $S \subset S'$.  
\end{definition}

The definition of important cuts is well-known in the parametrized complexity community. We also state the following claim, which is clear from definition:

\begin{claim} \label{clm:impcut farthest}
    If $(S, V-S)$ is an out-reachable $(X, Y)$-cut of size $|\delta^+_G(S)| \leq k$, then there exists an important $(X, Y)$-cut $(S', V-S')$ of size $|\delta^+_G(S')| \leq k$ such that $S \subseteq S'$. 
\end{claim}

\begin{proof}
    If $(S, V-S)$ is important, then we let $S'=S$ and the claim holds. Otherwise, since $(S, V-S)$ is not important, we can find $(S',V-S')$ that $|\delta^+_G(S')| \leq |\delta^+_G(S)|$ and $S \subset S'$. We can set $S=S'$ and repeat this procedure until $(S,V-S)$ is important.
\end{proof}

\subsection{Unbreakable sets and directed expander hierarchies}
\label{subsec:expander_hierarchy}

We begin this section by defining $(q, k)$-unbreakable sets. 

\unbreakable*

Unbreakable sets have the useful property that if we remove a small number of edges from the graph, then almost all of the vertices in the unbreakable set will lie in the same strongly connected component in the resulting graph. We defer its proof to \Cref{sec:omit_giant}.

\begin{claim}[Giant component]
    Let $G$ be a directed graph, let $U \subseteq V(G)$ be a $(q, k)$-unbreakable terminal set in $G$, and let  $F \subseteq E(G)$ be a set of edges with size at most $k$. Then $G-F$ contains a strongly connected component $C$ containing at least $|C \cap U| \geq |U|-2q$ of the nodes in $U$.  We refer to $C$ as a giant component of $U$ in $G-F$. 
    \label{clm:giant_component}
\end{claim}

Our upper bounds for fault-tolerant symmetric connectivity preservers will rely on the following directed expander hierarchy decomposition for unbreakable sets.

\begin{lemma}[Directed Expander Hierarchy for Unbreakable Sets]
Let $G$ be an $n$-node directed graph, and let $q, k$ be positive integers with $q \geq k/\phi$, where $0 < \phi \leq 1$. Then there exists a partition $\{V_1, V_2, \dots, V_{\ell}\}$ of  $V(G)$ with the following properties:
\begin{enumerate}
    \item Fix an $i \in [1, \ell]$.   Let $V_{\leq i} = V_1 \cup \dots \cup V_i$, and let $C \subseteq V(G)$ be a strong component in $G[V_{\leq i}]$. Then the set $V_i \cap C$ is $(q, k)$-unbreakable in $G[C]$.
    \item The number of sets in the partition is at most $\ell = O(\log n)$.
\end{enumerate}

Moreover, when $\phi = 1/2$, this partition can be computed in exponential time, and when $\phi = O(1/\sqrt{\log n})$, this partition can be computed in polynomial time.  
% \gary{TODO: include the time complexity here!}
% Moreover, if we require that $q \geq k \cdot c \sqrt{\log n}$ for a sufficiently large constant $c>0$, then this partition can be computed in time $O()$. \gary{TODO: include the time complexity here!}
    \label{lem:expander_hierarchy}
\end{lemma}

A version of expander hierarchies in \emph{undirected} graphs was first introduced in the context of oblivious routing \cite{racke2002minimizing,bienkowski2003practical,harrelson2003polynomial} and is called \emph{tree flow sparsifiers} or \emph{R\"{a}cke tree}. It later finds many more applications for devising almost-linear time flow algorithms \cite{racke2014computing,abraham2016fully,goranci2021expander,van2024almost,li2025congestion}.

In 2007, a weaker variant of the expander hierarchy was introduced by Patrascu and Thorup \cite{patrascu2007planning} for connectivity oracles under edge faults.
It turns out, in contrast to tree flow sparsifiers, that this weaker variant naturally generalizes from the usual edge expansion to vertex expansion \cite{long2022near,long2024connectivity}. 
Recently, \cite{bernstein2024maximum} introduced a directed expander hierarchy as a further generalization to directed expansion. They also presented a fast construction algorithm that is highly complicated.     

Here, to prove \Cref{lem:expander_hierarchy}, we give a very simple construction of the directed expander hierarchy. The proof follows closely and generalizes the construction in undirected graphs of  \cite{long2024connectivity}. For the sake of completeness, we defer its proof to \Cref{sec:expander_hierarchy}.

\section{Fault-Tolerant Connectivity Preservers}
\label{sec:conn_pres}

We now present our new results for fault-tolerant connectivity preservers. In \cref{subsec:existential}, we prove our existential upper bound for fault-tolerant connectivity preservers (\Cref{thm:existentialtwo}). In  \cref{subsec:fpt}, we show how to modify our construction to run in FPT-time (Theorem \ref{thm:fpt}).

Throughout this section, we will repeatedly use the following lemma about important cuts, which is proved in full generality in  \Cref{thm:ecut} of \Cref{sec:cut_theorem}.

\begin{restatable}[Important cut container]{lemma}{impcutlem} \label{lem:important_separators} %\thatchaphol{This name is not suggestive. How about this? Important Cuts Container}\benyu{Sounds good} 
    Let $k$ be a positive integer, let $G$ be an $n$-vertex directed graph, and let $X, Y \subseteq V(G)$ be disjoint sets. There exists an $(X, Y)$-cut $(S, V-S)$ of size at most $|\delta^+(S)| \leq 2^{k-1}$ such that every important $(X, Y)$-cut $(S', V-S')$ of size at most $|\delta^+(S')| \leq k$ satisfies $S' \subseteq S$. Cut $(S, V-S)$ can be computed in $O(k2^km+k^24^k)$ time.  
\end{restatable}

We define some new notation that we will use when we apply \Cref{lem:important_separators}. 
\begin{itemize}
\item We will let $\is(G, X, Y, k, +)$ denote an $(X, Y)$-cut $(S, V-S)$ in $G$ of size at most $|\delta^+(S)| \leq 2^{k-1}$ such that every out-reachable $(X, Y)$-cut $(S', V-S')$ of size $|\delta^+(S')|\leq k$ satisfies $S' \subseteq S$. 
\item We will let $\is(G, X, Y, k, -)$ denote an $(X, Y)$-cut $(S, V-S)$ in $G$ of size at most $|\delta^-(S)| \leq 2^{k-1}$ such that every in-reachable $(X, Y)$-cut $(S', V-S')$ of size $|\delta^-(S')| \leq k$ satisfies $S' \subseteq S$.
\end{itemize}
\noindent 
By Lemma \ref{lem:important_separators} and \Cref{clm:impcut farthest}, $\is(G, X, Y, k, +)$ and $\is(G, X, Y, k, -)$ are well-defined and can be computed in  $O(k2^km+k^24^k)$ time.

Additionally, throughout this section we will construct $k$-FT single-source connectivity preservers as a subroutine in our algorithms. As mentioned in \Cref{def:k-fault preservers}, $k$-FT single-source connectivity preservers maintain the strong connectivity between a source node $s$ and the rest of the vertex set $V(G)$ under $k$ edge faults. We use $\ssc(G, s, k)$ to denote a $k$-FT single-source connectivity preserver of $G$ rooted at source vertex $s$ with $O(2^k|V(G)|)$ edges. Subgraph $\ssc(G, s, k)$ always exists  by \cite{baswana2018fault}.

\subsection{Existential upper bound for fault-tolerant connectivity preservers}
\label{subsec:existential}

The goal of this section is to present our existential upper bound for fault-tolerant connectivity preservers. 

\begin{restatable}{theorem}{existentialtwo} \label{thm:existentialtwo}
For every positive integer $k$, every $n$-vertex directed graph admits a $k$-fault-tolerant connectivity preserver with $O(k4^kn \log n)$ edges.
\end{restatable}

We will prove \Cref{thm:existentialtwo} in two steps. In the first step, we will show how to construct sparse fault-tolerant \textit{sourcewise}  connectivity preservers between a well-connected set of source vertices $U$ and the rest of the graph (Lemma \ref{lem:unbreak_pres}). In the second step, we combine our sourcewise  preservers with a directed expander hierarchy (Lemma \ref{lem:expander_hierarchy}) to construct a sparse fault-tolerant connectivity preserver of the input graph.  

\subsubsection{Fault-tolerant sourcewise connectivity preservers for unbreakable sets} \label{subsec:sourcewise}
We begin with a formal definition of fault-tolerant sourcewise connectivity preservers.

\begin{definition}[Fault-tolerant sourcewise connectivity preservers] \label{def:sourcewise}
    Let $G$ be an $n$-vertex directed graph, and let $S \subseteq V(G)$ be a set of source vertices. We say that a subgraph $H \subseteq G$ is a $k$-FT  sourcewise connectivity preserver of $G$ with respect to set $S$ if for any pair of vertices $(s, v) \in S \times V(G)$ and any edge fault set $F \subseteq E(G)$ of size at most $|F| \leq k$, $s$ and $v$ are strongly connected in $G-F$ if and only if $s$ and $v$ are strongly connected in $H-F$. 
\end{definition}

\noindent Our  fault-tolerant sourcewise connectivity preservers can be summarized with the following lemma. 

\begin{lemma} 
    Let $G$ be an $n$-vertex directed graph, and let  $U \subseteq V(G)$ be a $(q, k)$-unbreakable set in $G$.
    There exists a $k$-FT  sourcewise connectivity preserver $H$ of graph $G$ with respect to source vertices $U$ of size $|E(H)| = O(4^kq n)$.
\label{lem:unbreak_pres}
\end{lemma}

\noindent 
We will construct our  fault-tolerant sourcewise connectivity preserver using a decremental greedy strategy that is standard in the area of graph spanners and preservers.

\paragraph{Greedy Algorithm.} Let $G = (V, E)$ be an $n$-vertex directed graph, and let $U \subseteq V$ be a $(q, k)$-unbreakable set in $G$. We may assume $G$ is strongly connected, without loss of generality.  We will initialize our fault-tolerant sourcewise connectivity preserver $H$ to be the input graph, so $H \leftarrow G$. 
While there exists an edge $e \in E(H)$ such that $H - e$ is a $k$-FT sourcewise connectivity preserver of $G$ with respect to set $U$, we will remove edge $e$ from $H$, so that $H \leftarrow H - e$. This algorithm  terminates and outputs $H$ when no such edge $e \in E(H)$ exists.

\paragraph{Correctness.} 

By construction, $H$ is a $k$-FT sourcewise connectivity preserver of $G$ with respect to set $U$. Specifically, $H$ is an \textit{edge-minimal} $k$-FT sourcewise connectivity preserver of $G$ with respect to $U$, by the terminating condition of the greedy algorithm.

\paragraph{Time Complexity.} There are $n^{O(k)}$ fault sets $F \subseteq E(G)$ of size $|F| \leq k$. Consequently, we can check if a subgraph $H \subseteq G$ is a $k$-FT sourcewise connectivity preserver of  $G$ with respect to  $U$ in $n^{O(k)}$ time. We conclude that the greedy algorithm can be implemented in $n^{O(k)}$ time.

\paragraph{Size Analysis.} Before we can bound the size of $E(H)$, we need to define a subgraph $J \subseteq H$ that is useful for our analysis. Let $U'$ be an arbitrary subset $U' \subseteq U$ of size $|U'| = \min(2q+1, |U|)$. Let $$J = \bigcup_{u \in U'} \ssc(H, u, k)$$
be the subgraph of $H$ obtained by unioning $|U'|$ many $k$-FT single-source connectivity preservers, each rooted at a distinct vertex in $U'$. 
The following claim states that subgraph $J$ preserves strong connectivity between pairs of vertices $(u, v) \in U \times V(G)$ for a certain category of fault sets. 

\begin{claim}[Large SCC fault sets]
 Let $F \subseteq E(G)$ be a fault set of size $|F| \leq k$, and let $C \subseteq V(G)$ be an SCC of $G-F$. If $|C \cap U| > 2q$, then vertices in $C$ are pairwise strongly connected  in $J-F$.
\label{clm:large_scc}
\end{claim}
\begin{proof}
By \Cref{clm:giant_component},  if $|C \cap U| > 2q$, then $|C \cap U| \ge |U| - 2q$. Since $|U'| = \min(2q+1, |U|)$, it follows that $C \cap U' \ne \emptyset$.  Let $x \in  C \cap U'$ be a vertex in the nonempty intersection of $U'$ and $C$.

Now we make the following sequence of observations. First, vertices in $C$ are pairwise strongly connected in $G-F$. Second, since $H$ is a $k$-FT sourcewise connectivity preserver of $G$ with respect to $U$, and $C \cap U \ne \emptyset$, it follows that vertices in $C$ are pairwise strongly connected in $H-F$. Finally, subgraph $J$ contains $\ssc(H, x, k)$, where $x \in C$, so we conclude that vertices in $C$ are pairwise strongly connected in $J-F$ as well.
\end{proof}

Since $J$ is the union of $O(q)$ many $k$-FT single-source connectivity preservers, $|E(J)| = O(2^kqn)$ by  Theorem 1.1 of \cite{baswana2018fault}.  Then to prove $|E(H)| = O(4^kqn)$, it will be sufficient to prove that $|E(H) - E(J)| = O(4^kqn)$. Before proceeding with the proof, we will need the following structural claim about the edges in $H$.

\begin{claim}
    For each edge $(s, t) \in E(H)$, there exists a vertex $u \in U$ and a fault set $F \subseteq E(H)$, $|F| \leq k$, such that:
    \begin{enumerate}
        \item Vertices $s$, $t$, and  $u$ are  strongly connected  in $H - F$, and  
        \item $H-F$ contains either 
    \begin{itemize}
        \item  an out-reachable $(u, t)$-cut $(S, V-S)$ such that $\delta^+_{H-F}(S) = \{(s, t)\}$, or
        \item  an in-reachable $(u, s)$-cut $(S, V-S)$ such that $\delta^-_{H-F}(S) =\{(s, t)\}$.
    \end{itemize}
    \end{enumerate}
    \label{claim:decremental_struct}
\end{claim}
\begin{proof}
If $(s, t) \in E(H)$, then there exists a fault set $F$, with $|F| \leq k$, such that $(H - (s, t)) - F$ does not preserve the strong connectivity of a pair of vertices $u, v$ in $H-F$, where $u \in U$ and $v \in V$. Specifically, vertices $u, v$ are strongly connected in $H - F$, but not strongly connected in $(H - (s, t)) - F$. This implies that edge $(s, t)$ is inside the SCC containing $u$ and $v$ in $H-F$, so vertices $s, t, $ and $u$ are all strongly connected in $H-F$, as claimed.

Since $u$ and $v$ are strongly connected in $H-F$ and not strongly connected in $(H-(s, t)) - F$, there must  exist a $(u, v)$-cut $(S, V-S)$ in $H$ such that $\delta_{H-F}^+(S) = \{(s, t)\}$ or $\delta_{H-F}^-(S) = \{(s, t)\}$. 
If $\delta_{H-F}^+(S) = \{(s, t)\}$, then by \Cref{clm: in out wlog} there exists an out-reachable $(u, t)$-cut $(S', V-S')$ such that $\delta^+_{H-F}(S') = \{(s, t)\}$. Else if $\delta_{H-F}^-(S) = \{(s, t)\}$, then by \Cref{clm: in out wlog} there exists an in-reachable $(u, s)$-cut $(S', V-S')$ such that $\delta^-_{H-F}(S') = \{(s, t)\}$.  
\end{proof}

% \begin{claim}
%     \label{clm:small_valid_cuts}
%     For every edge $(s, t) \in E(H)$, there exists a vertex $u \in U$ such that $H$ has either
%     \begin{itemize}
%         \item  a valid $(u, t)$-cut $(S, V-S)$ of size $|\delta^+_H(S)| \leq k+1$ such that $(s, t) \in \delta^+_H(S)$, or
%         \item  a valid $(s, u)$-cut $(S, V-S)$ of size $|\delta^+_H(S)| \leq k+1$ such that $(s, t) \in \delta^+_H(S)$.
%     \end{itemize}
% \end{claim}
% \begin{proof}
    
% \end{proof}

\noindent 
For each edge $e = (s, t) \in E(H) - E(J)$ we will associate a label $\ell(e) \in U \times V$ as follows.  By \Cref{claim:decremental_struct}, there exist a vertex $u \in U$ and a fault set $F \subseteq E(H)$, $|F| \leq k$, such that vertices $s$, $t$, and $u$ are all strongly connected in $H-F$, but not all strongly connected in $(H-(s, t))-F$. 
\begin{itemize}
    \item If $H$ has an out-reachable $(u, t)$-cut $(S, V-S)$ such that $\delta^+_{H}(S) \subseteq F \cup \{(s, t)\}$, then we let $\ell(e) = (u, t).$
    \item  If $H$ has an in-reachable $(u, s)$-cut $(S, V-S)$ such that $\delta^-_{H}(S) \subseteq F \cup \{(s, t)\}$,    then we let $\ell(e) = (u, s)$. 
\end{itemize}
\noindent 
By Claim \ref{claim:decremental_struct}, $\ell(e)$ is well-defined for all $e \in E(H)$.\footnote{If there are multiple valid choices of a label $\ell(e)$, then we select one arbitrarily.}
Let $$\mathcal{L} = \{\ell(e) \mid e \in E(H) - E(J)\}.$$ 
The set of labels $\mathcal{L}$ is useful because we can use it to upper bound the size of $E(H)$. 

 \begin{claim}
     $|E(H)| \leq 2^{k+1} \cdot |\mathcal{L}| + O(2^{k}qn)$
     \label{clm:H_vs_L_bound}
 \end{claim}
 \begin{proof}
Since $|E(J)| = O(2^kqn)$ by an earlier discussion, it will be sufficient to show that $|E(H)-E(J)| \leq 2^{k+1} \cdot |\mathcal{L}|$. Suppose towards contradiction that $|E(H) - E(J)| > 2^{k+1} \cdot |\mathcal{L}|$. 
Then there exists a vertex pair $(u, v) \in U \times V$ such that 
$
|\{e \in E(H)-E(J) \mid \ell(e) = (u, v) \}| > 2^{k+1}.
$
Note that every edge $e \in E(H)$ with label $\ell(e) = (u, v)$ is of the form $e = (w, v)$ or $e = (v, w)$ for some vertex $w \in V$. Let $E_1 = \{ (w, v) \in E(H)-E(J) \mid \ell((w, v)) = (u, v) \}$, and let $E_2 = \{ (v, w) \in E(H)-E(J) \mid \ell((v, w)) = (u, v) \}$. Either $|E_1| > 2^{k}$ or $|E_2| > 2^{k}$.

Suppose that $|E_1| > 2^k$. 
By the definition of $\ell(e)$, for each edge $(w, v) \in E_1$, there exists an out-reachable $(u, v)$-cut $(S_w, V-S_w)$ of size $|\delta^+_H(S_w)| \leq k+1$ such that $(w, v) \in \delta^+_H(S_w)$. Let $(S', V-S') = \is(H, u, v, k+1, +)$. Then by Lemma \ref{lem:important_separators},  $(S', V-S')$ is a $(u, v)$-cut such that $S_w \subseteq S'$, so $(w, v) \in \delta_H^+(S')$. This implies that $E_1 \subseteq \delta_H^+(S')$. However, by the important cut lemma (Lemma \ref{lem:important_separators}), $|\delta_H^+(S')| \leq 2^{k}$, a contradiction.

If instead $|E_2| > 2^k$, then we let $(S', V-S') = \is(H, u, v, k+1, -)$. By a symmetric argument as before, we can argue that $E_2 \subseteq \delta_H^-(S')$. However, by Lemma \ref{lem:important_separators}, $|\delta_H^-(S')| \leq 2^{k}$, a contradiction.
 \end{proof}

By \Cref{clm:H_vs_L_bound},  it is sufficient to upper bound $|\mathcal{L}|$ in order to upper bound $|E(H)|$. We prove that $|\mathcal{L}| = O(2^kqn)$ in \Cref{clm:L_bound} by appealing to our important cut lemma (Lemma \ref{lem:important_separators}). The proof of \Cref{clm:L_bound} uses the following structural claim about the edges in $E(H) - E(J)$.

% \thatchaphol{I feel both the statement and the proof of the claim below are quite convoluted. There should be a simpler way to write this.}
\begin{claim}
Let $(s, t) \in E(H) - E(J)$. By Claim \ref{claim:decremental_struct}, there exists a vertex $u \in U$ and a fault set $F \subseteq E(H)$, $|F| \leq k$, such that vertices $s$, $t$, and $u$ are strongly connected in $H-F$, but not strongly connected in $(H-(s, t))-F$. Let $C$ be the SCC in $H-F$ containing vertices $s$, $t$, and $u$. Then $|C \cap U| \leq 2q$. 
\label{clm:small_scc_struct}
\end{claim}
\begin{proof}
Suppose, towards a contradiction, that $|C \cap U| > 2q$.  Let $C'$ be the SCC in $G-F$ containing vertices $s$, $t$, and $u$. Since $H$ is a $k$-FT sourcewise connectivity preserver of $G$ with respect to $U$ and $|F| \leq k$, it must follow that $C = C'$. 

Then  $|C' \cap U| > 2q$, and by \Cref{clm:large_scc},  vertices $s$, $t$, and $u$ are strongly connected in $J-F$. Because $J \subseteq H$ and $s$ and $t$ are not strongly connected in $(H-(s, t))-F$, it follows that $s$ and $t$ are not strongly connected in $(J - (s, t)) - F$ as well. Since $s$ and $t$ are strongly connected in $J-F$ but not strongly connected in $(J-(s, t)) - F$, we conclude that $(s, t) \in E(J)$.  This contradicts our assumption that $(s, t) \in E(H) - E(J)$. 
\end{proof}

Additionally, our proof of \Cref{clm:L_bound} will need the following general claim about set systems.

\begin{claim}
    Let $\mathcal{S} = \{S_u\}_{u \in U}$ be a family of sets indexed by a set $U$ such that $u \in S_u$ and $|S_u \cap U| \leq q$ for all $S_u \in \mathcal{S}$. 
Then there exists a subset $\mathcal{S}_1 \subseteq \mathcal{S}$ of size $|\mathcal{S}_1| \geq |\mathcal{S}|/(2q)$ such that for all distinct sets $S_u, S_{u'} \in \mathcal{S}_1$, $u' \notin S_u$ and $u \notin S_{u'}$. 
    \label{clm:degeneracy}
\end{claim}
\begin{proof}
We define a directed graph $G_\mathcal{S}$ with vertex set $V(G_\mathcal{S}) = \mathcal{S}$. We add the directed edge $(S_u,S_w) \in \mathcal{S} \times \mathcal{S}$ to $G_{\mathcal{S}}$ if $w \in S_u$. Since $|S_u \cap U| \leq q$, each vertex $S_u \in \mathcal{S}$ has out degree at most $q-1$ in graph $G_{\mathcal{S}}$. 

We say that a (possibly directed) graph is $\alpha$-degenerate if every subgraph of the graph contains a vertex of total degree at most $\alpha$. Since every vertex of $G_{\mathcal{S}}$ has outdegree at most $q-1$, every subgraph of $G_{\mathcal{S}}$ has average degree at most $2q-2$. Consequently, graph $G_{\mathcal{S}}$ is $(2q-2)$-degenerate. By the standard greedy vertex coloring algorithm for bounded-degeneracy graphs, graph $G_{\mathcal{S}}$ has a proper vertex coloring with at most $2q-1$ colors. 

In particular, this implies that graph $G_{\mathcal{S}}$ has an independent set $\mathcal{S}_1 \subseteq \mathcal{S}$ of size at least $|\mathcal{S}| / (2q-1) \geq |\mathcal{S}|/(2q)$. 
Since set $\mathcal{S}_1 \subseteq \mathcal{S}$ is an independent set in $G_{\mathcal{S}}$, it follows that for all distinct sets $S_u, S_{u'} \in \mathcal{S}_1$, $u' \not \in S_u$ and $u \not \in S_{u'}$, as claimed. 
\end{proof}

We are now ready to prove \Cref{clm:L_bound}.

 \begin{claim}
     $|\mathcal{L}| \leq 2^{k+4}qn$. 
\label{clm:L_bound}
 \end{claim}
  \begin{proof}
     Suppose towards contradiction that $|\mathcal{L}| > 2^{k+4}qn$. Then there exists a vertex $v \in V$ such that $|\mathcal{L} \cap (U \times \{v\})| >  2^{k+4}q$.  Let $\mathcal{L}_v = \mathcal{L} \cap (U \times \{v\})$. For each $(u, v) \in \mathcal{L}_v$, let $e_{u} = (s_u, t_u) \in E(H) - E(J)$ be an edge such that $\ell(e_{u}) = (u, v)$.  By \Cref{claim:decremental_struct} and the definition of $\ell(e_u)$, there exists a fault set $F_{u} \subseteq E(H)$, $|F_u| \leq k$, such that vertices $v$ and $u$ are  strongly connected in $H-F_u$, but not   strongly connected in $(H-e_u)-F_u$. Let $C_u$ denote the SCC in $H-F_u$ containing vertices $v$  and $u$. By Claim \ref{clm:small_scc_struct}, $|C_u \cap U| \leq 2q$.

     We claim that for each $(u, v) \in \mathcal{L}_v$, there exists a cut $(W_u, V-W_u)$ in $H$ such that
     \begin{enumerate}
         \item $C_u \subseteq W_u$,
         \item  $|W_u \cap U| \leq 4q$, and 
         \item either  $|\delta^+_H(W_u)| \leq k$ or $|\delta^-_H(W_u)| \leq k$.
     \end{enumerate}
 By Claim \ref{clm:giant_component}, there is an SCC $C$ in $G - F_u$ such that $|C \cap U| \geq |U| - 2q$. Since $H$ is a $k$-FT sourcewise connectivity preserver of $G$ with respect to $U$, set $C$ is an SCC in $H-F_u$ as well. If $C_u = C$, then $|U| \leq |C_u \cap U| + 2q \leq 4q$, so we can let $W_u = V$. Otherwise, $C_u$ and $C$ are distinct SCC's in $H-F_u$. Then there exists an $(C_u, C)$-cut $(W_u, V-W_u)$ in $H-F_u$ such that either $\delta^+_{H-F_u}(W_u) = \emptyset$ or $\delta^-_{H-F_u}(W_u) = \emptyset$, so $|\delta^+_H(W_u)| \leq k$ or $|\delta^-_H(W_u)| \leq k$. Moreover, $|W_u \cap U| \leq 2q$, since $|S \cap U| \geq |U| - 2q$. We conclude that cut $(W_u, V-W_u)$ satisfies the claim.

    Let $\mathcal{W} = \{W_u \mid (u, v) \in \mathcal{L}_v\}$. Let $\mathcal{W}^+ = \{W_u \in \mathcal{W} \mid \delta^+_{H-F_u}(W_u) = \emptyset \}$, and let $\mathcal{W}^- = \{W_u  \in \mathcal{W} \mid \delta^-_{H-F_u}(W_u) = \emptyset \}$. Observe that $|\mathcal{W}^+| + |\mathcal{W}^-| \geq |\mathcal{L}_v|$. We now split our analysis into two cases.
\begin{itemize}
    \item \textbf{Case 1: $|\mathcal{W}^+| \geq |\mathcal{L}_v|/2 >  2^{k+3}q$.} By Claim \ref{clm:degeneracy}, there exists a subset $\mathcal{W}_1^+ \subseteq \mathcal{W}^+$ of size $|\mathcal{W}_1^+| \geq |\mathcal{W}^+|/(8q) > 2^{k}$ such that for all distinct $W_{u}, W_{u'} \in \mathcal{W}_1^+$, $u \not \in W_{u'}$ and $u' \not \in W_u$.
    Let $U_1^+ = \{u \in U \mid W_u \in \mathcal{W}_1^+\}$, and note that $|U_1^+| = |\mathcal{W}_1^+| > 2^k$. For each $u \in U_1^+$, we observe the following:
    \begin{enumerate}
        \item $|\delta^+_H(W_u)| \leq k$,
        \item $v$ can reach $u$ in $H-\delta^+_H(W_u)$, since $u, v \in C_u \subseteq W_u$, and
        \item $v$ cannot reach any vertex $u' \in U_1^+ - \{u\}$  in $H-\delta^+_H(W_u)$, since $u' \not \in W_u$ if $u \neq u'$.
    \end{enumerate}
See \Cref{fig:small-scc} for a visualization. We can apply the anti-isolation lemma (\Cref{lem:anti isolation}) with respect to source vertex $v$,  sink vertices $U_1^+$, and edge sets $\{\delta^+_H(W_u) \mid u \in U_1^+\}$ to argue that $|U_1^+| \le 2^k$. However, this contradicts our  earlier claim that $|U_1^+| > 2^{k}$. We conclude that $|\mathcal{L}| \leq 2^{k+3}qn$.

    \item \textbf{Case 2: $|\mathcal{W}^-| \geq |\mathcal{L}_v|/2 >  2^{k+3}q$.} This case is symmetric to Case 1. Specifically, we can obtain a contradiction in this case by considering the reverse graph of $H$ and applying an identical argument as in Case 1.  \qedhere
\end{itemize}
 \end{proof}

We  finish our size analysis of $H$ with the following claim. 

 \begin{claim}
     $|E(H)| = O(4^k qn)$.
 \end{claim}
 \begin{proof}
     $
     |E(H)| \leq 2^{k+1}\cdot |\mathcal{L}| + O(2^kqn)  \leq  2^{k+1}\cdot 2^{k+4}qn + O(2^kqn) = O(4^kqn),
     $ where the first inequality follows from Claim \ref{clm:H_vs_L_bound} and the second inequality follows from Claim \ref{clm:L_bound}.
 \end{proof}

\subsubsection{Applying the directed expander hierarchy}
\label{subsubsec:expanderhierarch}
We will combine our fault-tolerant sourcewise connectivity preservers from Lemma \ref{lem:unbreak_pres} with the existential directed expander hierarchy of Lemma \ref{lem:expander_hierarchy} to complete the construction of our fault-tolerant connectivity preserver claimed in \Cref{thm:existentialtwo}.

\existentialtwo*
\begin{proof}
    Let $G$ be an $n$-vertex directed graph, and let $k$ be a positive integer. 
    Let $\{V_1, \dots, V_{\ell}\}$ be the directed expander hierarchy for $(2k, k)$-unbreakable sets specified in Lemma \ref{lem:expander_hierarchy}. For each $i \in [1, \ell]$, let $V_{\leq i} = V_1 \cup \dots \cup V_i$, and let $C_i^1, \dots, C_i^{j_i} \subseteq V_{\leq i}$ be the SCCs of $G[V_{\leq i}]$, for some $j_i \in [1, n]$. Additionally, for each $i \in [1, \ell]$ and $j \in [1, j_i]$, we let $U_i^j = C_i^j \cap V_i$. 

    For each $i \in [1, \ell]$ and $j \in [1, j_i]$ we build a $k$-FT sourcewise connectivity preserver $H_i^j$ of graph $G[C_i^j]$ with respect to source vertices $U_i^j$ using Lemma \ref{lem:unbreak_pres}. We let $H_i = \cup_{j \in [1, j_i]}H_i^j$, and let $H = \cup_{i \in [1, \ell]} H_i$. We claim that $H$ is a $k$-FT connectivity preserver of $G$ with size $|E(H)| = O(k4^k \cdot n \log n)$. 

    First, we verify that $|E(H)| = O(k4^k \cdot n \log n)$. By Lemma \ref{lem:expander_hierarchy}, set $U_i^j$ is $(2k, k)$-unbreakable in $G[C_i^j]$ for all $i \in [1, \ell]$ and $j \in [1, j_i]$. Then  by Lemma \ref{lem:unbreak_pres}, $|E(H_i^j)| = O(4^kk|C_i^j|)$. Then we can upper bound the number of edges in $H$ as follows:
    $$
    |E(H_i)| \leq \sum_{j \in [1, j_i]} |E(H_i^j)| \leq O(k4^k) \cdot \sum_{[1, j_i]} |C_i^j| = O(k4^kn).
    $$
Since there are at most $\ell = O(\log n)$ levels in our directed expander hierarchy, we conclude that 
$$
|E(H)| \leq \sum_{i \in [1, \ell]} |E(H_i)| \leq O(k4^k n \log n). 
$$

We now verify that $H$ is a $k$-FT connectivity preserver of $G$. Fix a pair of vertices $s, t \in V$ and a fault set $F \subseteq E(G)$ of size $|F| \leq k$ such that $s$ and $t$ are strongly connected in $G-F$. Let $C \subseteq V(G)$ be the SCC in $G-F$ containing $s$ and $t$. Let $i \in [1, \ell]$ be the largest index such that $V_i \cap C \neq \emptyset$, and let $x \in V_i \cap C$.  Since $C \subseteq V_{\leq i}$, set $C$ is contained in some SCC $C_i^j$ in $G[V_{\leq i}]$, where $j \in [1, j_i]$. 
It follows that $x \in  U_i^j$, since $V_i \cap C \subseteq V_i \cap C_i^j = U_i^j$. 
Recall that subgraph $H_i^j$ is a sourcewise connectivity preserver of $G[C_i^j]$ with respect to source vertices $U_i^j$. Then since $x \in U_i^j$ is strongly connected to $s, t \in C_i^j$ in $G[C_i^j] - F$, it follows that $x$ is strongly connected to $s$ and $t$ in $H_i^j - F$. We conclude that $s$ and $t$ are strongly connected in $H_i^j - F \subseteq H-F$, as claimed. \end{proof}

\subsection{Constructing fault-tolerant connectivity  preservers in FPT time}
\label{subsec:fpt}

The goal of this section is to present our FPT-time construction of fault-tolerant connectivity preservers. 

\begin{restatable}{theorem}{fpt}
For every positive integer $k$, every $n$-vertex directed graph admits a $k$-FT connectivity preserver with $O(8^k  n\log^{5/2}n)$ edges that can be computed in $2^{O(k)}\text{\normalfont poly}(n)$ time, with high probability. 
\label{thm:fpt}
\end{restatable}

We use a decremental greedy strategy to construct our $k$-FT connectivity preserver $H$. Initially, we  let $H \gets G$. Then we repeatedly identify an edge $e \in E(H)$ such that $H-e$ is a $k$-FT connectivity preserver of $G$. Once we find such an edge $e$, we can update $H$ so that $H \gets H - e$.  Our algorithm terminates when $|E(H)| = O(8^k n \log^{5/2} n)$. 

Under this framework, the main algorithmic subroutine of our construction is to find an edge in $H$ that is safe to remove. If an edge is not safe to remove from $H$, then we say it is \textit{$k$-fault critical.}

\begin{definition}[$k$-fault critical]
Let $G$ be an $n$-vertex directed graph. We say that an edge $e \in E(G)$ is $k$-fault critical in $G$ if there exists a pair of vertices $s, t \in V(G)$ and a fault set $F \subseteq E(G)$ of size $|F| \leq k$ such that $s$ and $t$ are strongly connected in $G-F$, but $s$ and $t$ are not strongly connected in $(G-e)-F$. 
\end{definition}

We are able to adapt the tools from our proof of Theorem \ref{thm:FT} to efficiently compute a small set of edges $E' \subseteq E(G)$ that contains all $k$-fault critical edges in $G$.

\begin{restatable}{lemma}{edgeremoval}
Let $G$ be an $n$-vertex directed graph. For every positive integer $k$, there exists an edge set $E' \subseteq E(G)$ of size $|E'| = O(8^k n \log^{5/2} n)$ that contains every edge in $G$ that is $k$-fault critical.  Moreover,  $E'$ can be computed with high probability in time $2^{O(k)}\cdot \text{\normalfont poly}(n)$.  
\label{lem:dec_fpt}
\end{restatable}

Lemma \ref{lem:dec_fpt} directly implies  Theorem \ref{thm:fpt}.

\begin{proof}[Proof of Theorem \ref{thm:fpt}]
We initialize  our $k$-FT connectivity preserver $H$ of $G$  to be $H \gets G$. If $|E(H)| = O(8^k n \log^{5/2} n)$, then we are done. Otherwise, we can find  an edge $e \in E(H)$ that is not $k$-fault critical in $H$, with high probability in time $2^{O(k)} \cdot \text{poly}(n)$, by Lemma \ref{lem:dec_fpt}. Since $e \in E(H)$ is not $k$-fault critical in $H$, subgraph $H-e$ is a $k$-fault connectivity preserver of $H$. Moreover, since $H$ is a $k$-FT connectivity preserver of $G$, it follows that $H-e$ is a $k$-FT connectivity preserver of $G$ as well. 
We update $H$ so that $H \gets H-e$.

In general, given a subgraph $H$ of $G$ that is a $k$-FT connectivity preserver of $G$, either $|E(H)| = O(8^k n \log^{5/2} n)$ or  we can apply \Cref{lem:dec_fpt} on $H$ to find an edge $e \in E(H)$ such that $H-e$ is a $k$-FT connectivity preserver of $G$ as well.
Our algorithm terminates and outputs $H$ when $|E(H)| = O(8^k n \log^{5/2} n)$. 
This algorithm outputs a $k$-FT connectivity preserver $H$ of $G$ of size $|E(H)| = O(8^k n \log^{5/2} n)$ with high probability in time $2^{O(k)} \cdot \text{poly}(n)$. 
% such that $H-e$ is a $k$-FT connectivity preserver of $H$ in time $2^{O(k)}\cdot \text{\normalfont poly}(n)$, with high probability. Then we let $H \gets H-e$, and we again apply Lemma \ref{lem:dec_fpt} on the updated $H$. We terminate this process and output $H$ when $|E(H)|= O(k4^kn\log^{5/2}n)$ and Lemma \ref{lem:dec_fpt} no longer guarantees that we can find an edge we can safely remove from $H$.  
% Finally, we observe that if $H$ is a $k$-FT connectivity preserver of $G$ and $H-e$ is a $k$-FT connectivity preserver of $H$, then $H-e$ is a $k$-FT connectivity preserver of $G$ as well. Then  this algorithm computes a $k$-FT connectivity preserver $H$ of $G$ with size $|E(H)|= O(k4^kn\log^{5/2}n)$ with high probability in time $2^{O(k)}\cdot \text{\normalfont poly}(n)$.
\end{proof}

For the remainder of this section, we focus on proving \Cref{lem:dec_fpt}. 

% The proof of Theorem \ref{thm:fpt} generally follows a similar framework as the proof of Theorem \ref{thm:FT}. However, we will need the following 

% In Section \ref{subsubsec:fpt_swise} we will show how to construct fault-tolerant sourcewise connectivity preservers in FPT time, and in Section \ref{subsubsec:fpt_hierarchy} we will combine our sourcewise preservers with an algorithmically efficient directed expander hierarchy. 

\subsubsection{$k$-fault critical edges for sourcewise connectivity}
\label{subsubsec:fpt_swise}

Similar to the proof of Theorem \ref{thm:existentialtwo}, we will begin our proof of Lemma \ref{lem:dec_fpt} by first considering fault-tolerant \textit{sourcewise} connectivity. We will need the following definition.

\begin{definition}[$k$-fault critical with respect to $P$]
    Let $G = (V, E)$ be an $n$-vertex directed graph. We say that an edge $e \in E$ is $k$-fault critical in $G$ with respect to vertex pairs $P \subseteq V \times V$ if there exists $(s, t) \in P$ and a fault set $F \subseteq E$ of size $|F| \leq k$ such that $s$ and $t$ are strongly connected in $G-F$, but $s$ and $t$ are not strongly connected in $(G-e)-F$. 
\end{definition}

We will prove the following algorithmic lemma about $k$-fault critical edges. 

% \noindent Our FPT-time construction of fault-tolerant sourcewise connectivity preservers can be summarized with the following lemma. 

\begin{lemma}
    Let $G$ be an $n$-vertex, $m$-edge directed graph, and let  $U \subseteq V(G)$ be a $(q, 2^k)$-unbreakable set in $G$. There exists an edge set $E' \subseteq E(G)$ of size $|E'| = O(2^kq^2n + 2^kq n \log n)$ that contains every edge $e \in E(G)$ that is $k$-fault critical with respect to $U \times V(G)$.  
    Moreover,  $E'$ can be computed with high probability in time $O(q^22^kmn + qn \log(n) \cdot (k2^km+k^24^k))$. 
\label{lem:fpt_swise}
\end{lemma}

We will use Lemma \ref{lem:fpt_swise}, combined with a directed expander hierarchy, to prove Lemma \ref{lem:dec_fpt} and complete the proof of Theorem \ref{thm:fpt}. We first present the construction of $E'$, and then we analyze the time complexity and correctness of our construction.

\paragraph{Construction of $E'$.}

Let $G = (V, E)$ be an $n$-vertex directed graph, and let $U \subseteq V$ be a $(q, 2^k)$-unbreakable set in $G$. Let $U' \subseteq U$ be an arbitrary subset of $U$ of size $|U'| = \min(|U|, 5q^2)$. Let $J \subseteq G$ be the subgraph of $G$ obtained by unioning $|U'|$ many $k$-FT single-source connectivity preservers of $G$, each rooted at a distinct vertex in $U'$. Formally, let $J = \bigcup_{u \in U'} \ssc(G, u, k)$. Subgraph $E(J)$ will be part of our final edge set $E'$. 

We define a collection of $\lambda = 50 \log n$ subsets of $U$ as follows. For $j \in [1, \lambda]$, let $Q_j \subseteq U$ be an element of $\binom{U}{q}$ sampled uniformly at random. We will use the collection of sets $\{Q_j\}_{j \in [1, \lambda]}$ to construct our final edge set $E'$.

Let $V = \{v_1, \dots, v_n\}$. For each $i \in [1, n]$, we will define a set $U_i$ and a collection of edges $E_i \subseteq E$ as follows. For each $j \in [1, \lambda]$, if $v_i \not \in Q_j$, then let $(U_i^j, V-U_i^j) = \is(G, v_i, Q_j, k, +)$ and let $(W_i^j, V-W_i^j) = \is(G, v_i, Q_j, k, -)$.
Otherwise, if $v_i \in Q_j$, then let $U_i^j = W_i^j = \emptyset$. 
Let $S_i = \bigcup_{j \in [1, \lambda]} U_i^j \cup W_i^j$, and let $U_i = S_i \cap U$.

With set $U_i \subseteq U$ defined, we are ready to define edge set $E_{i} \subseteq E$.  For each $u \in U_i$,  let $(S_i^u, V-S_i^u) = \is(G, u, v_i, k+1, +)$ and let $(T_i^u, V-T_i^u) = \is(G, u, v_i, k+1, -)$. 
We define edge set $E_i$ to be $E_i = \cup_{u \in U_i } \delta^+_{G}(S_i^u) \cup \delta^-_{G}(T_i^u)$.
% If there exists a vertex $u \in U_i$ such that $e \in \delta^+_{G}(S_1^u) \cup \delta^-_{G}(S_2^u)$, then we let $H_{i+1} \gets H_i$. Otherwise, we let $H_{i+1} \gets H_i - e_i$. 

We finish our construction of $E'$ by letting $E' =  E(J) \cup_{i \in [1, n]} E_i$. Our construction is summarized in Algorithm \ref{alg:cap}.

\begin{algorithm}[ht!]
\caption{Construction of edge set $E'$ from Lemma \ref{lem:fpt_swise}} \label{alg:cap}
\begin{algorithmic}
\Require  Directed graph $G = (V, E)$ and $(q, 2^k)$-unbreakable set $U \subseteq V$.

% \If{$|U| \leq 2^{k+2}$}

\State $J \gets (V, \emptyset)$
\State Let $U' \subseteq U$ be a subset of $U$ of size $|U'|=\min(5q^2, |U|)$. 
\State $J \gets \cup_{u \in U'}  \ssc(G, u, k)$
\State Let $\lambda = 50\log n$
\For{$j \in [1, \lambda]$}
\State Let $Q_j \subseteq U$ be an element of $\binom{U}{q}$ sampled uniformly at random. 
\EndFor
\State Let $V = \{v_1, \dots, v_n\}$. 
\For{$i \in [1, n]$} 
\For{$j \in [1, \lambda]$}
\If{$v_i \not \in Q_j$} 
\State $(U_i^j, V - U_i^j) \gets \is(G, v_i, Q_j, k, +)$.
\State $(W_i^j, V - W_i^j) \gets \is(G, v_i, Q_j, k, -)$.
\Else{}
\State $U_i^j \gets \emptyset$, $W_i^j \gets \emptyset$
\EndIf

\EndFor 
\State $S_i \gets \bigcup_{j \in [1, \lambda]} U_i^j \cup W_i^j$
\State $U_i \gets S_i \cap U$

\For{$u \in U_i$}
\State $(S_i^u, V-S_i^u) \gets \is(G, u, v_i, k+1,  +)$.
\State $(T_i^u, V-T_i^u) \gets \is(G, u, v_i, k+1, -)$.
\EndFor
\State $E_i \gets \cup_{u \in U_i} \delta_G^+(S_i^u) \cup \delta_G^-(T_i^u)$
\EndFor
\State $E' \gets E(J) \cup_{i \in [1, n]} E_{i}$
\end{algorithmic}
\label{alg:swise_scc}
\end{algorithm}

\paragraph{Time Complexity.} We now verify that Algorithm \ref{alg:cap} has the time complexity claimed in Lemma \ref{lem:fpt_swise}. We will first need the following result bounding the size of $U_i$. 

\begin{claim}
For each $i \in [1, n]$, $|U_i| \leq 2\lambda q$.
\label{clm:small_U_i}
\end{claim}
\begin{proof}
    We claim that for each $j \in [1, \lambda]$, $|U_i^j \cap U| < q$ and $|W_i^j \cap U| < q$. Note that by Lemma \ref{lem:important_separators}, $|\delta^+_{G}(U_i^j)| \leq 2^k$ and $|\delta^-_{G}(W_i^j)| \leq 2^k$. Additionally, $Q_j \subseteq V - U_i^j$ and $Q_j \subseteq V - W_i^j$, so $|(V-U_i^j) \cap U| \geq q$ and $|(V-W_i^j) \cap U| \geq q$. Then if $|U_i^j \cap U| \geq q$ or $|W_i^j \cap U| \geq q$, it would imply that $U$ is $(q, 2^k)$-breakable, contradicting our assumption that $U$ is $(q, 2^k)$-unbreakable in $G$. 
    \end{proof}

Using Claim \ref{clm:small_U_i}, we can now bound the time complexity of our construction of $H$. 

\begin{claim}
    Algorithm \ref{alg:cap} runs in time $O(q^22^kmn + qn\log(n) \cdot (k2^km+k^24^k))$.
\end{claim}
\begin{proof}
    We can compute a $k$-FT single-source (symmetric) connectivity preserver in time $O(2^kmn)$ by Theorem 2 of \cite{baswana2018fault}. Consequently, subgraph $J$ can be computed in $O(q^22^kmn)$ time. 
    Sets $Q_1, \dots, Q_{\lambda}$ can be computed in $O(\lambda qn)$ time. 

    For each vertex $v_i \in V(G)$, we make $O(\lambda + |U_i|)$ calls to $\is(G, \cdot, \cdot, k+1, \cdot)$. Each call can be completed in time $O(k2^km + k^24^k)$ by Lemma \ref{lem:important_separators}. By Claim \ref{clm:small_U_i}, $|U_i| \leq 2\lambda q$. 
    Then for each vertex $v_i \in V(G)$, we spend time $O((\lambda + |U_i|) \cdot (k2^km + k^24^k)) = O(\lambda q \cdot (k2^km + k^24^k))$. Then the total time complexity contributed by the for loop in our construction is $O(\lambda qn(k2^km + k^24^k))$. 

    We conclude that our construction of $E'$ runs in time
    $
     O(q^22^kmn + q n \log(n) \cdot (k2^km + k^24^k)).
    $
\end{proof}

\paragraph{Correctness.}
We now verify that our edge set $E'$ contains all edges in $E(G)$ that cannot be safely removed from $G$. This property of $E'$ is summarized in the following lemma. 

\begin{restatable}{lemma}{fptcorrectness} \label{lem:fpt_correctness}
    For all $e \in E(G)$, if $e$ is $k$-fault critical with respect to $U \times V(G)$, then $e \in E'$.  
\end{restatable}

Fix an edge $e = (v_i, v_j) \in E(G)$. We will show that if $e$ is $k$-fault critical with respect to $U \times V(G)$, then $e \in E'$. We start with the following structural claim.

\begin{claim} \label{clm:fpt_crit_struct}
    If $e = (v_i, v_j) \in E(G)$ is $k$-fault critical with respect to $U \times V(G)$, then there exists a vertex $u_e \in U$ and a fault set $F_e \subseteq E(G)$ of size $|F_e| \leq k$ such that vertices $v_i$, $v_j$, and $u_e$ are strongly connected in $G-F_e$, and $G-F_e$ contains either 
    \begin{itemize}
        \item an out-reachable $(u_e, v_j)$-cut $(S, V-S)$ such that $\delta^+_{G-F_e}(S) = \{e\}$, or
        \item an in-reachable $(u_e, v_i)$-cut $(S, V-S)$ such that $\delta^-_{G-F_e}(S) = \{e\}$.
    \end{itemize}
\end{claim}
\begin{proof}
    Claim \ref{clm:fpt_crit_struct} follows from an argument identical to the proof of Claim \ref{claim:decremental_struct}.
\end{proof}

 Let $X_e$ be the SCC in $G-F_e$ containing vertices $v_i, v_j, u_e$. 

    \begin{claim} \label{clm:fpt_wlog}
     In the proof of Lemma \ref{lem:fpt_correctness}, we may assume  that $|X_e \cap U| \leq 2q$ and $|U| > 5q^2$, without loss of generality.
    \end{claim}
    \begin{proof}
            By Claim \ref{clm:large_scc}, if $|X_e \cap U| > 2q$, then $v_i$, $v_j$, and $u_e$ are strongly connected in $J-F_e$. However, $v_i$, $v_j$, and $u_e$ are not all strongly connected in $(J-e)-F_e$, so $e \in E(J) \subseteq E(H)$. Then in the proof of Lemma \ref{lem:fpt_correctness}, we may assume without loss of generality that $|X_e \cap U| \leq 2q$. Additionally, if $|U| \leq 5q^2$, then $u_e \in U'$, so $v_i$, $v_j$, and $u_e$ are strongly connected in $J-F_e$. Again this implies that $e \in E(J) \subseteq E(H)$. Then we may assume without loss of generality that $|U| > 5q^2$. 
    \end{proof}

  We may assume without loss of generality that $|X_e \cap U| \leq 2q$ and $|U| > 5q^2$ by Claim \ref{clm:fpt_wlog}. Then by Claim \ref{clm:giant_component}, there is an SCC $Z_e \neq X_e$ in $G-F_e$ such that $|Z_e \cap U| \geq |U| - 2q$. Then there is an $(X_e, Z_e)$-cut $(Y_e, V-Y_e)$ in $G-F_e$ such that  either $\delta^+_{G-F_e}(Y_e) = \emptyset$ or $\delta^-_{G-F_e}(Y_e) = \emptyset$, so $|\delta^+_G(Y_e)| \leq k$ or $|\delta^-_G(Y_e)| \leq k$.  Moreover, $|Y_e \cap U| \leq 2q$, since $|Z_e \cap U| \geq |U| - 2q$.

    % Additionally, if $2q < |X \cap U| < 3q$,  then by Claim \ref{clm:giant_component} $|U| \leq |X \cap U| + 2q \leq 5q$. In this case, we conclude that $u \in U'$, so $s_i$, $t_i$, and $u$ are strongly connected in $J-F$, by construction. Consequently, we may assume without loss of generality that $|X \cap U| \leq 2q$ 

   %  We claim that there exists a  cut $(Y, V-Y)$ in $G$ such that 
   %  \begin{enumerate}
   %      \item $X \subseteq Y$,
   %      \item $|Y \cap U| \leq 5q$, and
   %      \item $\delta^+_G(Y) \leq k$ or $|\delta^-_G(Y)| \leq k$.
   %  \end{enumerate}
   % By Claim \ref{clm:giant_component}, if $X$ is a giant component of $U$ in $G-F$, then $|U| \leq |X \cap U| + 2q \leq 5q$, so we can let $Y = V$. Otherwise, there exists a giant SCC $Z \neq X$  in $G-F$  such that $|Z \cap U| \geq |U| - 2q$, by Claim \ref{clm:giant_component}. Then there is an $(X, Z)$-cut $(Y, V-Y)$ in $G-F$ such that either $\delta^+_{G-F}(Z) = \emptyset$ or $\delta^-_{G-F}(Z) = \emptyset$, so $|\delta^+_G(Z)| \leq k$ or $|\delta^-_G(Z)| \leq k$.  Moreover, $|Y \cap U| \leq 2q$, since $|Z \cap U| \geq |U| - 2q$. We conclude that cut $(Y, V-Y)$ satisfies the claim. 

    \begin{claim}
        With high probability, at least one of the sets $Q_j \subseteq U$, $j \in [1, \lambda]$, sampled during the construction of $H$ satisfies $Q_j \cap Y_e = \emptyset$. 
        \label{clm:fpt_whp}
    \end{claim}
    \begin{proof}
 We first observe that for a fixed index $j \in [1, \lambda]$, $Q_j \cap Y_e = \emptyset$ with constant probability:
   $$
   \Pr[Q_j \cap Y_e \neq \emptyset] \leq \sum_{u \in Y_e \cap U} \Pr[u \in Q_j] \leq |Y_e \cap U| \cdot \frac{\binom{|U|}{q-1}}{\binom{|U|}{q}} \leq 2q \cdot \frac{q}{|U|-q+1} \leq \frac{4q^2}{|U|} \leq \frac{4}{5}.
   $$
   Then the probability that none of our sampled sets $Q_j \subseteq U$, $j \in [1, \lambda]$ are disjoint from $Y$ is at most
   $$ 
   \left(\frac{4}{5}\right)^{\lambda}  \leq e^{- \lambda / 5} = e^{-10 \log n} \leq n^{-10}. 
   $$
   Then with high probability, there exists a set $Q_j$, $j \in [1, \lambda]$, such that $Q_j \cap Y_e = \emptyset$. 
    \end{proof}

    With the above discussion in mind, we can now prove the following claim about edges in $E(G)$ that are $k$-fault critical with respect to $U \times V(G)$.

   \begin{claim} \label{clm:u_e_in_U_i}
    If edge $e = (v_i, v_j) \in E(G)$ is $k$-fault critical with respect to vertex pair $(u_e, v) \in U \times V(G)$, then $u_e \in U_i \cap U_j$.
   \end{claim}
   \begin{proof}
        By the earlier discussion, there exists an out-minimal $(X_e, Z_e)$-cut $(Y_e, V(G)-Y_e)$ in $G-F_e$ such that $|Y_e \cap U| \leq 2q$, and either $\delta^+_{G-F_e}(Y_e) = \emptyset$ or $\delta^-_{G-F_e}(Y_e) = \emptyset$. Moreover, by Claim \ref{clm:fpt_whp}, there exists a $\ell \in [1, \lambda]$ such that $Q_{\ell} \subseteq V(G) - Y_e$, with high probability.  

        We now split our analysis into two cases based on whether $\delta^+_{G-F_e}(Y_e) = \emptyset$ or $\delta^-_{G-F_e}(Y_e) = \emptyset$. 
        \begin{itemize}
            \item If  $\delta^+_{G-F_e}(Y_e) = \emptyset$, then  $\delta^+_{G}(Y_e) \subseteq F_e$, so $|\delta^+_G(Y_e)| \leq k$. We observe that $(Y_e, V-Y_e)$ is a $(v_i, Q_{\ell})$-cut as well.   Since  $v_i$ and $u_e$ lie in the same SCC $X_e \subseteq Y_e$ in $G-F_e$, we conclude  that there exists an \textit{out-reachable} $(v_i, Q_{\ell})$-cut $(Y_e^*, V-Y_e^*)$ in $G$ of size $|\delta^+_G(Y_e^*)| \leq k$ such that $u_e \in X_e \subseteq Y_e^*$. 
            Then by Lemma \ref{lem:important_separators}, $Y_e^* \subseteq U_i^{\ell}$, so $u_e \in U_i^{\ell} \cap U \subseteq  U_i$. 
            By an identical argument, $u_e \in U_j$.
            \item  If  $\delta^-_{G-F_e}(Y_e) = \emptyset$, then  $\delta^-_{G}(Y_e) \subseteq F_e$, so $|\delta^-_G(Y_e)| \leq k$. We observe that $(Y_e, V-Y_e)$ is a $(v_i, Q_{\ell})$-cut as well.   Since  $v_i$ and $u_e$ lie in the same SCC $X_e \subseteq Y_e$ in $G-F_e$, we conclude that there exists an \textit{in-reachable} $(v_i, Q_{\ell})$-cut $(Y_e^*, V-Y_e^*)$ in $G$ of size $|\delta^-_G(Y_e^*)| \leq k$ such that $u_e \in X_e \subseteq Y_e^*$. Then by Lemma \ref{lem:important_separators}, $Y_e^* \subseteq W_i^{\ell}$, so $u_e \in W_i^{\ell} \cap U \subseteq  U_i$. 
            By an identical argument, $u_e \in U_j$.
        \end{itemize}
        We conclude that $u_e \in U_i \cap U_j$, as claimed.
   \end{proof}

    We are ready to finish the proof of  Lemma \ref{lem:fpt_correctness}.

   \fptcorrectness*
   \begin{proof}[Proof of Lemma \ref{lem:fpt_correctness}]
       If $e = (v_i, v_j)$ is $k$-fault critical with respect to $U \times V(G)$, then by Claim \ref{clm:fpt_crit_struct},  there exists a vertex $u_e \in U$ and a fault set $F_e \subseteq E(G)$ of size $|F_e| \leq k$ such that vertices $v_i$, $v_j$, and $u_e$ are strongly connected in $G-F_e$, and $G-F_e$ contains either 
    \begin{itemize}
        \item an out-reachable $(u_e, v_j)$-cut $(S, V-S)$ such that $\delta^+_{G-F_e}(S) = \{e\}$, or
        \item an in-reachable $(u_e, v_i)$-cut $(S, V-S)$ such that $\delta^-_{G-F_e}(S) = \{e\}$.
    \end{itemize}
    Moreover, by Claim \ref{clm:u_e_in_U_i}, $u_e \in U_i \cap U_j$. Now we split our analysis into two cases.
    \begin{itemize}
        \item \textbf{Case 1:} $G - F_e$ contains an out-reachable $(u_e, v_j)$-cut $(S, V-S)$ such that $\delta^+_{G-F_e}(S) = \{e\}$. Then  $(S, V-S)$ is an   out-reachable $(u_e, v_j)$-cut in $G$ of size $|\delta^+_G(S)| \leq k+1$. 
        Recall that in our construction of set $E_j$ in Algorithm \ref{alg:cap}, $(S_j^{u_e}, V-S_j^{u_e}) = \is(G, u_e, v_j, k+1, +)$.         
        We conclude that $S \subseteq S_j^{u_e}$ by Lemma \ref{lem:important_separators}, so $e \in \delta^+_G(S_j^{u_e}) \subseteq E_j \subseteq E'$.
        \item \textbf{Case 2:} $G - F_e$ contains an in-reachable $(u_e, v_i)$-cut $(S, V-S)$ such that $\delta^-_{G-F_e}(S) = \{e\}$. Then $(S, V-S)$ is an in-reachable $(u_e, v_i)$-cut in $G$ of size $|\delta^-_G(S)| \leq k+1$.
        Recall that in our construction of set $E_i$ in Algorithm \ref{alg:cap},  $(T_i^{u_e}, V-T_i^{u_e}) = \is(G, u_e, v_i, k+1, -)$. 
        We conclude that $S \subseteq T_i^{u_e}$ by Lemma \ref{lem:important_separators}, so $e \in \delta^-_G(T^{u_e}_i) \subseteq E_i \subseteq E'$.   
    \end{itemize}
    We conclude that $e \in E'$. 
   \end{proof}

\paragraph{Size Analysis.}
We now  conclude the proof of Lemma \ref{lem:fpt_swise} by bounding $|E'|$.

\begin{claim}
    $|E'| = O(2^kq^2n+2^kqn\log n)$.
\end{claim}
\begin{proof}
    Since $J$ is the union of $|U'| = O(q^2)$ many $k$-FT single-source connectivity preservers, $|E(J)| = O(2^kq^2n)$ by Theorem 2 of \cite{baswana2018fault}. 

    We conclude that
    \begin{align*}
        |E'| & = |E(J)| + \sum_{i \in [1, n]} |E_i| \\
        & \leq O(2^kq^2n) + \sum_{i \in [1, n]} \sum_{u \in U_i}\left(|\delta^+_G(S_i^u)| + |\delta^-_G(T_i^u)|\right) \\
        & \leq O(2^kq^2n) + \sum_{i \in [1, n]} \sum_{u \in U_i} 2^{k+1} & \text{by Lemma \ref{lem:important_separators}} \\
          & \leq O(2^kq^2n) +  2^{k+1} \cdot 2\lambda q \cdot n & \text{by Claim \ref{clm:small_U_i}} \\
          & \leq O(2^kq^2n+2^kqn\log n),
    \end{align*}
    as claimed.
\end{proof}

\subsubsection{Applying the directed expander hierarchy}

We will combine  Lemma \ref{lem:fpt_swise} with the polynomial directed expander hierarchy construction of Lemma \ref{lem:expander_hierarchy} to complete the proof of Lemma \ref{lem:dec_fpt}.

\edgeremoval*
\begin{proof}
 Let $G$ be an $n$-vertex directed graph, and let $k$ be a positive integer. Let $\{V_1, \dots, V_{\ell}\}$ be the directed expander hierarchy for $(\Theta(2^k \sqrt{\log n}), 2^k)$-unbreakable sets specified in Lemma \ref{lem:expander_hierarchy}. Since $q = \Omega(2^k \sqrt{\log n})$, this expander hierarchy can be computed in polynomial time. For each $i \in [1, \ell]$, let $V_{\leq i} = V_1 \cup \dots \cup V_i$, and let $S_i^1, \dots, S_i^{j_i} \subseteq V_{\leq i}$ be the SCCs of $G[V_{\leq i}]$, for some $j_i \in [1, n]$. Additionally, for each $i \in [1, \ell]$ and $j \in [1, j_i]$, we let $U_i^j = S_i^j \cap V_i$. 

    For each $i \in [1, \ell]$ and $j \in [1, j_i]$ we apply Lemma \ref{lem:fpt_swise} on graph $G[S_i^j]$ and $(\Theta(2^k \sqrt{\log n}), 2^k)$-unbreakable set $U_i^j$. Then in $2^{O(k)} \cdot \text{poly}(n)$ time, with high probability, we can compute an edge set $E_i^j \subseteq G[S_i^j]$ that contains every edge in $G[S_i^j]$ that is  $k$-fault critical in $G[S_i^j]$ with respect to vertex pairs $U_i^j \times S_i^j$.

    We let $E_i = \cup_{j \in [1, j_i]}E_i^j$, and let $E' = \cup_{i \in [1, \ell]} E_i$. Set $E'$ can be computed in $2^{O(k)} \cdot \text{poly}(n)$ time.     
    We claim that, with high probability, edge set $E'$ contains every edge in $G$ that is $k$-fault critical, and $|E'| = O(8^k n \log^{5/2} n)$. 

    First, we verify that $|E'| = O(8^k n \log^{5/2} n)$. By Lemma \ref{lem:expander_hierarchy}, set $U_i^j$ is $(\Theta(2^k \sqrt{\log n}), 2^k)$-unbreakable in $G[S_i^j]$ for all $i \in [1, \ell]$ and $j \in [1, j_i]$. Then  by Lemma \ref{lem:fpt_swise}, $|E_i^j| = O(8^k \log^{3/2} n |S_i^j|)$. Then we can upper bound the number of edges in $H$ as follows:
    $$
    |E_i| \leq \sum_{j \in [1, j_i]} |E_i^j| \leq O(8^k \log^{3/2} n) \cdot \sum_{[1, j_i]} |S_i^j| = O(8^k n \log^{3/2}n).
    $$
Since there are at most $\ell = O(\log n)$ levels in our directed expander hierarchy, we conclude that 
$$
|E'| \leq \sum_{i \in [1, \ell]} |E_i| \leq O(8^k n \log^{5/2} n). 
$$

We now verify that $E'$ contains every edge in $G$ that is $k$-fault critical, with high probability. Fix an edge $e \in E(G)$ that is $k$-fault critical in $G$. Then there exists a pair of vertices $s, t \in V$ and a fault set $F \subseteq E(G)$ of size $|F| \leq k$ such that  $s$ and $t$ are strongly connected in $G-F$, but not strongly connected in $(G-e)-F$. Let $C \subseteq V(G)$ be the SCC in $G-F$ containing $s$ and $t$. Let $i \in [1, \ell]$ be the largest index such that $V_i \cap C \neq \emptyset$, and let $x \in V_i \cap C$.  Since $C \subseteq V_{\leq i}$, set $C$ is contained in some SCC $S_i^j$ in $G[V_{\leq i}]$, where $j \in [1, j_i]$. Then  $x$ is strongly connected to vertices $s$ and $t$ in $G[S_i^j] - F$. However, $x$ cannot be strongly connected to both $s$ and $t$ in $(G[S_i^j] - e) - F$, since $s$ and $t$ are not strongly connected in $(G-e)-F$. We conclude that edge $e$ is $k$-fault critical in $G[S^j_i]$ with respect to vertex pairs $\{x\} \times \{s, t\}$. Recall that edge set $E_i^j$ contains every edge that is $k$-fault critical in $G[S^j_i]$ with respect to vertex pairs $U_i^j \times S_i^j$, with high probability. Since $x \in C \cap  V_i \subseteq S_i^j \cap V_i = U_i^j$ and  $s, t \in C \subseteq S_i^j$, we conclude that edge $e \in E_i^j \subseteq E'$ with high probability. 

Then by the union bound, $E'$ contains every edge in $G$ that is $k$-fault critical, with high probability. 
\end{proof}

\section{$k$-Connectivity Preservers}

\label{sec:kconn}

In this section, we present our upper bound on $k$-connectivity preservers in \Cref{thm:lambda_conn}.

\lambdaconn*

Recall that a subgraph $H \subseteq G$ is a $k$-connectivity preserver of $G$ if for all $s, t \in V(G)$, $\lambda_H^k(s, t) = \lambda_G^k(s,t)$, where $\lambda_G^k(s,t) = \min(\lambda_G(s,t),k)$ is the $k$-bounded connectivity between $s,t$ in $G$. 

\begin{lemma}[transitivity]
    For three graphs $H' \subseteq H \subseteq G$, if $H'$ is a $k$-connectivity preserver of $H$, and $H$ is a $k$-connectivity preserver of $G$, then $H'$ is a $k$-connectivity preserver of $G$.
    \label{lem:transitivity}
\end{lemma}

\begin{proof}
    For each pair of vertices $s, t \in V(G)$, by the definition of $k$-connectivity preservers we have $\lambda_{H'}^k(s, t) = \lambda_H^k(s, t) = \lambda_G^k(s, t)$. Therefore, $H'$ is a $k$-connectivity preserver of $G$.
\end{proof}

Our proof of \Cref{thm:lambda_conn} requires the following definition. 

\begin{definition}[$k$-critical graphs] \label{def:k_critical}
    We say that a graph $H$ is \textit{$k$-critical} if no strict edge subgraph $H' \subset H$ of $H$ is a $k$-connectivity preserver of $H$.
\end{definition}

We prove Theorem \ref{thm:lambda_conn} by establishing the following upper bound on the number of edges in an $n$-vertex, $k$-critical directed graph.
\begin{theorem} \label{thm:k_critical}
    Let $H$ be an $n$-vertex, $k$-critical directed graph. Then $|E(H)| = O(k^{1/2}n^{3/2})$.  
\end{theorem}

We quickly verify that Theorem \ref{thm:k_critical} implies Theorem \ref{thm:lambda_conn}.

\begin{proof}[Proof of Theorem \ref{thm:lambda_conn}]
Let $G = (V, E)$ be an $n$-vertex directed graph, and let $H \subseteq G$ be an edge-minimal $k$-connectivity preserver of $G$. By \Cref{lem:transitivity} and the minimality of $H$, there does not exist a subgraph $H' \subsetneq H$ that is a $k$-connectivity preserver of $H$. Then $H$ is an $n$-vertex, $k$-critical graph. By \Cref{thm:k_critical}, $|E(H)| = O(k^{1/2}n^{3/2})$. 
\end{proof}

We will prove \Cref{thm:k_critical} using a two-level decomposition with regard to $(q,k)$-unbreakability.  Let $H$ be an $n$-vertex, $k$-critical directed graph. Let $q \in [1, n]$ be a parameter of the decomposition that we will optimize later. We maintain a partition $\mathcal{S}$ of vertex set $V$ and a collection of cuts $\mathcal{C}$ as follows. Initially, we let $\mathcal{S} = \{V\}$ and $\mathcal{C} = \emptyset$. While there exists a set $S \in \mathcal{S}$ and a cut $(L, R)$ in $H$ of size at most $\min(|\delta^+_H(L)|, |\delta^-_H(L)|) \leq k$ such that $|S \cap L| \geq q$ and $|S \cap R| \geq q$, we remove set $S$ from $\mathcal{S}$ and we add sets $S \cap L$ and $S \cap R$ to $\mathcal{S}$. Additionally, we add cut $(L, R)$ to collection $\mathcal{C}$. Note that $\mathcal{S}$ remains a partition of $V$ after this operation. We repeat this process until no such set $S \in \mathcal{S}$ and cut $(L, R)$ exists.

\begin{claim} \label{clm:termcondition}
 This process must terminate. After the above process terminates, every set $U \in \mathcal{S}$ is $(q, k)$-unbreakable in $H$. Moreover, $|\mathcal{C}| \leq n/q$. 
\end{claim}
\begin{proof}
     This process must terminate, since  $|\mathcal{S}|$ increases by one after each iteration. By the terminating condition of this procedure,  every set $U \in \mathcal{S}$ is $(q,k)$-unbreakable in $H$. Additionally, every set $U \in \mathcal{S}$ has size at least $|U| \geq q$. We conclude that  $|\mathcal{S}| \leq n/q$, and therefore $|\mathcal{C}| \leq |\mathcal{S}| \leq n/q$. 
\end{proof}

We will bound the number of edges inside each set in $\mathcal{S}$ in \Cref{lemma: k-conn-intra}, and  we will bound the number of edges crossing sets in $\mathcal{S}$ using \Cref{lemma: k-conn-inter}.

\begin{lemma}
    \label{lemma: k-conn-intra}
    Consider a $k$-critical directed graph $H$. For every subset $U \subseteq V(H)$, if $U$ is $(q,k)$-unbreakable in $H$, then the induced subgraph $H[U]$ has at most $O((q+k)|U|)$ edges.
\end{lemma}

\begin{proof}
Since $H$ is $k$-critical, for each edge $e \in E(H)$, graph $H-e$ is not a $k$-connectivity preserver of $H$. Therefore, there exists a pair of nodes $s, t \in V(H)$ such that $\lambda_{H-e}^k(s,t) < \lambda_H^k(s,t)$. Thus, for each edge $e \in E(H)$, there exists a minimum $(s, t)$-cut $(L_e, R_e)$ of size  $|\delta^+_H(L_e)| \leq k$ such that  $e \in \delta^+_H(L_e)$.

Now we will show that for each edge $e = (u, v) \in H[U]$, $$\min(\text{\normalfont outdeg}_{H[U]}(u), \text{\normalfont indeg}_{H[U]}(v)) \leq q+k.$$
Suppose towards contradiction that  $\min(\text{\normalfont outdeg}_{H[U]}(u), \text{\normalfont indeg}_{H[U]}(v)) > q+k.$ Then since $(L_e, R_e)$ is a cut of size at most $|\delta^+_H(L_e)| \leq k$ and $e \in \delta^+_H(L_e)$, this implies that $$|L_e \cap S| \geq |N^+_{H[U]}(u) \cap L_e| \geq |N^+_{H[U]}(u)| - k > q$$ and  $$|R_e \cap U| \geq |N^-_{H[U]}(v) \cap R_e| \geq |N^-_{H[U]}(v)| - k > q.$$ However, this contradicts our assumption that set $U$ is $(q, k)$-unbreakable in $H$. We conclude that for each edge $(u, v) \in H[U]$,   $\min(\text{\normalfont outdeg}_{H[U]}(u), \text{\normalfont indeg}_{H[U]}(v)) \leq q+k,$ as claimed.

Using this inequality, it becomes straightforward to prove that the total number of edges in $H[U]$ is at most $|E(H[U])| \leq 2(q+k)|U|$. For each edge $(u, v) \in E(H[U])$, if $\text{\normalfont outdeg}_{H[U]}(u) \leq q+k$, then we charge edge $(u, v)$ to vertex $u \in U$. Otherwise, if  $\text{\normalfont indeg}_{H[U]}(v) \leq q+k$, then we charge edge $(u, v)$ to vertex $v \in U$. Under this charging scheme  each vertex in $U$ is charged at most $2(q+k)$ times, so the total number of edges in $H[U]$ is at most  $|E(H[U])| \leq 2(q+k)|U|$.
\end{proof}

We also need the next lemma, which reduces the task of constructing an all-pairs $k$-connectivity preserver to preserving $k$-connectivity between just $O(n)$ demand pairs. This lemma follows from well-known results about Gomory-Hu trees \cite{gomory1961multi,cheng1991ancestor}.

\begin{lemma}
    For any $n$-vertex directed graph $H = (V, E)$, there exists  a collection of vertex demand pairs $P \subseteq V \times V$ of size $|P| = O(n)$ such that for any edge subgraph $H' \subseteq H$, $H'$ is a $k$-connectivity preserver of $H$ if and only if for every demand pair $(s, t) \in P$, $\conn_{H'}^k(s, t) = \conn_H^k(s, t)$.
    \label{lemma: demand-pairs}
\end{lemma}

\begin{proof}
    The ``only if" direction is trivial, so  we focus on proving the ``if" direction below.   Let $C_H$ be a complete undirected weighted graph with vertex set $V$. We assign each edge $(u, v) \in E(C_H)$ weight $\conn_H(u,v)$. Let $T$ be a maximum spanning tree of $C_H$.

    For any edge $(u,v) \in E(C_H)$, consider the $(u, v)$-path $Q = (u_0, u_1, \dots, u_\ell)$ in $T$, where $u = u_0$ and $v = u_\ell$. We claim that $\conn_H(u,v) = \min_{(s, t) \in Q} \conn_H(s, t)$. On one side, $\conn_H(u,v) \geq \min_{(s, t) \in Q} \conn_H(s, t)$, by the triangle inequality of symmetric connectivity. On the other side, $\conn_H(u,v) \leq \min_{(s, t) \in Q} \conn_H(s, t)$, since $T$ is a maximum spanning tree in $C_H$.

    We define our set of vertex demand pairs $P \subseteq V \times V$ to be the edges of tree $T$, so that $P = E(T)$. Now suppose that $H' \subseteq H$ is an edge subgraph of $H$ that satisfies $ \lambda_{H'}^k(s, t) = \lambda_H^k(s,t)$ for each demand pair $(s, t) \in P$. Then by the  triangle inequality of symmetric connectivity, we have \[\conn_{H'}^k(u,v) \geq \min_{(s, t) \in Q} \conn_{H'}^k(s, t) = \min_{(s, t) \in Q} \conn_{H}^k(s, t) = \conn_{H}^k(u,v).\]

    Therefore, $H'$ preserves the $k$-bounded connectivity for all vertex pairs $(u,v)$ in $H$.
\end{proof}

We will need the following lemma in order to eventually bound the number of edges between components in $\mathcal{S}$. 

\begin{lemma}
    \label{lemma: k-conn-cutsize} 
    Let $H$ be an $n$-vertex,  $k$-critical directed graph. For every cut $(L,R)$ with size $|\delta^+(L)| \leq k$, we have $|\delta^-(L)| = O(kn)$.
\end{lemma}
\begin{proof}
    Let $P \subseteq V(H) \times V(H)$ be the set of $|P| = O(n)$ demand pairs specified in \Cref{lemma: demand-pairs}. For each demand pair $(s, t) \in P$, there exists a collection $\mathcal{Q}_{s, t}$ of $\conn_H^k(s,t)$ pairwise edge-disjoint paths from $s$ to $t$ in $H$, and there exists a collection $\mathcal{Q}_{t, s}$ of $\conn_H^k(s, t)$ pairwise edge-disjoint paths  from $t$ to $s$ in $H$. For each path $Q \in \mathcal{Q}_{s,t}\cup \mathcal{Q}_{t,s}$, it must pass through edges in $\delta_H^+(L)$ and edges in $\delta_H^-(L)$ interleavingly, which means the number of edges in $\delta_H^-(L)$ on the path is at most the number of edges in $\delta_H^+(L)$ on the path, plus one. Since the paths in $\mathcal{Q}_{s, t}$ are pairwise edge-disjoint, the paths in $\mathcal{Q}_{s, t}$ contain at most $|\delta_H^+(L)|+ \lambda^k_H(s, t) \leq 2k$ edges in $\delta_H^-(L)$. An identical argument implies that the paths in $\mathcal{Q}_{t, s}$ contain at most $2k$ edges in $\delta_H^-(L)$ as well.  
    % Since the number of edges in $\delta^+(L)$ on all of these paths is at most $k$, it follows that the number of edges in $\delta^-(L)$ on all of these paths will be at most $k+\conn_H^k(s_i,t_i) \leq 2k$. 

    % We have shown that for each demand pair $(s_i, t_i) \in P$, we can preserve $\lambda^k_H(s_i, t_i)$ using only $2k$ edges crossing cut $(L, R)$. 

    % Let $H' \subseteq H$ be an edge subgraph of $H$ such that $\lambda^k_{H'}(s_i, t_i) = \lambda^k_H(s_i, t_i)$. 

    % For any pair $(s_i,t_i) \in P$, we pick an arbitrary flow from $s_i$ to $t_i$ with size $\conn_H^k(s_i,t_i)$, and an arbitrary flow from $t_i$ to $s_i$ with size $\conn_H^k(s_i,t_i)$.

    Let $H' \subseteq H$ be the subgraph of $H$ obtained by unioning all the paths in all collections $\mathcal{Q}_{s, t} \cup \mathcal{Q}_{t, s}$ for all $(s, t) \in P$, so that $H' =  
    \cup_{\{Q \in \mathcal{Q}_{s, t} \cup \mathcal{Q}_{t, s} \mid (s, t) \in P\}} Q$.  We observe that $\conn_{H'}^k(s, t) = \conn_H^k(s, t)$. 
    Graph $H'$ is the union of $2|P|$ collections of paths $\mathcal{Q}_{s, t}$, and the paths in $\mathcal{Q}_{s, t}$ contain at most $2k$ edges in $\delta_H^-(L)$, by the earlier discussion. We conclude that there are at most  $|\delta_{H'}^-(L)| \le 2k \cdot 2|P|=O(kn)$ edges in $\delta_{H'}^-(L)$. Now observe that for any demand pair $(s, t) \in P$, $\conn_{H'}^k(s, t) = \conn_H^k(s, t)$,  so $H'$ is a $k$-connectivity preserver of $H$ by \Cref{lemma: demand-pairs}. Finally, since $H$ is $k$-critical, we know $H'=H$ and thus $|\delta_H^-(L)| = |\delta_{H'}^-(L)| = O(kn)$.
\end{proof}

We can now bound the number of cross-component edges in $H[S]$.

\begin{lemma}
    \label{lemma: k-conn-inter}
    Let $E'=\{(u,v)\in E(H[S]) \mid (u, v) \in S_1 \times S_2, S_1 \ne S_2 \in \mathcal{S}\}$ be the set of edges in $H[S]$ between different components in $\mathcal{S}$. Then $|E'| = O(kn^2/q)$.
\end{lemma}
\begin{proof}
Recall that $\mathcal{C}$ is the family of all cuts $(L, R)$ that we considered in our construction of partition $\mathcal{S}$. By our construction of $\mathcal{S}$, every edge $e \in E'$ crosses a cut $(L, R)$ in cut family $\mathcal{C}$ (i.e., $e \in \delta^+(L) \cup \delta^-(L)$ for some cut $(L, R) \in \mathcal{C}$). 

By \Cref{clm:termcondition}, $|\mathcal{C}| \leq n/q$, and by \Cref{lemma: k-conn-cutsize}, we have that $|\delta^+(L) \cup \delta^-(L)| = O(kn)$ for each cut $(L, R) \in \mathcal{C}$. 
    % Let $\mathcal{C}$ be the set of cuts that we performed in our construction of partition $\mathcal{S}$. Every time when we cut $(L_i,R_i)$ through some component $S_i \in \mathcal{S}$, $|\mathcal{S}|$ is increased by one. Since finally $|\mathcal{S}|\leq n/q$, we also have $|\mathcal{C}|\leq n/q$.
    % Observe that $E' \subseteq \cup_{(L,R)\in \mathcal{C}} (L,R)$. By \Cref{lemma: k-conn-cutsize}, since any of these cut $(L,R)$ has $|\delta^+(L)| \leq k$, the total number of edges crossing $L$ and $R$ will be $|\delta^+(L) \cup \delta^-(L)| \leq O(kn)$. 
    Therefore,\[|E'| \leq \sum_{(L,R)\in \mathcal{C}} |\delta^+(L) \cup \delta^-(L)| \leq O((n/q )\cdot kn) = O(kn^2/q). \]
\end{proof}

\begin{proof}[Proof of \Cref{thm:k_critical}]
     By \Cref{clm:termcondition} and \Cref{lemma: k-conn-intra}, we know that the total number of edges in $\cup_{S\in \mathcal{S}} H[S]$ is at most $$\left|\bigcup_{S\in \mathcal{S}} H[S]\right| \le \sum_{S\in \mathcal{S}} O((q+k)|S|) = O((q+k)n).$$     
     Additionally, by \Cref{lemma: k-conn-inter}, $|E'| = O(kn^2/q)$.
     We conclude that
     $$
     |E(H)| \leq |\cup_{S\in \mathcal{S}} H[S]| + |E'| = O((q+k)n + kn^2/q) = O(k^{1/2}n^{3/2}), 
     $$
     where the final inequality follows by choosing $q=\sqrt{nk}$ and observing that $k \leq n$. 
\end{proof}

\section{A Small Cut Containing All Important Cuts}
% \section{Covering all Important Cuts with a Small Cut}
\label{sec:cut_theorem}

In this section, we develop a new theorem regarding important cuts that is useful for our fault-tolerant symmetric connectivity preservers. Roughly, we show that for any two disjoint sets of vertices $X, Y$, there exists a (small) $(X, Y)$-cut containing every (small) important $(X, Y)$-cut. We state our cut theorem below, in full generality. 

\begin{theorem}
    Let $G$ be an $m$-edge directed multigraph, and let $X, Y \subseteq V(G)$ be two disjoint sets of vertices with $\flow(G, X, Y) = \lambda$. For every integer $k \geq 0$, there exists an $(X, Y)$-cut $(S, V-S)$  of size $|\delta^+(S)| \leq \lambda 2^k$ such that every important  $(X, Y)$-cut $(S', V - S')$ of size $|\delta^+(S')| \leq \lambda + k$ satisfies $S' \subseteq S$. Moreover, cut $(S, V-S)$ can be computed in time 
    $
    O(\lambda k2^km+\lambda^2k^24^k).
    $
    \label{thm:ecut}
\end{theorem}

We use \Cref{thm:ecut} frequently in \Cref{sec:conn_pres} in the following (less general) form.

\impcutlem*

We quickly verify that \Cref{lem:important_separators} follows directly from \Cref{thm:ecut}. 

\begin{proof}[Proof of \Cref{lem:important_separators}]
    Let $\flow(G, X, Y) = \lambda$. If $\lambda > k$, then there are no important $(X, Y)$-cut $(S, V-S)$ in $G$ of size at most $|\delta^+(S)| \leq k$, and \Cref{lem:important_separators} is vacuously true. 

    Otherwise, $\lambda \leq k$. Let $k^* = k - \lambda$. We apply \Cref{thm:ecut} with respect to vertex sets $X, Y$ in graph $G$ and integer parameter $k^*$. This implies that there exists an $(X, Y)$-cut $(S, V-S)$ of size $|\delta^+(S)| \leq \lambda 2^{k^*}$ such that  every important  $(X, Y)$-cut $(S', V - S')$ of size $|\delta^+(S')| \leq \lambda + k^*$ satisfies $S' \subseteq S$. Moreover, $(S, V-S)$ can be computed in the time $O(\lambda k^* 2^{k^*} m + \lambda^2 k^{*2}4^{k^*})$. 

    Now we can use the fact that $k = \lambda + k^*$ to prove \Cref{lem:important_separators}. Since $\lambda \leq 2^{\lambda-1}$, we observe that $|\delta^+(S)| \leq \lambda 2^{k^*} \leq 2^{k-1}$. Likewise, $|\delta^+(S')| \le \lambda + k^* = k$. Then the $(X, Y)$-cut $(S, V-S)$ has size at most $|\delta^+(S)| \leq 2^{k-1}$ and has the property that every important $(X, Y)$-cut $(S', V-S')$ of size at most $|\delta^+(S')| \leq k$ satisfies $S' \subseteq S$. Finally, using the fact that $\lambda \leq 2^{\lambda}$, cut $(S, V-S)$ can be computed in $O(k2^km+k^24^k)$ time.
\end{proof}

% Qualitatively, our cut theorems roughly state that for disjoint sets of vertices $X, Y \subseteq V(G)$, there exists a \textit{small} $(X, Y)$ cut 

% By a well-known theorem in parameterized complexity, there are at most $4^k$ important $(X, Y)$ separators of size at most $k$ \cite{?}. In this section, we show that there exists a single $(X, Y)$ cut of size at most $2^k$ that contains every $(X, Y)$ important separator of size at most $k$. 

% \subsection{Vertex Separator}

% The goal of this subsection is to prove the following theorem.

% \begin{theorem}
%   Let $G$ be an $m$-edge directed graph, and let $X, Y \subseteq V(G)$ be two disjoint sets of vertices with $\vflow(X, Y) = \lambda$. For every integer $k \geq 0$, there exists an $(X, Y)$ vertex cut $(L, S, R)$ of size $|S| \leq \lambda 2^k$ such that every   important vertex $(X, Y)$-separator $(L', S', R')$ of size $|S'| \leq \lambda + k$ satisfies $L' \cup S' \subseteq L \cup S$.  Moreover, $(L, S, R)$ can be computed in time 
%     $$
%     \min\left(k(m+\lambda k 2^k)^{1+o(1)}, \hspace{1mm} O(\lambda k2^km+\lambda^2k^24^k) \right).
%     $$
%     \label{thm:vcut}
% \end{theorem}

For the remainder of this section, let $G$ be an $n$-vertex, $m$-edge directed multigraph, and let $X, Y \subseteq V(G)$ be two disjoint sets of vertices with $\flow(G, X, Y) = \lambda$. We may assume without loss of generality that $X$ and $Y$ are singleton sets, so that $X = \{x\}$ and $Y = \{y\}$ (e.g., by creating an artificial source $x$ with edges to $X$ and an artificial sink $y$ with edges from $Y$).

We first describe how to efficiently compute the cut claimed in Theorem \ref{thm:ecut}, and then we prove that it has the desired properties.  Our proof of \Cref{thm:ecut} will follow \cite{baswana2018fault,bansal2024faulttolerantboundedflowpreservers} and our results can be seen as a generalization of them.

\paragraph{Constructing the cut.}

We iteratively define a sequence of graphs $G_0, \dots, G_k$ over vertex set $V(G)$. Let graph $G_0 := G$. Given a graph $G_i$, where $i \in [0, k]$, let $(S_i, V - S_i) = \fmc(G_i, x, y)$ be the farthest minimum $(x, y)$-cut in $G_i$. Given set $S_{i}$, where $i \in [0, k-1]$, we define graph $G_{i+1}$ as follows:
\begin{itemize}
    \item Initially, let $G_{i+1} := G_{i}$.
    \item For each edge $(u, v) \in \delta_{G_i}^+(S_i)$, add an edge $(x, v)$ to $G_{i+1}$. (If edge $(x, v)$ already exists in $G_i$, then we create an additional parallel $(x, v)$ edge.) This completes the construction of $G_{i+1}$. 
\end{itemize}
Observe that $(S_i, V - S_i)$ is an $(x, y)$-cut   in $G$ for all $i \in [0, k]$. We return $(S_k, V - S_k)$ as the $(x, y)$-cut claimed in Theorem \ref{thm:ecut}. 

\paragraph{Time complexity.} We now prove that cut $(S_k, V-S_k)$ can be computed in the time claimed in Theorem \ref{thm:ecut}. 

\begin{lemma} \label{lem:time_comp}
Cut $(S_k, V-S_k)$ can be computed in time $O(\lambda k2^km+\lambda^2k^24^k)$.
\end{lemma}
We will require the following claim bounding the size of cut $(S_k, V-S_k)$ in graph $G_k$. 

\begin{claim}
In graph $G_k$, the $(x, y)$-cut $(S_k, V-S_k)$ has size at most $|\delta^+_{G_k}(S_k)| \leq \lambda 2^k$. 
\label{clm:vcut}
\end{claim}
\begin{proof}
We will prove by induction that for every $i \in [0, k]$, $\flow(G_i, x, y) \leq \lambda 2^i$. When $i=0$, the claim is satisfied since we assumed $\flow(G, x, y) = \lambda$. Next, we show that for all $i \in [0, k-1]$, $$\flow(G_{i+1}, x, y) \leq 2\flow(G_{i}, x, y).$$ We will need several observations to prove this:
\begin{itemize}
    % \item $\vflow(X, Y, G_i) = c_i(Z_{i+1})$. This follows from the (vertex) max-flow min-cut theorem because $(L_i, S_{i}, R_i)$ is a minimum (vertex) cut between $X$ and $Y$ in $G_i$. 
    \item For each $i \in [0, k]$, $\flow(G_i, x, y) = |\delta^+_{G_i}(S_i)|$, since $(S_i, V-S_i)$ is the farthest minimum $(x, y)$-cut in $G_i$. 
    \item For each $i \in [0, k-1]$,  $\flow(G_{i+1}, x, y) \leq |\delta_{G_{i+1}}^+(S_{i})|$. This holds because $(S_{i}, V - S_{i})$ is an $(x, y)$-cut in $G_{i+1}$. 
    \item For each $i \in [0, k-1]$, $|\delta^+_{G_{i+1}}(S_{i})| \leq 2|\delta^+_{G_{i}}(S_{i})|$. We construct graph $G_{i+1}$ from graph $G_i$ by adding at most $|\delta^+_{G_i}(S_i)|$ edges to $G_{i}$. Then there can be at most $|\delta^+_{G_i}(S_i)|$ additional edges crossing cut $(S_i, V - S_i)$ in $G_{i+1}$. We conclude that $|\delta^+_{G_{i+1}}(S_{i})| \leq 2|\delta^+_{G_{i}}(S_{i})|$, as claimed. 
\end{itemize}
Putting these observations together, we find that
\begin{align*}
\flow(G_{i+1}, x, y) 
& \leq  |\delta^+_{G_{i+1}}(S_{i})|  \leq 2 |\delta^+_{G_i}(S_i)| =  2 \cdot \flow(G_{i}, x, y).
\end{align*}
Then by induction, we conclude that $\flow(G_k, x, y) \leq \lambda 2^k.$  In particular, since $(S_k, V-S_k)$ is a minimum $(x, y)$-cut in $G_k$, $|\delta^+_{G_k}(S_k)| \leq \lambda 2^k$. 
\end{proof}

We are now ready to prove Lemma \ref{lem:time_comp}. 

\begin{proof}[Proof of Lemma \ref{lem:time_comp}]
Cut $(S_k, V-S_k)$ can be obtained by computing an $(x, y)$ maximum flow on the $k+1$ graphs $G_0, \dots, G_k$. By Claim \ref{clm:vcut}, for each $i \in [0, k-1]$,  $$|E(G_{i+1})| \leq |E(G_i)| + \lambda 2^i \leq m + k \lambda 2^k.$$ 
Likewise, by Claim \ref{clm:vcut}, $\flow(G_i, x, y) \leq  \lambda 2^k$ for all $i \in [0, k]$. Then we can compute an $(x, y)$ maximum flow on graph $G_i$ in time $O(\lambda 2^k \cdot (m+k \lambda 2^k))$ using Ford-Fulkerson \cite{ford2015flows}. Then the overall time complexity becomes $O(\lambda k 2^k m + \lambda^2 k^2 4^k)$, as claimed.  
\end{proof}

\paragraph{Completing the proof of Theorem \ref{thm:ecut}.} We have identified an $(x, y)$-cut $(S_k, V-S_k)$ in $G$ that can be computed in time $O(\lambda k 2^k m + \lambda^2 k^2 4^k)$ time by Lemma \ref{lem:time_comp}. Moreover, $|\delta^+_{G}(S_k)| \leq \lambda 2^k$ by Claim \ref{clm:vcut}. 
All that remains to complete the proof of Theorem \ref{thm:ecut} is to verify that cut $(S_k, V-S_k)$ contains every important $(x, y)$-cut $(S, V-S)$ of size at most $|\delta^+_G(S)| \leq \lambda + k$. We will need a technical lemma about farthest minimum cuts that was proved in \cite{baswana2018fault}. 

% \begin{lemma}[Lemma 4.3 of \cite{baswana2018fault}]
%         Let $G$ be a directed graph with edge capacities. Let  $s, t \in V(G)$ be vertices in $G$, and let $\fmc(s, t, G) = E(S, V-S)$. For each vertex $v\in V - S$, let $G_v$ be the graph obtained by adding a (possibly parallel) edge $(s, v)$ to $G$. Then
%     $$
%     \flow(s, t, G_v) = \flow(s, t, G) + 1,
%     $$
%     and $E(S, V-S)$ is an $(s, t)$-minimum cut in $G_v$. 
%     \label{clm:fmc}
%     \label{lem:baswana}
% \end{lemma}

% We can directly port this lemma to the setting of vertex cuts by a standard reduction.

\begin{lemma}[cf. Lemma 10 of \cite{baswana2018fault}]
    Let $G$ be a directed multigraph. Let  $s, t \in V(G)$ be vertices in $G$, and let $(S, V-S) = \fmc(G, s, t)$. For each vertex $v\in V -S$, let $G_v$ be the graph obtained by adding edge $(s, v)$ to $G$, i.e., $G_v = G + (s,v )$. Then
    $$
    \flow(G_v, s, t) = \flow(G, s, t) + 1,
    $$
    and $(S, V-S)$ is a minimum $(s, t)$-cut in $G_v$. 
    \label{lem:fmc}
\end{lemma}

Lemma \ref{lem:fmc} has the following general implication about farthest minimum cuts in directed multigraphs. 

\begin{claim} \label{clm:fmc_further}
    Let $G = (V, E)$ be a directed multigraph. Let $s, t \in V$ be vertices in $G$, and let $(S, V-S) = \fmc(G, s, t)$. Let $G' = (V, E \cup E')$, where $E' \subseteq \{s\} \times V$, and let $(S', V-S') = \fmc(G', s, t)$. Then $S \subseteq S'$.  
\end{claim}
\begin{proof}
We can prove \Cref{clm:fmc_further} by repeatedly applying \Cref{lem:fmc}. Let $E' = \{e_1, \dots, e_p\}$, Let $G_0 = G$, and let $G_i = G + e_1 + \dots + e_i$ for $i \in [1, p]$. Let $(T_i, V-T_i) = \fmc(G_i, s, t)$ for $i \in [0, p]$. By \Cref{lem:fmc}, $T_i \subseteq T_{i+1}$ for $i \in [0, p-1]$. Then in particular,  $S = T_0 \subseteq T_p = S'$ as claimed.  
\end{proof}

In particular, \Cref{clm:fmc_further} implies that the sets $S_0, \dots, S_k$ are nested. 

\begin{claim}
   For each $i \in [0, k-1]$,  $S_i \subseteq S_{i+1}$. 
   \label{clm:nested}
\end{claim}
\begin{proof}
Notice that for each $i \in [0, k-1]$, $G_{i+1} = (V, E(G_i) \cup E_{i+1})$, where edges $E_{i+1}$ satisfy $E_{i+1} \subseteq \{x\} \times V$. Then by \Cref{clm:fmc_further}, $S_i \subseteq S_{i+1}$ as claimed. 
\end{proof}

% \begin{claim}
%     Let $G = (V, E)$ and $G'= (V, E')$ be directed graphs with vertex capacities. Fix vertices $s, t \in V(G)$, and let $\vfmc(s, t, G) = (L, S, R)$ and  $\vfmc(s, t, G') = (L', S', R')$. If $G' \subseteq G$ and $E(G) \setminus E(G') \subseteq \{s\} \times V$, then $L' \cup S' \subseteq L \cup S$. 
%     \label{clm:fmc_edge}
% \end{claim}
% \begin{proof}
    
% \end{proof}

We  now finish our proof of Theorem \ref{thm:ecut} with the following lemma. 

\begin{lemma}
If $(S, V-S)$ is an important $(x, y)$-cut  with size $|\delta^+_G(S)| \leq \lambda + k,$ then $S \subseteq S_k$. 
\label{lem:vcut}
\end{lemma}
\begin{proof}
Suppose towards contradiction that there exists an important $(x, y)$-cut $(S, V-S)$  with size $|\delta^+_G(S)| \le \lambda + k$ such that $S \not \subseteq S_k$, and let $v \in S \cap (V - S_k)$. Since $(S, V-S)$ is an important $(x, y)$-cut, there exists an $(x, v)$-path $P$ in $G - \delta^+_G(S)$. We will use the existence of $P$ to obtain a contradiction. 
Since $v \not \in S_k$ and $x \in S_i \subseteq S_k$ for each $i \in [0, k]$,  path $P$ contains an edge in $\delta^+_G(S_i)$ for each $i \in [0, k]$. We let $(u_i, v_i) \in E(P)$ denote an arbitrary edge in  $E(P) \cap \delta^+_G(S_i)$ for each $i \in [0, k]$. 

For analysis purposes, we  define a sequence of multigraphs $H_0, H_1 \dots, H_{k+1}$ over vertex set $V$. 
Let graph $H_{0} := G$. For each $i \in [0, k]$, we let $H_{i+1} = H_0 + (x, v_0) + \dots + (x, v_i)$. 
Notice that in particular, $H_{i+1} = H_{i} + (x, v_i)$.

Let $(U_i, V-U_i) = \fmc(H_i, x, y)$ for each $i \in [0, k]$. We observe that for each $i \in [0, k]$, $H_i \subseteq G_i$ and  $E(G_i) - E(H_i) \subseteq \{x\} \times V$. Then by \Cref{clm:fmc_further}, $U_i \subseteq S_i$ for each $i \in [0, k]$. Then since $H_{i+1} = H_i + (x, v_i)$, where $v_i \in V - S_i \subseteq V - U_i$, it follows that $$\flow(H_{i+1}, x, y) \geq \flow(H_i, x, y) + 1$$ for each $i \in [0, k]$ by \Cref{lem:fmc}. In particular, since $\flow(H_0, x, y) = \lambda$, we conclude that $\flow(H_{k+1}, x, y) \geq \lambda + k + 1$.

We are ready to obtain our contradiction. Since $\flow(H_{k+1}, x, y) \geq \lambda + k + 1$ and $|\delta^+_G(S)| \leq \lambda + k$, there exists an $(x, y)$-path $P^*$ in graph $H_{k+1} - \delta^+_G(S)$. If path  $P^*$ is contained in  graph $G -  \delta^+_G(S)$, then this immediately implies a contradiction because $(S, V-S)$ is an $(x, y)$-cut. Otherwise, path $P^*$ is of the form $P^* = (x, v_i) \circ P^*[v_i, y]$ for some $i \in [0, k]$, where edge $(x, v_i)$ is in  graph $H_{k+1}$ and path $P^*[v_i, y]$ is in graph $G$. However, in this case, path $P[x, v_i]$ is contained in graph $G - \delta^+_G(S)$, so we conclude that the $(x, y)$-path $P[x, v_i] \circ P^*[v_i, y]$ is contained in graph $G - \delta^+_G(S)$, contradicting our assumption that $(S, V-S)$ is an $(x, y)$-cut. 
\end{proof}

\subsection{Anti-isolation lemma} 

\label{subsec:isolation}

In this section, we will use \Cref{lem:important_separators} to show our improved anti-isolation lemma (\Cref{lem:anti isolation}).

\isolation*
\begin{proof}
For each edge set $F_i$,  where $i \in [1, r]$, there exists a cut $(S_i, V-S_i)$ such that
\begin{enumerate}
    \item $\delta^+_G(S_i) \subseteq F_i$,
        \item $t_i \in S_i$, and
    \item $(S_i, V-S_i)$ is an out-reachable $(s, t_j)$-cut for each $j \neq i \in [1, r]$. 
\end{enumerate}

Now to obtain an upper bound on $r$ using \Cref{lem:important_separators}, we define a new graph $H$  by adding an artificial sink vertex $t^*$ to $G$ and adding the edge $(t_i, t^*)$ to $H$ for all $i \in [1, r]$. Then graph $H$ is defined as $H = G + t^* + (t_1, t^*) + \dots + (t_r, t^*)$.

We observe that $(S_i, V \cup \{t^*\} - S_i)$ is an out-reachable $(s, t^*)$-cut in $H$ of size $|\delta^+_H(S_i)| \leq k +1$ for each $i \in [1, r]$.  By  \Cref{lem:important_separators} and \Cref{clm:impcut farthest}, there exists an $(s, t^*)$-cut $(S^*, V-S^*)$ in $H$ of size $|\delta^+_H(S^*)| \leq 2^k$ such that $S_i \subseteq S^*$ for each $i \in [1, r]$. This implies that $(t_i, t^*) \in \delta^+_H(S^*)$ for each $i \in [1, r]$. We conclude that $r \leq |\delta^+_H(S^*)| \leq 2^k$. 
\end{proof}

\section{Conclusion and Open Problems}

\label{sec:open}

\Cref{thm:FT} and \Cref{thm:kconn} present the first $k$-FT connectivity preservers and $k$-connectivity preservers in directed graphs with $O(n\log n)$ edges for every constant $k$, while \Cref{thm:polysize} shows a $k$-connectivity preserver with a non-trivial size for any $k = o(n)$. Given the significance and wide-ranging applications of connectivity preservers in undirected graphs \cite{nagamochi1992linear,benczur1996approximating,benczur2015randomized,dynamic1,thorup2007fully,abraham2016fully,thurimella1995sub,daga2019distributed,parter2019small,ahn2012graph,guha2015vertex,assadi2023tight}, we are optimistic that our generalization to directed graphs in \Cref{thm:FT} will also yield further applications. Below, we outline some intriguing open problems motivated by our results.

\paragraph{Optimal Size Bounds.}

Our approach for proving \Cref{thm:FT} incurs a factor of $O(\log n)$ due to the depth of the expander hierarchy. Eliminating this $O(\log n)$ factor from the $k$-FT connectivity preservers might require a completely different approach, making it an exciting challenge.

Additionally, it is  important to determine the optimal constant in the factor $2^{\Theta(k)}$ across all variants of $k$-FT connectivity preservers presented in \Cref{tab:bounds}. Lastly, we believe that our upper bound of $O(n\sqrt{nk})$ for $k$-connectivity preservers can be improved to match the lower bound:

\begin{conjecture}

For every positive integer $k$, every $n$-vertex directed graph admits a $k$-connectivity preserver with $O(kn)$ edges.

\end{conjecture}

Resolving this conjecture would be exciting.

\paragraph{Near-Linear Construction Time.}

For many applications including fast algorithms for computing minimum cuts, it is crucial that the preservers are constructed very efficiently. We consider it an important open problem whether one can construct $k$-fault-tolerant connectivity preservers or $k$-connectivity preservers in directed graphs in near-linear time, even for constant $k$. However, it is worth noting that this remains unresolved even in the single-source version; the algorithm in \cite{baswana2018fault} requires $O(2^{k}mn)$ time, or $O(kmn^{1+o(1)})$ when using the almost-linear time max flow algorithm by \cite{algs1}.

% \paragraph{Fault-Tolerant Roundtrip Spanners.}
% Roundtrip spanners were first introduced in \cite{roditty2008roundtrip} as a  generalization of multiplicative spanners to directed graphs. Roundtrip spanners are sparse subgraphs that approximately preserve the roundtrip distance $\text{dist}_G(s, t) + \text{dist}_G(t, s)$ for all pairs of vertices $u, v \in V(G)$. The best-known general size-stretch tradeoff for roundtrip spanners is due to \cite{cen2020roundtrip}, which constructs roundtrip spanners with size $\widetilde{O}(kn^{1+1/k})$ and stretch $2k-1$. To the best of our knowledge, there is no prior work on fault-tolerant roundtrip spanners. It seems plausible to conjecture that there exists $f$-fault-tolerant roundtrip spanners with size $ \widetilde{O}(4^f \cdot kn^{1+1/k})$ and stretch $2k-1$. One possible approach to achieving these bounds could be to extend our arguments to handle distances using a length-constrained directed expander hierarchy.

\paragraph{Acknowledgment}{We thank an anonymous reviewer for many constructive suggestions that helped us improve the writing.}

\bibliographystyle{alpha}
\bibliography{ref}

\appendix

\section{A hierarchy of connectivity preserver problems}

\label{sec:reduc}

% The strong separation between $k$-connectivity and $k$-FT connectivity established by \Cref{thm:lambdaconn} hints at a rich landscape of connectivity preserver problems in directed graphs. % The separation between  $k$-connectivity preservers and $k$-FT connectivity preservers implied by \Cref{thm:lambda_conn} 
% hints at a rich landscape of directed connectivity preserver problems. 
In this section, we aim to survey the landscape of 
connectivity preserver problems in directed graphs and develop an improved understanding of it. We begin by defining a small family of related connectivity preserver problems. We will need the following definition.

% In this section, we aim to reveal a full landscape of current results in (symmetric) connectivity preservers, 
% in both $k$-FT and $k$-connectivity settings. We define our preservers below, and explore the relations between them.

% Note that $\lambda_G^{1}(s,t)=1$ if and only if $s$ and $t$ are strongly connected in $G$. 

\begin{definition}[$k$-Bounded Connectivity] For a graph $G$, we define $\lambda_{G}^{k}(s,t)=\min(\lambda_{G}(s,t),k)$ as the $k$-bounded connectivity between $s$ and $t$ in $G$. (Recall that we use $\lambda_{G}(s,t)$ to denote the symmetric edge-connectivity between $s$ and $t$ in $G$.) We also use $\lambda^{k}(G)=\min(\lambda(G),k)$ to denote the $k$-bounded connectivity of a graph. \end{definition}

With the definition of $k$-bounded connectivity, we can define two families of directed connectivity preservers as follows.

\begin{definition}[$k$-Connectivity Family]\label{def:k-preservers} Let $G=(V,E)$ be a directed graph. For a subgraph $H$ of $G$, we say $H$ is a
\begin{itemize}
\item[\textendash{}] \textbf{$\textbf{\textit{s}}$-$\textbf{\textit{t}}$} $k$-connectivity preserver if $\lambda_{H}^{k}(s,t)=\lambda_{G}^{k}(s,t)$ for a given source-sink pair $(s,t)\in V\times V$, 
\item[\textendash{}] \textbf{global} $k$-connectivity preserver if $\lambda^{k}(H)=\lambda^{k}(G)$, 
\item[\textendash{}] \textbf{single-source} $k$-connectivity preserver if $\lambda_{H}^{k}(s,t)=\lambda_{G}^{k}(s,t)$ for a given $s\in V$ and any $t\in V$, 
\item[\textendash{}] \textbf{all-pairs} $k$-connectivity preserver if $\lambda_{H}^{k}(s,t)=\lambda_{G}^{k}(s,t)$ for any $s,t\in V$. 
\end{itemize}
% is a $s$-$t$ / global / single-source (symmetric) $k$-connectivity preserver of $G$ if the $k$-bounded $s$-$t$ / global / single-source connectivity (i.e. the minimum of connectivity and $k$) in $H$ is the same as in $G$.    
\end{definition}

Note that an all-pairs $k$-connectivity preserver is exactly the same as the $k$-connectivity preserver defined in \Cref{subsec:intro kconn}. We can define an analogous family of connectivity preservers for the fault-tolerant setting.

\begin{definition}[$k$-FT Connectivity Family]\label{def:k-fault preservers} Let $G=(V,E)$ be a directed graph. For an subgraph $H$ of $G$, we say $H$ is a
\begin{itemize}
\item[\textendash{}] $\textbf{\textit{s}}$\textbf{-}$\textbf{\textit{t}}$ $k$-FT connectivity preserver if $\lambda_{H-F}^{1}(s,t)=\lambda_{G-F}^{1}(s,t)$ for a given source-sink pair $(s,t)\in V\times V$ and any set of edges $F\subseteq E$ with $|F|\leq k$, 
\item[\textendash{}] \textbf{global} $k$-FT connectivity preserver if $\lambda_{}^{1}(H-F)=\lambda_{}^{1}(G-F)$ for any set of edges $F\subseteq E$ with $|F|\leq k$, 
\item[\textendash{}] \textbf{single-source} $k$-FT connectivity preserver if $\lambda_{H-F}^{1}(s,t)=\lambda_{G-F}^{1}(s,t)$ for a given source $s\in V$, any $t\in V$, and any set of edges $F\subseteq E$ with $|F|\leq k$, 
\item[\textendash{}] \textbf{all-pairs} $k$-FT connectivity preserver if $\lambda_{H-F}^{1}(s,t)=\lambda_{G-F}^{1}(s,t)$ for any $s,t\in V$ and any set of edges $F\subseteq E$ with $|F|\leq k$. 
\end{itemize}
% A subgraph $H \subseteq G$ is a $k$-FT $s$-$t$ / global / single-source (symmetric) connectivity preserver of $G$ if for any set $F \subseteq E(G)$ of at most $|F| \leq k$ edge faults, the $s$-$t$ / global / single-source connectivity in $H-F$ is the same as in $G - F$.    
\end{definition} We again note that an all-pairs $k$-FT connectivity preserver is exactly the same as the $k$-FT connectivity preserver defined in the introduction.

As it turns out, the preservers in our connectivity preserver families can be ordered in a simple hierarchy according to how many edges they require. We formalize this hierarchy with the following theorem, which we prove in \Cref{app:reduction}.

% \begin{remark} \gary{ I personally feel like this remark is not really needed }
% For $s$-$t$ ($k$-FT / $k$-connectivity) preservers, it is almost equivalent to state symmetric connectivity and (asymmetric) reachability, as they can reduce to each other within a constant factor as follows: (The same observation also holds for single-source preservers.)
%     \begin{itemize}
%         \item To construct a symmetric connectivity preserver, we can construct two reachability preservers, with one from $s$ to $t$ and the other from $t$ to $s$.
%         \item To construct a reachability preserver from $s$ to $t$ in $G$, we add $(k+1)$ edges from $t$ to $s$ to form a graph $G'$ and construct a symmetric connectivity preserver $H'$ in $G'$. One can see $H' \cap G$ is a reachability preserver from $s$ to $t$ in $G$.
%     \end{itemize}
% \end{remark}
% \bwnoteinline{I don't want to make a lot of definitions. I also wrote a remark here to avoid tedious proofs. These don't look rigorous. Is it ok?}
% \ghnoteinline{I think it's good. Actually, I kind of feel like we do not need the remark at all. }

\begin{theorem} \label{thm:reductions} The sizes of the preservers in the $k$-FT preserver family can be related as follows: 
\begin{enumerate}
\item $\textbf{\textit{s}}$\textbf{-}$\textbf{\textit{t}}$ $\Rightarrow$ \textbf{global}. If there exists a global $k$-FT preserver with size $f(n,k)$, then there exists an $s$-$t$ $k$-FT preserver with size $O(f(n,k))$. 
\item \textbf{global} $\Rightarrow$ \textbf{single-source}. If there exists a single-source $k$-FT preserver with size $f(n,k)$, then there exists a global $k$-FT preserver with size $O(f(n,k))$. 
\item \textbf{single-source} $\Rightarrow$ \textbf{all-pairs}. If there exists an all-pairs $k$-FT preserver with size $f(n,k)$, then there exists a single-source $k$-FT preserver with size $O(f(n,k))$. 
\end{enumerate}
Analogous reductions also hold for the $k$-connectivity preserver family. \Cref{thm:reductions} can be visualized as below. \end{theorem}

\definecolor{34f7703b-c631-5d19-a2aa-8444d08d9619}{RGB}{255, 255, 255}
\definecolor{f3551e38-74df-57e2-b793-83d7fe876c85}{RGB}{0, 0, 0}
\definecolor{0b71a967-1f15-55a5-9bb9-70efa7b4fc58}{RGB}{51, 51, 51}
\definecolor{5856d031-3da1-575c-834e-c77e9e438c62}{RGB}{162, 177, 195}

\tikzstyle{064d39ef-8229-5b34-82aa-8cec5ded83a2} = [rectangle, minimum width=3cm, minimum height=1cm, text centered, font=\normalsize, color=0b71a967-1f15-55a5-9bb9-70efa7b4fc58, draw=f3551e38-74df-57e2-b793-83d7fe876c85, line width=1, fill=34f7703b-c631-5d19-a2aa-8444d08d9619]
\tikzstyle{7be24b85-97d0-5b76-ba9e-d94005dca8f2} = [thick, draw=5856d031-3da1-575c-834e-c77e9e438c62, line width=2, ->, >=stealth]

\begin{tikzpicture}[node distance=2cm]
\centering
\label{fig:reductions}
\node (990574b8-119d-487e-b924-02c7dd7078c0) [064d39ef-8229-5b34-82aa-8cec5ded83a2] {$s$-$t$};
\node (950d45e1-5d4e-450b-b18d-3719cd91c914) [064d39ef-8229-5b34-82aa-8cec5ded83a2, right of=990574b8-119d-487e-b924-02c7dd7078c0, xshift=2cm] {global};
\node (95f17a38-7cf6-443e-9183-ea639c27e9f0) [064d39ef-8229-5b34-82aa-8cec5ded83a2, right of=950d45e1-5d4e-450b-b18d-3719cd91c914, xshift=2cm] {single-source};
\node (99c5bb58-a488-440f-a093-36194836e6cd) [064d39ef-8229-5b34-82aa-8cec5ded83a2, right of=95f17a38-7cf6-443e-9183-ea639c27e9f0, xshift=2cm] {all-pairs};
\draw [7be24b85-97d0-5b76-ba9e-d94005dca8f2] (990574b8-119d-487e-b924-02c7dd7078c0) --  (950d45e1-5d4e-450b-b18d-3719cd91c914);
\draw [7be24b85-97d0-5b76-ba9e-d94005dca8f2] (950d45e1-5d4e-450b-b18d-3719cd91c914) --  (95f17a38-7cf6-443e-9183-ea639c27e9f0);
\draw [7be24b85-97d0-5b76-ba9e-d94005dca8f2] (95f17a38-7cf6-443e-9183-ea639c27e9f0) --  (99c5bb58-a488-440f-a093-36194836e6cd);
\end{tikzpicture}

By \Cref{thm:FT} and our reductions in \Cref{thm:reductions}, all preservers in the $k$-FT connectivity preserver family have size at most $O(k4^{k}n\log n)$. However, do all preservers in this family require size that depends exponentially on $k$? We confirm that this exponential dependence is necessary with the following theorem, which we prove in \Cref{app:lowerbound}.

\begin{theorem} \label{thm:stlower} For any positive integers $n,k$ with $n\geq2^{k/2}$, there exists an $n$-vertex directed graph $G$ and a pair of vertices $s,t\in V(G)$, such that any $s$-$t$ $k$-FT connectivity preserver of $G$ must have $\Omega(n2^{k/2})$ edges. \end{theorem}

%\section{Different Types of Symmetric Connectivity Preservers}

\subsection{Reductions between Preservers: Proof of \Cref{thm:reductions}}

\label{app:reduction}

% \begin{theorem}
%     There exist the following reductions among different types of $k$-FT tolerant connectivity preservers in the above variants we have stated:
%     \begin{itemize}
%         \item If there exists a global $k$-FT preserver with size $f(n,k)$, then there exists a $s$-$t$ $k$-FT preserver with size $O(f(n,k))$.
%         \item If there exists a single-source $k$-FT preserver with size $f(n,k)$, then there exists a global $k$-FT preserver with size $O(f(n,k))$.
%         \item If there exists a all-pairs $k$-FT preserver with size $f(n,k)$, then there exists a global $k$-FT preserver with size $O(f(n,k))$.
%     \end{itemize}
% \end{theorem}

\begin{proof}[Proof of \Cref{thm:reductions}] We prove the three reductions in \Cref{thm:reductions} as follows:

\begin{enumerate}
        \item $\textbf{\textit{s}}$\textbf{-}$\textbf{\textit{t}}$ $\Rightarrow$  \textbf{global}. Suppose there exists a global $k$-FT connectivity preserver with size $f(n,k)$ for every $n$-vertex graph. We claim that for any $n$-vertex graph $G=(V,E)$ and every $s,t \in V$, we can build an $s$-$t$ $k$-FT connectivity preserver with size $O(f(n,k))$% \benyu{Note: Here we actually need the reachability version of preservers. I will not define explicitly. }\benyu{Why it's this long...? Any better idea?}\gary{It looks good to me! I don't know how to make it shorter}. 
        
        We let  $G_1=G + \{(v,s)|v \in V\} + \{(t,v)|v \in V\}$ be the graph obtained by adding $(v,s)$ for any $v \in V$ and $(t,v)$ for any $v \in V$. Let $H_1$ be a global $k$-FT preserver of $G_1$. Suppose $t$ is \emph{reachable} from $s$ in $G-F$ for some fault set $F \subseteq G$ with $|F| \leq k$, then $t$ is \emph{reachable} from $s$ in $G_1-F$. For all $v \in V$, both $(v,s)$ and $(t,v)$ are in $G_1-F$. Given that $t$ is \emph{reachable} from $s$ in $G_1-F$, we know all vertices $v$ will be in the same strong component with $s,t$. 
        
        Thus, $G_1-F$ is strongly connected. Since $H_1$ is a global $k$-FT preserver of $G_1$, $H_1-F$ is also strongly connected, and $t$ is reachable from $s$ in $H_1-F$. Finally, since any simple directed path from $s$ to $t$ cannot pass any edge $(v,s)$ or $(t,v)$, we deduce that $t$ is \emph{reachable} from $s$ in $(H_1 \cap G) - F$. 

        Let $G_2=G \cup \{(v,t)|v \in V\} \cup \{(s,v)|v \in V\}$, and let $H_2$ be a global $k$-FT preserver of $G_2$. By the same argument we know if $s$ is \emph{reachable} from $t$ in $G-F$, then $s$ is \emph{reachable} from $t$ in $(H_2\cap G)-F$.
        
        Finally, we define $H=(H_1 \cup H_2) \cap G$, and observe that $|E(H)| \leq 2f(n,k) + 4n = O(f(n, k))$. Now suppose $s$ and $t$ are \emph{strongly connected} in $G-F$ for some fault set $F \subseteq G$ with $|F| \leq k$, then from above we get $t$ is \emph{reachable} from $s$ in $(H_1 \cap G)-F$ and $s$ is \emph{reachable} from $t$ in $(H_2\cap G)-F$. Thus, they are still strongly connected in $H-F$. We conclude that $H$ is an $s$-$t$ $k$-FT preserver with size $O(f(n,k))$.

        \item  \textbf{global} $\Rightarrow$ \textbf{single-source}. Suppose there exists a single-source $k$-FT  preserver with size $f(n,k)$ for every $n$-vertex graph. Now for any $n$-vertex graph $G$, we can pick any source $s\in V(G)$ and build a single-source $k$-FT preserver $H$ with source $s$.

        $H$ will have size $f(n,k)$, and we prove that $H$ is a global $k$-FT preserver of $G$. Suppose $G-F$ is strongly connected for some fault set $F$ with $|F| \leq k$, then in particular, every vertex $v$ will be strongly connected to the source $s$ in $G-F$. Therefore, since $H$ is a single-source $k$-FT preserver, we know every vertex $v$ will also be strongly connected to the source $s$ in $H-F$, and thus $H-F$ is also connected.
        
        \item  \textbf{single-source} $\Rightarrow$ \textbf{all-pairs}. Since any all-pairs $k$-FT connectivity preserver is also a single-source $k$-FT connectivity preserver by definition, we immediately get this result.
    \end{enumerate}
\end{proof}

\subsection{Lower Bound for $s$-$t$ $k$-FT Connectivity Preservers: Proof of \Cref{thm:stlower}}

% \begin{theorem}
%     \label{thm:stlower}
%     For any positive integers $n, k$ with $n \geq 2^{k/2}$, there exists a directed graph $G$ on $n$ vertices and a pair of vertices $s,t \in G$, where the $s$-$t$ $k$-FT Tolerant reachability preserver must have $\Omega(n2^{k/2})$ edges.
% \end{theorem}
\label{app:lowerbound}
\begin{figure}[ht]
    \centering
    \includegraphics[width=\textwidth]{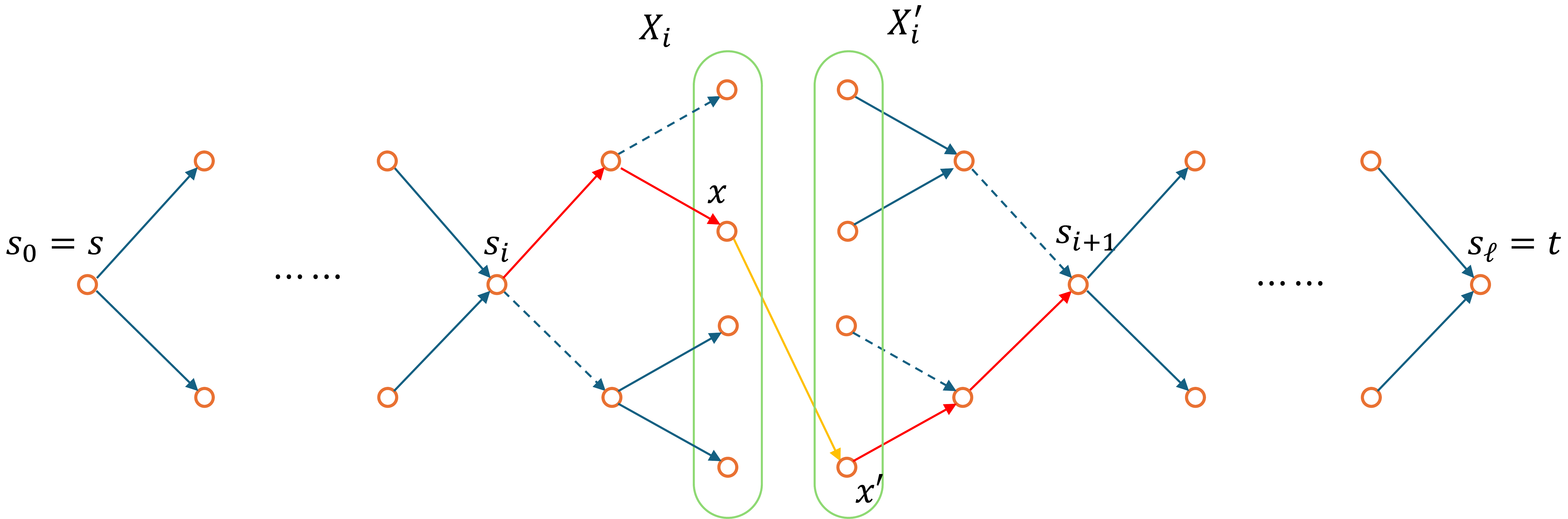}
    \caption{An $\Omega(n2^{k/2})$ Lower Bound for $s$-$t$ $k$-FT Preserver}
    \label{fig:stlower}
\end{figure}
\begin{proof}[Proof of \Cref{thm:stlower}]
    We consider the following construction, as in figure \Cref{fig:stlower}. First, we define a set of $\ell = \Theta(n/2^{k/2})$ vertices $s_0, s_1,\dots, s_\ell$. We will let $s = s_0$ and $t = s_{\ell}$. 
    Between each pair of vertices $s_i$ and $s_{i+1}$, $i \in [0, \ell - 1]$, we construct our lower bound graph as follows:

    \begin{itemize}
        \item Construct a full binary tree with depth $k/2$ rooted as $s_i$, with edges pointing away from root $s_i$. Denote the tree by $T_i$ and the set of leaves by $X_i$.
        \item Construct a full binary tree with depth $k/2$ rooted as $s_{i+1}$, with edges pointing towards the root $s_{i+1}$. Denote the the tree by $T'_i$ and set of leaves by $X'_i$.
        \item For every $x \in X_i$ and $x' \in X'_i$, add an edge from $x$ to $x'$.
    \end{itemize}
    Additionally, we add an edge from $t$ to $s$ (not shown in \Cref{fig:stlower}). 
    Consider the graph $G$ we constructed. We will prove that any $s$-$t$ $k$-FT connectivity preserver $H \subseteq G$ must have $\Omega(n2^{k/2})$ edges. 
    
    Consider any edge $(x,x')$ such that $x \in X_i$ and $x' \in X'_i$ for some $i\in[0, \ell-1]$. Let $P$ be an $(s_i, x)$-path in $T_i$. 
    Let $F_1 \subseteq E(T_i) - E(P)$ be the set of $k/2$ edges in tree $T_i$ that are incident to path $P$ but not contained in it. 
    We observe that $x$ is the only leaf in $X_i$ that is reachable from $s_i$ in $G-F_1$. Similarly, there is a set $F_2 \subseteq E(T'_i)$ of $k/2$ edges in $T'_i$ such that $x'$ is the only leaf in $X'_i$ that can reach $s_{i+1}$ in $G-F_2$. 
    
    Finally, we let $F=F_1 \cup F_2$. Now we can observe that any path from $s$ to $t$ in $G-F$ must use edge $(x,x')$. Since $s,t$ are still strongly connected in $G-F$, it must be true that $(x,x')\in H$. Finally, since there are $\ell=\Omega(n/2^{k/2})$ layers, and each layer has $2^{k/2} \times 2^{k/2}=2^k$ edges, we know that $H$ must contain at least $\ell2^k=\Omega(n2^{k/2})$ edges. Therefore, the statement holds.
\end{proof}

\subsection{Cut Characterizations}

In this section, we state the cut characterizations of (all-pairs) $k$-connectivity preservers and $k$-FT connectivity preservers. We prove these characterizations are equivalent to the original flow characterizations of these connectivity preservers in \Cref{def:k-preservers} and \Cref{def:k-fault preservers}. Our cut characterizations highlight the difference and relation between $k$-connectivity preservers and $k$-FT connectivity preservers. 

% At a high level, we show $k$-connectivity preservers maintain all small minimum cuts of a certain kind, whereas $k$-FT connectivity preservers maintain all small \textit{minimal} cuts of a certain kind. We now formalize precisely what we mean by minimum cuts and minimal cuts.

% \begin{definition}[Minimum symmetric cut]
%     We say that an $(s, t)$-cut $(S, V-S)$ is a minimum symmetric $(s, t)$-cut if there does not exist a cut $(S', V-S')$ such that $|\delta^+(S')| \le |\delta^+(S)|$ and $(S', V-S')$ is an $(s, t)$-cut or a $(t, s)$-cut. 
% \end{definition}

\begin{definition}[Minimal symmetric cut]
    We say that an $(s, t)$-cut $(S, V-S)$ is a minimal symmetric $(s, t)$-cut if there does not exist a cut $(S', V-S')$ such that $\delta^+(S') \subsetneq \delta^+(S)$ and $(S', V-S')$ is an $(s, t)$-cut or a $(t, s)$-cut. 
\end{definition}

% Note that every minimum symmetric cut is also a minimal symmetric cut. However, there exist minimal symmetric cuts that are not minimum symmetric cuts. 

% Note that a minimal (even minimum\benyu{Should I draw a figure here? }\gary{Yeah I feel like the example we discussed could be good to add here if we have time.}) $(s,t)$-cut is not necessarily a minimal symmetric $(s, t)$-cut.\gary{I like this intuition! One issue is we have not formally defined what is a minimal cut anywhere in the paper. For example, is it out-minimal or in-minimal? Actually, in the prelims I only define a minimum cut, but I wonder if I should really define it as an out-minimum cut?}

% \ghnoteinline{Maybe we only need to give a cut characterization for (all-pairs) preservers in directed and undirected graphs?}

%\bwnoteinline{I have given both k-FT and k-conn preservers in directed graphs. Please check again since I feel they are a little subtle and like magic.}
%\ghnoteinline{I read through the proofs and they look really good! I edited them a little bit. Please let me know if they still look good or there are any issues.}

We now prove our cut characterization of $k$-FT connectivity preservers.

\begin{theorem} \label{thm:ft-conn-pres}
Given a directed graph $G$, a subgraph $H$ is an  (all-pairs) $k$-FT connectivity preserver of  $G$ if and only if the following flow and cut characterizations hold:
\begin{itemize}
\item (flow): for every $s,t \in V(G)$ and $F\subseteq E(G)$ where $|F|\le k$, we have $\lambda_{H-F}^{1}(s,t)=\lambda_{G-F}^{1}(s,t)$.
\item (cut): for every $s,t \in V(G)$ and every minimal symmetric $(s,t)$-cut $(S, V-S)$ in $G$, we have
$\min(|\delta^{+}_H(S)|, k+1) = \min(|\delta^{+}_G(S)|, k+1)$.
% $\delta^{k+1,+}_H(S) = \delta^{k+1,+}_G(S)$\benyu{needs explanation}. 
\end{itemize}
    %\ghnoteinline{I feel like the notation $\delta^{k+1,+}_H(S) = \delta^{k+1,+}_G(S)$ can be confusing because $\delta_H^+(S)$ is a set, but we are not saying these sets are equal, I think. }
\end{theorem}

\begin{proof}
The flow characterization of $k$-FT connectivity preservers is identical to the definition of $k$-FT connectivity preservers in \Cref{def:k-fault preservers}.  What remains is to prove that $H$ is a $k$-FT connectivity preserver of $G$ if and only if the cut characterization holds.

\textbf{(Only if direction.)} Let $H$ be a $k$-FT connectivity preserver of $G$. Suppose towards contradiction that $H$ does not satisfy the cut characterization, and let $(S, V-S)$ be a  minimal symmetric $(s,t)$-cut such that $\min(|\delta^+_H(S)|, k+1) < \min(|\delta^+_G(S)|, k+1)$. Let $F$ be the edge fault set $F = \delta^+_H(S)$, and note that $|F| \leq k$. We  observe that $s, t$ are strongly connected in $G-F$ (by symmetric minimality and the fact that $\delta_H^+(S) \subsetneq \delta_G^+(S)$), but not in $H-F$, contradicting our assumption that $H$ is a $k$-FT connectivity preserver. We conclude that $H$ satisfies the cut characterization.

% Since $(S, V-S)$ is a minimal symmetric cut in $G$, we know that $s$ and $t$ are strongly connected in $G$. In fact, for every strict subset $F \subsetneq \delta^+_G(S)$ of the edges in $\delta^+_G(S)$, vertices $s$ and $t$ are strongly connected in $G-F$, or else $(S, V-S)$ is not a minimal symmetric $(s, t)$-cut in $G$. Now if $|\delta^+_H(S)| < \min(|\delta^+_G(S)|, k+1)$, then letting $F = \delta^+_H(S)$, we observe that $s$ and $t$ are strongly connected in $G-F$ but not strongly connected in $H-F$. This contradicts our assumption that $H$ is a $k$-FT connectivity preserver, so we conclude that $|\delta^+_H(S)| \ge \min(|\delta^+_G(S)|, k+1)$. 

\textbf{(If direction.)} Let $H$ be a subgraph of $G$ that satisfies the cut characterization, and suppose towards contradiction that $H$ is not a $k$-FT connectivity preserver. Then there exist a pair of vertices $s, t \in V(G)$ and a set of edge faults $F \subseteq E(G)$ with $|F|\leq k$ such that $s$ and $t$ are strongly connected in $G-F$, but not in $H-F$. Without loss of generality, we may assume that $s$ cannot reach $t$ in $H-F$, and we can let $(S, V-S)$ be an $(s, t)$-cut in $H-F$ of size $|\delta^+_{H-F}(S)| = 0$. 

Let $(S', V-S')$ be a cut in $G$ such that $\delta^+_G(S') \subseteq \delta^+_G(S)$, and $(S', V-S')$ is either a minimal symmetric $(s, t)$-cut or a minimal symmetric $(t, s)$-cut. 
We observe that  $\delta^{+}_H(S') \subseteq F$, whereas $\delta^{+}_G(S') \not \subseteq F$  because $s$ and $t$ are strongly connected in $G-F$. This implies that $\delta^+_H(S') \subsetneq \delta^+_G(S')$. Since $|\delta^+_H(S')| \leq |F| \leq k$, we conclude that  $$\min(|\delta^{+}_H(S')|, k+1) < \min(|\delta^{+}_G(S')|, k+1),$$ contradicting our assumption that $H$ satisfies the cut characterization. % \ghnoteinline{Your proof was a little bit fast for me, so I added more details.  It is a very magical argument.}
% On the other side, suppose $H$ does not satisfy the cut characterization, then consider a \emph{symmetric minimal} $(s,t)$ cut $\delta^+(S)$ that $\delta^{k+1,+}_H(S) < \delta^{k+1,+}_G(S)$. Let failures be $F = \delta^+_H(S)$, then $|F| \leq k$. We finally observe that $s,t$ are strongly connected in $G-F$ (by symmetric minimality), but not in $H-F$.
\end{proof}

Then, we prove our cut characterization for $k$-connectivity preservers.

\begin{theorem} \label{thm:k-conn-press}
Given a directed graph $G$, a subgraph $H$ is an (all-pairs) $k$-connectivity preserver of $G$ if and only if the following flow and cut characterizations hold:
\begin{itemize}
\item (flow): for every $s,t \in V(G)$, we have $\lambda_{H}^{k}(s,t)=\lambda_{G}^{k}(s,t)$.
\item (cut): for every $s,t \in V(G)$ and every minimal symmetric $(s,t)$-cut $(S, V-S)$ in $G$, we have
$|\delta^{+}_H(S)| \geq \lambda_G^k(s, t)$.
\end{itemize}
\end{theorem}
\begin{proof}

The flow characterization of $k$-connectivity preservers is identical to the definition of $k$-connectivity preservers in \Cref{def:k-preservers}.  What remains is to prove that $H$ is a $k$-connectivity preserver of $G$ if and only if the cut characterization holds.

\textbf{(Only if direction.)} Let $H$ be a $k$-connectivity preserver of $G$. Suppose towards contradiction that $H$ does not satisfy the cut characterization, and let $(S, V-S)$ be a  minimal symmetric $(s,t)$-cut such that $|\delta^{+}_H(S)| < \lambda_G^k(s, t)$, then we immediately have the contradiction that $\lambda_H^k(s, t) \leq |\delta^{+}_H(S)| < \lambda_G^k(s, t)$. We conclude that $H$ satisfies the cut characterization. 

\textbf{(If direction.)} Let $H$ be a subgraph of $G$ that satisfies the cut characterization, and suppose towards contradiction that $H$ is not a $k$-connectivity preserver. Then there exist a pair of vertices $s, t \in V(G)$ that $\lambda_{H}^{k}(s,t) < \lambda_{G}^{k}(s,t)$. Without loss of generality, we may assume there is a cut $(S, V-S)$ be an $(s, t)$-cut in $H$ of size $|\delta^+_{H}(S)| < \lambda_{G}^{k}(s,t)$. 

Let $(S', V-S')$ be a cut in $G$ such that $\delta^+_G(S') \subseteq \delta^+_G(S)$, and $(S', V-S')$ is either a minimal symmetric $(s, t)$-cut or a minimal symmetric $(t, s)$-cut. We observe that  $|\delta^{+}_H(S')| \leq |\delta^+_{H}(S)| < \lambda_{G}^{k}(s,t)$, contradicting our assumption that $H$ satisfies the cut characterization. % \ghnoteinline{Your proof was a little bit fast for me, so I added more details.  It is a very magical argument.}
% On the other side, suppose $H$ does not satisfy the cut characterization, then consider a \emph{symmetric minimal} $(s,t)$ cut $\delta^+(S)$ that $\delta^{k+1,+}_H(S) < \delta^{k+1,+}_G(S)$. Let failures be $F = \delta^+_H(S)$, then $|F| \leq k$. We finally observe that $s,t$ are strongly connected in $G-F$ (by symmetric minimality), but not in $H-F$.

\end{proof}

\begin{remark}
\Cref{thm:ft-conn-pres} and \Cref{thm:k-conn-press} highlight the similarities and differences between $k$-FT connectivity preservers and $k$-connectivity preservers. In particular, they both characterize $k$-connectivity preservers as preserving the minimum symmetric cuts: In $k$-FT connectivity preservers, we preserve the minimal symmetric cuts of size at most $k+1$, and in $k$-connectivity preservers, we preserve the minimal symmetric cuts of size at most $\lambda_G^k(s,t) \leq k$. These also imply any $k-1$-FT connectivity preserver is a $k$-connectivity preserver.
\end{remark}

\section{Strong Connectivity Preservers in Other Fault Models}

% \bwnoteinline{Finished. I will check again shortly tomorrow.}
% \ghnoteinline{Hi Benyu, I have read through this section. It looks great! }

As mentioned in \Cref{subsec:mainresult}, we will prove that both the \textbf{bounded-degree fault model} (introduced in \cite{boundeddegree0}) and the \textbf{color fault model} (introduced in \cite{color0}) do not admit sparse fault-tolerant connectivity preservers in directed graphs, even in the easiest settings of 1-faulty-degree and 1-color fault, respectively. In this appendix, we will formally define these fault models and prove these lower bounds.

% Remark 1.9 (Other fault models). In recent years, two new graph fault models have emerged
% that generalize the ‘standard’ fault model used in our paper. We refer to these two models as
% the bounded-degree fault model (introduced in [BHP24]) and the color fault model (introduced in
% [PST24a]). These new fault models have attracted attention in part because classical fault-tolerance
% results on connectivity preservers and spanners have been extended to these stronger models with
% essentially the same performance guarantees [BHP24, PT25, PST24a]. A priori, it is quite plausible
% that our directed fault-tolerant connectivity preservers could be generalized to these stronger fault
% models.
% Unfortunately, we prove that we cannot generalize Theorem 1.5 to these stronger fault models.
% Specifically, in the bounded-degree fault model we show that 1-faulty-degree (symmetric) connectivity
% preservers require Ω(n2) edges in directed graphs. Likewise, in the color fault model we show that
% 1-color fault-tolerant (symmetric) connectivity preservers require ˜Ω(n2) edges in directed graphs.
% These lower bounds strongly separate the ‘standard’ fault model from these new fault models in
% directed graphs. We formally define these fault models and prove these lower bounds in Appendix B.

\label{app:other fault models}

\subsection{Bounded-Degree Faults}

As in \cite{boundeddegree0}, we  define the bounded-degree fault model as follows:

\begin{definition}[Bounded-Degree Fault-Tolerant Connectivity Preservers]
    Given a directed graph $G$, a subgraph $H$ is called a $k$-faulty-degree-tolerant connectivity preserver if for any pair of vertices $(u, v)$ and any set of failed edges $F$ such that $|F \cap (\delta^+(x) \cup \delta^-(x))| \leq k$\footnote{As a generalization to degree in undirected graphs, one may define this as $|F \cap \delta^+(x)|\leq k$ or $|F \cap \delta^-(x)|\leq k$, but these variants are even stronger than what we define here so we cannot hope for anything better.} for each vertex $x \in V(G)$, we have the guarantee that $u,v$ are strongly connected in $G-F$ if and only if $u,v$ are strongly connected in $H-F$.
\end{definition}

We prove the following lower bound in the bounded-degree fault setting:

\begin{theorem} \label{thm:faultydegree}
    There exists a directed graph $G$ with $n$ vertices such that any $1$-faulty-degree-tolerant connectivity  preserver $H \subseteq G$ of $G$ requires $\Omega(n^2)$ edges.
\end{theorem}

\begin{figure}[ht]
    \centering
    \includegraphics[width=0.7\textwidth]{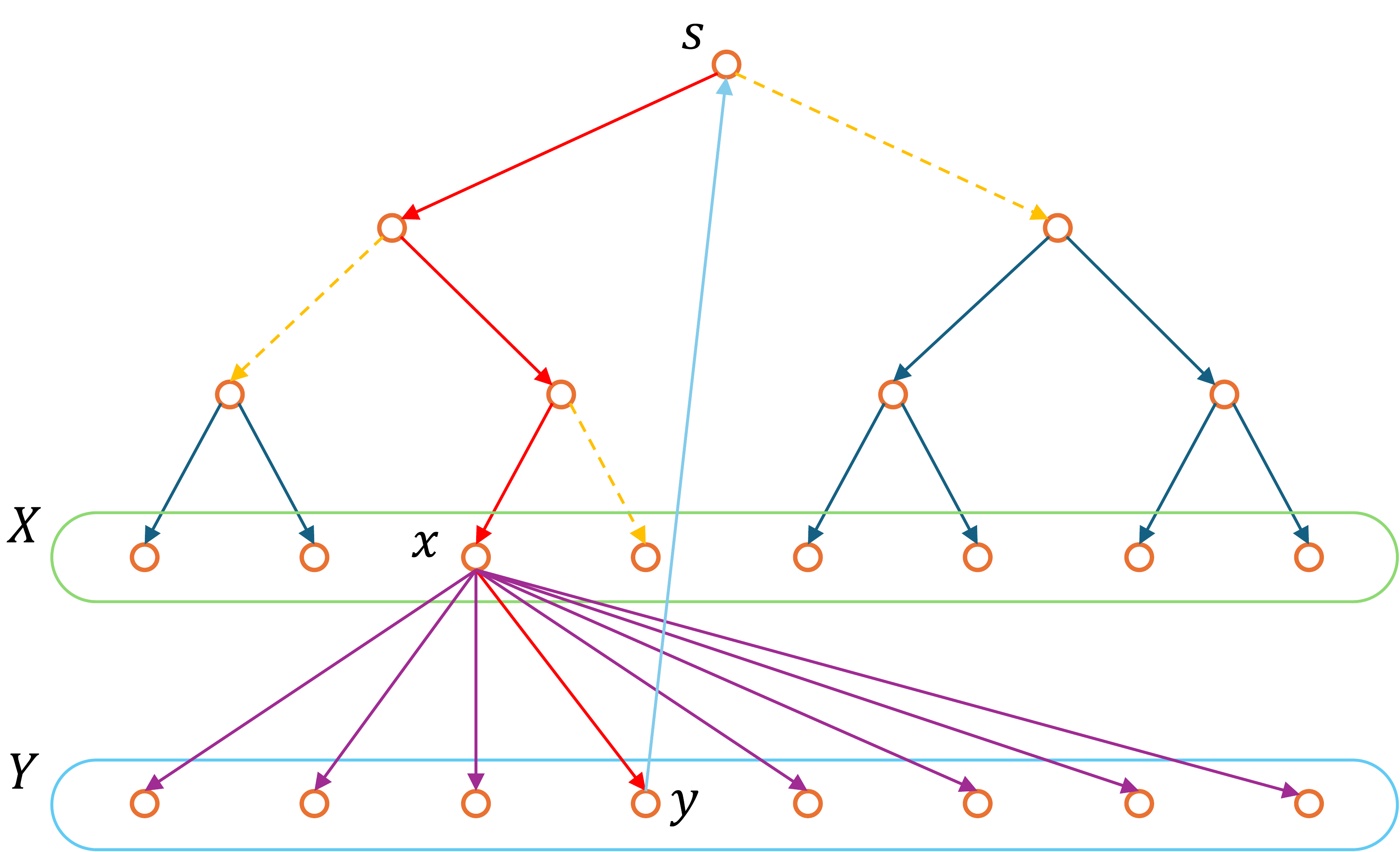}
    \caption{An $\Omega(n^2)$ lower bound for  $1$-faulty-degree preservers.}
    \label{fig:bounded-deg}
\end{figure}

\begin{proof}
    We show that a construction similar to the one in \cite{baswana2018fault} gives the desired lower bound. As in \Cref{fig:bounded-deg}, we construct our instance $G$:
    \begin{itemize}
        \item Construct a full binary tree $T$ with root $s$ and leaf set $X$ of size $|X|=\Theta(n)$. The directed edges in $T$ are oriented away from root $s$.
        \item Construct another vertex set $Y$ with $|Y|=\Theta(n)$. Add a directed edge  from every $x \in X$ to every $y \in Y$, and add a directed edge from every $y \in Y$ to $s$.
    \end{itemize}
    This final graph $G$ has $\Theta(n)$ vertices and $\Theta(n^2)$ edges.

    We prove that any $1$-faulty-degree preserver $H \subseteq G$ must contain all edges in $X \times Y$. Consider any $x \in X$ and $y \in Y$. 
    Let $P$ be the unique $(s, x)$-path in tree $T$. We let the fault set $F$  be all edges in $T - E(P)$ that are incident to a vertex in  $(s, x)$-path $P$. This corresponds to all the dashed edges in \Cref{fig:bounded-deg}. We observe that fault set $F$ is a $1$-bounded-degree fault. Additionally, we observe that $x$ is the only leaf in $X$ that is reachable from $s$ in $G-F$. Therefore, $s,y$ are strongly connected in $G-F$ but not strongly connected in $G-(x,y)-F$. Therefore, to preserve the connectivity between $(s,y)$, the edge $(x,y)$ must be present in the preserver $H$. Finally, since there are $\Theta(n^2)$ edges in $X \times Y$, we conclude that there must be $\Theta(n^2)$ edges in $H$.
\end{proof}

As a brief aside, we also quickly point out that the graph constructed in the proof of \Cref{thm:faultydegree} implies a large separation between $k$-FT connectivity preservers and $k$-connectivity preservers.
\begin{theorem} \label{thm:conn ft gap}
    There exists a directed graph $G$ with $n$ vertices that admits a $k$-connectivity preserver of size $O(n)$, but any $k$-FT connectivity preserver of $G$ requires $\Omega(n2^k)$ edges.
\end{theorem}
\begin{proof}
    The construction comes from \cite{baswana2018fault}. Consider the graph $G$ in the proof of \Cref{thm:faultydegree} but instead of letting $|X|=\Theta(n)$, we let $|X|=2^k$. Each pair of vertices in $G$ is only $1$-strongly-connected, so we can preserve the $k$-connectivity of $G$ for any $k \in [1, n]$ using a strongly connected spanning subgraph of size $O(n)$. On the other hand, by cutting the same set of edges as in the proof of \Cref{thm:faultydegree} (which is with size at most $k$ since the tree has $k$ layers), we can show that every $k$-FT connectivity preserver of $G$, where $k \in [1, \log_2 n]$, requires $\Omega(2^kn)$ edges.
\end{proof}

\subsection{Color Faults}

As in \cite{color0}, we will define the color fault model as follows:

\begin{definition}[Color-Fault-Tolerant Connectivity Preservers]
    Given an edge-colored directed graph $G$ with an associated partition $\mathcal{C}$ of $E(G)$  into $|\mathcal{C}|$ colors, a subgraph $H$ is called a $k$-color-fault-tolerant connectivity preserver if for any pair of vertices $(u, v)$ and any family of failed colors $\mathcal{F} \subseteq \mathcal{C}$ such that $|\mathcal{F}| \leq k$, we have the guarantee that $u,v$ are strongly connected in $G-\cup\mathcal{F}$ if and only if $u,v$ are strongly connected in $H-\cup\mathcal{F}$, where $\cup\mathcal{F}$ denotes the union of all edge sets contained in $\mathcal{F} \subseteq \mathcal{C}$.
\end{definition}

We prove the following lower bound in the color fault setting:

\begin{theorem}
    There exists a directed graph $G$ with $n$ vertices  such that any $1$-color-fault-tolerant connectivity  preserver $H \subseteq G$ of $G$ requires $\Tilde{\Omega}(n^2)$ edges.
\end{theorem}

\begin{figure}[ht]
    \centering
    \includegraphics[width=0.7\textwidth]{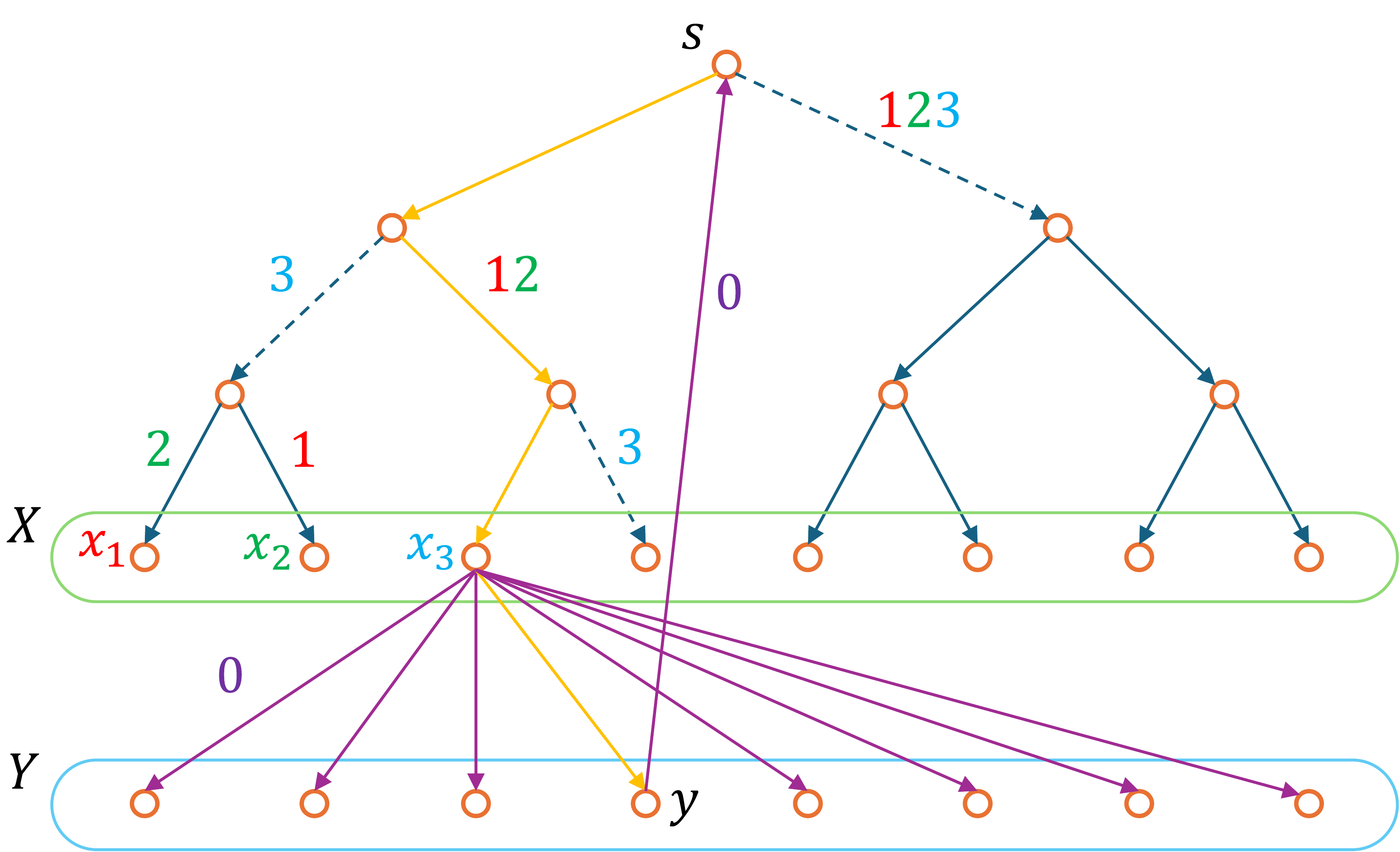}
    \caption{An $\Tilde{\Omega}(n^2)$ lower bound for $1$-color fault preservers. When the edges in color class $3$ fail, every $(s, y)$-path must pass through vertex $x_3$.}
    \label{fig:color-fault}
\end{figure}

\begin{proof}
    As in \Cref{fig:color-fault} (where a line with multiple colors denotes a path with an edge of each color), we construct our instance $G$:
    \begin{itemize}
        \item Construct a full binary tree $T$ with root $s$ and leaf set $X$ of size $|X|=\Theta(n/\log n)$. The directed edges in $T$ are oriented away from source $s$.
        \item For each vertex $x_i \in X$, where $i \in [1, |X|]$, we 
        let $P_i$ be the unique $(s, x_i)$-path in tree $T$. We color all edges in $T - E(P_i)$ that are incident to a vertex in path $P_i$ using color $i$. Note that under this coloring, one edge may have been assigned multiple colors. 
        \item For an edge $e$ in tree $T$ that is assigned a set $C_e \subseteq [1, |X|]$ of $|C_e| > 1$ distinct colors by the above coloring, 
        we replace edge $e$ with a path of length $|C_e|$, where each edge is colored with a different color in $C_e$. After applying this operation to all edges $e$ in tree $T$, the resulting graph will have a valid edge coloring. 

        Since each leaf $x_i \in X$ colors only $|P_i| = O(\log n)$ edges in $T$ with color $i$, and $|X| = \Theta(n / \log n)$, the previous operation increases the number of vertices in $G$ by at most $|X| \cdot O(\log n) = O(n)$. 
        \item Construct another vertex set $Y$ with $|Y|=\Theta(n)$. Add a directed edge from every $x \in X$ to every $y \in Y$, and add a directed edge from every $y \in Y$ to $s$. Color all these edges with a new color $0$.
    \end{itemize}
    This final graph $G$ has $\Theta(n)$ vertices and $\Theta(n^2/ \log n)$ edges. 

    Now we prove that any $1$-color-fault-tolerant connectivity preserver $H \subseteq G$ must contain all edges in $X \times Y$. Consider any $x_i \in X$ and $y \in Y$. We set the color fault set $F$ to be the set of all edges in $E(G)$ with color $i$. We observe that $x_i$ is the only leaf in $X$ that is reachable from $s$ in $G-F$. Therefore, $s,y$ are strongly connected in $G-F$ but not strongly connected in $G-(x_i,y)-F$. Therefore, to preserve the connectivity between $(s,y)$, the edge $(x_i,y)$ must be present in the preserver $H$. Finally, since there are $\Theta(n^2/ \log n)$ edges in $X \times Y$, we conclude that there must be $\Theta(n^2 / \log n)$ edges in $H$.
\end{proof}

\section{Omitted Proofs from Preliminaries}

\subsection{A Giant Component from Unbreakable Sets: Proof of \Cref{clm:giant_component}}
\label{sec:omit_giant}

% \bwnoteinline{TODO: Change $S$ to $C$ or change $C$ in 3.7 back.}

\begin{proof}[Proof of \Cref{clm:giant_component}]
Let $C_1, \dots, C_r \subseteq V(G)$ be the strongly connected components in $G-F$. We may assume without loss of generality that for each edge $e \in E(G-F)$, there exist strongly connected components $C_i, C_j$ of $G-F$  such that $e \in C_i \times C_j$, where $i \leq j$. This follows from the fact that the condensation graph of $G-F$ is acyclic. 

Let $C_{\leq i} = C_1 \cup \dots \cup C_i$ for each $i \in [1, r]$. Suppose towards contradiction that for each $i \in [1, r]$, SCC $C_i$ in $G-F$ satisfies $|C_i \cap U| < |U| - 2q$. Let $i \in [1, r]$ be the largest index such that $|C_{\leq i} \cap U| \leq q$.  Then $|C_{\leq i+1} \cap U| > q$. Moreover, since $|C_{i+1} \cap U| < |U| - 2q$, it follows that
$$
|(V - C_{\leq i+1}) \cap U| = |U| - |C_{\leq i} \cap U| - |C_{i+1} \cap U| > |U| - q -   (|U| - 2q)  = q.  
$$

Then the cut $(C_{\leq i+1}, V - C_{\leq i+1})$ will satisfy $|C_{\leq i+1} \cap U| > q$ and $|(V - C_{\leq i +1}) \cap U| > q$. Additionally, since $\delta^-_{G-F}(C_{\leq i+1}) = \emptyset$ by our earlier discussion, it follows that $|\delta^-_{G}(C_{\leq i+1})| \leq |F| \leq k$. However, this implies that set $U$ is not $(q, k)$-unbreakable in $G$, a contradiction. 

We conclude that there exists a strongly connected component $C$ in $G-F$ such that $|C \cap U| \geq |U| - 2q$, as claimed. 
\end{proof}

\subsection{Directed Expander Hierarchy: Proof of \Cref{lem:expander_hierarchy}}
\label{sec:expander_hierarchy}

% \begin{proof}
%     We claim that there is a nonempty intersection between set $U'$ and set $S$, i.e., $S \cap U' \neq \emptyset$. Suppose towards contradiction that $S \cap U'  = \emptyset$. Let $S_1 = \{s \in S \cap U  \mid s \not \leadsto_{G-F} U'\}$ and let $S_2 = \{s \in S \cap U   \mid U' \not \leadsto_{G-F} s\}$. Since $S$ is an SCC in $G-F$ and $S \cap U' = \emptyset$, we conclude that $S \cap U  = S_1 \cup S_2$. Assume wlog that $|S_1| \geq |S \cap U|/2$ (the case where $|S_2| \geq |S \cap U|/2$ is symmetric). Then $F$ is a edge cut of size at most $|F|\leq k$ such that set $S_1 \subseteq U$ and set $U' \subseteq U$ lie on different sides of the cut. Since $|S_1| \geq |S|/2 \geq q$ and $|U'| = q$, this contradicts our assumption that set $U$ is $(q, k)$ unbreakable. We conclude that $U' \cap S \neq \emptyset$. 
% \end{proof}

In this appendix, we give a proof of the hierarchy in \Cref{lem:expander_hierarchy} based on approaches in \cite{long2024connectivity}. Although in \Cref{lem:expander_hierarchy} we state the unbreakability in each layer, in this appendix we will construct the hierarchy using \emph{expansion}.

\begin{definition}[Expansion]
    % For any directed graph $G=(V,E)$, suppose that $\delta^+(S)$ is a cut\thatchaphol{which direction? don't you need to refer to both directions?} in $G$, we define the expansion to be \[\phi(S) = \frac{|\delta^{+}(S)|}{\min(|S|,|V-S|)}.\]
    % We define the expansion $\phi(G)$ of a graph $G$ to be $\min_{\emptyset \neq S\subsetneq V(G)} \phi(S)$.\thatchaphol{You need $\emptyset \neq S \subsetneq V$}

    For any directed graph $G=(V,E)$, and any terminal set $U \subseteq V$, we say $U$ is \emph{$\phi$-expanding} in $G$, if
    \[|\delta^{+}(S)| \geq \phi \cdot {\min(|S \cap U|,|U-S|)}\] %\thatchaphol{which direction? don't you need to refer to both directions?}% \thatchaphol{Just write $U-S$}
    for any cut $(S,V-S)$ in $G$ that separates $U$. We say a graph $G$ is $\phi$-expanding if the whole vertex set $V$ is $\phi$-expanding.
    
    %Then we let $\phi(U)=\min(\phi_U(S))$\thatchaphol{this notation is very bad. It is the same as $\phi(S)$} among all cuts $\delta^+(S)$ separating $U$. We also say that $U$ is \emph{$\phi$-expanding} in $G$, if $\phi(U) \geq \phi$.
\end{definition}

By replacing $S$ using $V-S$, we know $|\delta^{-}(S)| \geq \phi \cdot {\min(|S \cap U|,|U-S|)}$ for any cut $(S,V-S)$ in $G$ that separates a $\phi$-expanding set $U$ as well.

\begin{remark}\label{remark:expansion}
    If a terminal set $U$ is $\phi$-expanding, then $U$ is $(k/\phi,k)$-unbreakable for any $k$. %\thatchaphol{This is wrong because your definition of expansion is wrong. } 
\end{remark}
\begin{proof}
    Since $U$ is $\phi$-expanding, we know $|\delta^{+}(S)|/\phi \geq {\min(|S \cap U|,|U-S|)}$. Therefore, for any cut $(S,V-S)$ with $|\delta^+(S)| \leq k$ that separates $U$, we know either $k/\phi \geq |S \cap U|$ or $k/\phi \geq |U-S|$. The same result holds when $|\delta^-(S)|\leq k$. Thus, by \Cref{def:unbreakable}, the above argument holds.
\end{proof}

Generally, to determine expansion is NP-hard and enumerating cuts requires exponential time. However, \cite{agarwal2005log} shows that we can $O(\sqrt{\log n})$ approximate the expansion of any vertex set $U$ in polynomial time, which we conclude in the following theorem:

\begin{theorem}[\cite{agarwal2005log}] 

Given any graph $G$ and a terminal set $U$, there is a polynomial time algorithm that outputs a %balanced\thatchaphol{what do you mean here by balanced? Just remove it?} 
cut $(S,V-S)$ that approximates $\phi(U)$ within a factor of $O(\sqrt{\log n})$. i.e. $\phi_U(S)=O(\phi(U) \cdot \sqrt{\log n})$.
    
\end{theorem}

Therefore, to build the hierarchy as in \Cref{lem:expander_hierarchy}, we can consider an \emph{expander hierarchy} with regard to expansion $\phi$ as follows:

\begin{definition}[Directed Expander Hierarchy]
Let $G$ be an $n$-node directed graph, and let $\phi$ be any positive expansion. We define a $\phi$-expander hierarchy to be a partition $\{V_1, V_2, \dots, V_{\ell}\}$ of $V(G)$ into $\ell$ layers, with the following properties:
\begin{enumerate}
    \item Fix an $i \in [1, \ell]$. Let $V_{\leq i} = V_1 \cup \dots \cup V_i$, and let $C \subseteq V(G)$ be a SCC in $G[V_{\leq i}]$. Then the terminal set $V_i \cap C$ is $\phi$-expanding in $G[C]$.
    \item The number of layers $\ell$ in the partition is at most $\ell = O(\log n)$.
\end{enumerate}

% Moreover, if we require that $q \geq k \cdot c \sqrt{\log n}$ for a sufficiently large constant $c>0$, then this partition can be computed in time $O()$. \gary{TODO: include the time complexity here!}
    \label{def:expander_hierarchy}
    
\end{definition}

From \Cref{remark:expansion}, in the above expander hierarchy, the set $V_i \cap C$ is $(k/\phi,k)$-unbreakable in $G[C]$. Thus, it will satisfy \Cref{lem:expander_hierarchy}.Below we show that such a hierarchy in \Cref{def:expander_hierarchy} can be built by the following theorem: 

\begin{theorem}
    There exists an algorithm that, given a graph G, computes a directed expander hierarchy with $\phi = 1/2$ in exponential time\footnote{Same as \cite{long2024connectivity}, the factor $1/2$ here is artificial and can be $1$ via the same proof.}, or $\phi = O(1/\sqrt{\log n})$ in polynomial time.
\end{theorem}

\begin{proof}    

WLOG we suppose $G$ is strongly connected, otherwise we consider every strong component of $G$. As in \cite{long2024connectivity}, we construct the whole hierarchy by a top-down approach, from $V_\ell$ to $V_1$.  To construct the top level $U=V_\ell$, we hope that $U$ satisfies the following two criteria:

\begin{enumerate}
    \item $U$ is $\phi$-expanding in $G$.
    \item Any strong component in $G[V-U]$ has at most $|V|/2$ vertices.
\end{enumerate}

% \bwnoteinline{Edit Barrier :)}

Initially, we let $U=V$ and finally we want to make $U$ $\phi$-expanding in $G$ while satisfying the second criterion. Consider the following procedure. While $U$ is not $\phi$-expanding in $G$, there must exist a sparse cut $\delta^+(S)$ with $\phi_{U}(S) \leq 1/2$. Moreover, such cut $\delta^+(S)$ could be found in exponential time when $\phi(U) \leq 1/2$, or in polynomial time when $\phi(U) \leq O(1/\sqrt{\log n})$. Suppose $L=\{s \in S | \exists t \in V-S \text{ s.t. } (s,t)\in \delta^+(S)\}$, then:

\begin{itemize}
    \item If $|S| \leq |V|/2$, then we update $U$ as $U' = (U-S) \cup L$.
    \item Else, $|V-S| \leq |V|/2$, and we update as $U' = (U-(V-S)) \cup L$.
\end{itemize}

First, we show the procedure will terminate (and thus we have a final set $U$ that is $\phi$-expanding in $G$). In every update, we add at most $|L| \leq |\delta^{+}(S)|$ vertices to $U$, but we also delete at least $|U \cap S| \geq |\delta^{+}(S)|/(1/2)$ or $|U \cap (V-S)| \geq |\delta^{+}(S)|/(1/2)$ vertices from $U$. Thus, we know $|U|$ is strictly decreasing in every update and the procedure must terminate.

Then, we show that the second criterion will be invariant while we update $U$ to $U'$. WLOG, we consider below the case when $|S| \leq |V|/2$, and the other case can be proved symmetrically. In one update, we let $U' = (U-S) \cup L$. Let's consider any strong component $C$ in $G[V-U']$. We consider the following two cases:

\begin{enumerate}
    \item If $C \cap S \neq \emptyset$, since $C$ contains no vertex in $L \subseteq U'$, we know $C$ must be fully contained in $S$, and with size $|C| \leq |S| \leq |V|/2$. 
    \item Otherwise, $C \cap S = \emptyset$. Therefore, $C$ is a strong component in $G[V-U'-S]$. From the update, we know $(V-S) - U' = (V-S) - U$, thus, $G[V-U'-S] = G[V-U-S]$. We conclude that $C$ must be also strongly connected in $G[V-U]$, which implies $|C| \leq |V|/2$.
\end{enumerate}

Therefore, when the procedure terminates, all strong component $C$ in $G[V-U]$ has at most $|V|/2$ vertices. We let $V_\ell=U$, so $V_{\leq \ell-1}=V-U$ then recurse in every strong component $C$ in $G[V_{\leq \ell-1}]$ to obtain $V_{\ell-1}$ through $V_1$. From the above construction, we know the size of any component $C$ in $G[V_{\leq i}]$ will be halved compared to $G[V_{\leq (i+1)}]$ for any $i$. Since $G$ contains one component with $n$ vertices, there can be at most $\log n$ layers.
\end{proof} 

Finally, we have constructed a hierarchy that satisfies \Cref{lem:expander_hierarchy}.

\end{document}